 \documentclass{imsart} 
\RequirePackage[OT1]{fontenc}

\usepackage{mypackage}
\usepackage{mymacros_3}
\usepackage{mymacros2}
\usepackage{mytheoremstyle}
\usepackage{mymacros_vaccine}
\usepackage{mymacros_brackets, convex_opti}
\newtheorem{fact}{Fact}[theorem]

\def\rint{\int_{\mathbb{R}}}

\begin{document}
\numberwithin{equation}{section}

\begin{frontmatter}
\title{Improved inference for vaccine-induced immune responses via shape-constrained methods}
\runtitle{Shape-constrained inference in vaccine trials}

\begin{aug}
\author{\fnms{Nilanjana} \snm{Laha}},
\author{\fnms{Zoe} \snm{Moodie,}
} \\
\author{\fnms{Ying} \snm{Huang,}
}
and
\author{\fnms{Alex} \snm{Luedtke}
}

\runauthor{Laha et al.}


\end{aug}

\begin{abstract}
We study the performance of shape-constrained methods for evaluating immune response profiles from early-phase vaccine trials. The motivating problem for this work involves quantifying and comparing the IgG binding immune responses to the first and second variable loops (V1V2 region) arising in  HVTN 097 and HVTN 100 HIV vaccine trials. 
We consider unimodal and log-concave shape-constrained methods to compare the immune profiles of the two vaccines, which is reasonable because the data support that the underlying densities of the immune responses could have these shapes. 
To this end, we develop novel shape-constrained tests of stochastic dominance and shape-constrained plug-in estimators of  the squared Hellinger distance between two densities. Our techniques are either tuning parameter free, or rely on only one tuning parameter, but their performance  is either better (the  tests of stochastic dominance)  or comparable with the  nonparametric methods (the estimators of the squared Hellinger distance). The minimal dependence on tuning parameters is especially desirable in clinical contexts where analyses must be prespecified and reproducible.   \textcolor{black}{Our methods are supported by theoretical results and simulation studies.}
\end{abstract}



\end{frontmatter}
\section{Introduction}
\label{sec: intro}

To date, the RV144 trial conducted in Thailand is the only vaccine efficacy  trial to show a signal of efficacy (31\%) against HIV infection  \citep{RV144}. RV144 inspired a phase 1b trial, named HIV Vaccine Trials Network (HVTN) 097, which evaluated the safety and immunogenicity  of the same regimen  in a South African population \citep{HVTN097}. The predominant subtype of HIV in South Africa 
is Clade C.
Therefore, in an effort to  increase the potential for high efficacy against clade C infections, scientists modified the HVTN 097 regimen to  include HIV strains matched to the South African clade C infections \citep{HVTN100_primary}.
A phase 1/2 trial named HVTN 100 assessed the  safety and immunogenicity of the modified regimen in South Africa. However, the HVTN 702 phase 2B/3 trial of the HVTN 100 regimen in South Africa met its non-efficacy criteria in February 2020 at a planned interim analysis.
 
 

 In light of the above,  the comparison between  the immune response profiles of HVTN 097 and HVTN 100 trials becomes important because the latter can shed some light on why the HVTN 100 regimen lacks efficacy.
 To fix ideas, here we focus on one class of immune responses, namely the binding of IgG antibodies to the first and second variable loops (V1V2 region) of the HIV envelope. This immune response is of particular interest because the RV144 trial revealed an inverse association between HIV infection and this immune response among vaccinees \citep{RV144_immuno}. Using data from HVTN 097 and HVTN 100, we focus on answering the following three questions:

 \begin{compactitem}
 \item[Q1.] How can we estimate the densities of the aggregated IgG binding immune responses (to HIV-1 envelope proteins)? 
 
 
 \item[Q2.] Is there any ordering between the distributions of the immune responses from the two trials?
 
 
 \item[Q3.] How can we measure the discrepancy between the densities of the immune responses from the two trials?
\end{compactitem}
 
 Although Q1 is not directly related with comparison of the two vaccine trials,
answering Q1 is important because the resulting density estimators can  help  in designing subsequent vaccines.  To answer Q2,  we resort to testing for stochastic dominance, which addresses the ordering of the underlying distribution functions.
  To answer Q3, we rely on the squared Hellinger distance as a measure of discrepancy between two densities, where, for densities $\fx$ and $\fy$, the  Hellinger distance $\hd(\fx,\fy)$ is defined by
\[\hd(\fx,\fy)=\sqrt{\frac{1}{2}\rint\lb\sqrt{f(x)}-\sqrt{g(x)}\rb^2 dx}.\]
 \textcolor{black}{The significance of  Q3 may not be immediately obvious. However, the squared Hellinger distance between the immune responses of HVTN 097 and HVTN 100 trials can serve as a benchmark when new pairs of vaccines are compared in future vaccine trials. The reason behind choosing the squared Hellinger distance  as the measure of discrepancy in particular is discussed in Section~\ref{sec: measure of discrep}.   }
 



{\color{black}
  There are numerous  nonparametric methods that can be implemented to carry out the aforementioned  steps. However, traditional nonparametric methods do not exploit any information on the shape of the underlying densities.
  The uniformity of the trial population and exploratory analyses bear evidence that the underlying densities can be unimodal.  \textcolor{black}{The data are also consistent with the possibility  that the densities are log-concave. The latter is an important subclass of unimodal densities, often advocated for use in modelling because it contains most of the well-known subexponential unimodal densities and allows powerful density estimation tools \citep{walther2009}.}  Of late, shape-constrained density estimation has gained much attention. Among other reasons, the reduced burden of external tuning parameters  \citep{samworth2018, johnson2018} makes shape-constrained density estimation an attractive alternative to \textcolor{black}{traditional nonparametric approaches like  kernel-based or basis expansion type methods, which  are known to be sensitive to the choice of the tuning parameters \citep[cf. pp. 327,][]{efromovich2008, laha2021}.} Also, shape-constrained density estimation methods require weaker smoothness assumptions for asymptotic consistency than do the nonparametric methods. Shape-constrained techniques are widely used in economics and operations research \citep{johnson2018}, and have seen application in other domains such as  circuit design \citep{hannah2012}. \textcolor{black}{For a detailed account on the recent development of the shape constraint literature, we refer the survey articles \cite{samworth2018} and  \cite{review}.}
  }
  
 

 Although leveraging shape information   can potentially increase efficiency, there is little to no literature on the application of shape-constrained tools in  vaccine trials. To answer our motivating questions, therefore,  we develop  new   methods using shape-constrained tools. 
 The application of the new methods   is not limited to vaccine trials. 
 \textcolor{black}{
 For example, our tests can be applied to other areas of medical research where tests of stochastic dominance are relevant. See \cite{leshno2004} for an in-depth discussion of potential uses of tests of stochastic dominance in medical research. In particular, our methods are applicable to the mortality data considered in \cite{leshno2004} to infer whether a surgery increases   the mortality of patients with abdominal aortic aneurysma. See also \cite{stinnett1998} and \cite{defauw2011} for the use of stochastic 
dominance in cost effective analysis of healthcare.
 Outside medical research, our tests have substantial applicability  in finance, economics, social welfare, and  operations research, where stochastic dominance is a popular tool to  compare  portfolios, income, utility, poverty, opportunity etc.; see \cite{levy1992},  \cite{ sriboonchita2009}, \cite{le1991}, among others, for a detailed account.   On the other hand, Hellinger distance has also seen successful application as a measure of discrepancy in various disciplines ranging from machine learning \citep{cieslak2009, gonzalez, gonzalez2010} to ecology  \citep{rao1995} to fraud detection \citep{yamanishi2004}. However, for income data, log-concavity based methods should be used with caution since  distributions with log-concave density are always sub-exponential \citep{theory}, whereas income data can have heavier tails.  Diagnostic tools such as those in \cite{asmussen2017} can be used to enquire if the data has heavier tail, e.g. regular varying tails, in these cases. 
 }



 \subsection{Organization of article and main contributions}
 \label{sec: contribution}
 
  Although the methods developed in this paper are general, central to our application lies the HVTN data, which we describe in Section~\ref{sec: data description}.
Below we briefly discuss  our methods and  main contributions. We write $\fn$ and $\fh$ for the densities of the  immune responses in the HVTN 097 and HVTN 100 trials respectively, and $\Fn$ and $\Fh$ for the corresponding distribution functions.
 \subsection*{Estimating $\fn$ and $\fh$:}  {In vaccine trials,  traditionally a kernel density estimator (KDE) is used for the purpose of density estimation  \citep[cf.][]{miladinovic2014}. However, using  cross-validation, in Section \ref{sec: density estimation}, we show that the log-concave maximum likelihood estimator (MLE) based estimators of \cite{2009rufi} and \cite{smoothed} minimize the estimated mean integrated squared error (MISE) among a class of shape-constrained density estimators and  KDEs. }
\subsection*{Shape-constrained tests of stochastic dominance:}
{The claim of a stochastic ordering between two samples is made stronger when it is backed by a test of stochastic dominance.
We say a distribution function $\Fx$  stochastically dominates  another distribution function $\Fy$ in first order ($\Fx\succeq\Fy$)  if  $\Fx(x)\leq \Fy(x)$ for all $x\in\RR$. \textcolor{black}{If $X$ and $Y$ are two random variables with distribution functions $\Fx$ and $\Fy$, respectively, then in this case, we say $X$ stochastically dominates $Y$, and write $X\succeq Y$.}  The dominance is regarded as  ``strict" $(\Fx\succ\Fy$ or $X\succ Y$) if, in addition, there exists $x\in\RR$ such that  $\Fx(x)<\Fy(x)$. 
If  $\Fx$ does not strictly stochastically dominate $\Fy$, then this event is defined as the non-dominance  ($\Fx\nsucceq\Fy$) of $\Fx$ over $\Fy$ \citep[][p. 25]{whang2019}.}

 The rejection of the null of non-dominance against the alternative of stochastic  dominance makes the strongest case for ranking one distribution over the other \citep{davidson2013, alvarez2016, LW2012}.  However, the resulting  test suffers from lack of power because the  overlap of distribution functions at the tails make unrestricted stochastic dominance  almost impossible to establish via hypothesis testing \citep[cf.][]{davidson2013,  whang2019}.  There are many ways to deal with this difficulty. For example, \cite{LW2012}  exploits the fact that the dominance between two distribution functions employs a dominance between their Fourier coefficients under a carefully chosen basis. However, their method 
 rejects the null of non-dominance in some scenarios when the distribution functions cross each other.
Another remedy discussed in the literature \citep{kaur, davidson2013,  alvarez2016} involves excluding the tail region  because it does not contain enough reliable information for the problem at hand. This type of tests focus only on a compact set $D$ inside the interior of the combined support $\{0<F+G<2\}$. In Section~\ref{sec: tests}, we follow the latter strategy because it allows for many ways of incorporating shape-constraints with desirable asymptotic properties.



\textcolor{black}{
 To the best of our knowledge, we are the first to introduce the use of shape constraints in the context of testing the null of non-dominance against stochastic dominance. 
 Moreover,  one of our nonparametric test statistics in Section~\ref{sec: null of non dominance} has not, to our knowledge, previously been studied in the context of testing the null of non-dominance. \textcolor{black}{In Section \ref{section: asymptotic}, we show that the shape-constrained and nonparametric versions of our tests control the asymptotic type I error at any null configuration under reasonable conditions.}  We also show that  our tests are asymptotically unbiased, and  consistent against all alternatives lying in the interior of the class of alternative distributions. \textcolor{black}{In Section \ref{sec:simulation:1}, we empirically show that the shape-constrained tests have better power than their nonparametric counterparts although they have the same asymptotic critical values.}   Section \ref{sec: application to our data} analyses the application of these tests to our data. The proofs are deferred to Appendix~\ref{sec: appendix: tests}.}
 \subsection*{Shape-constrained plug-in estimators of the squared Hellinger distance:}
In Section~\ref{sec: measure of discrep}, we construct plug-in estimators of the squared Hellinger distance. It is well known that,
unless bias-corrected, plug-in estimators based on the KDE  generally have a first-order bias \citep[cf. Section 2 of][]{robins2009}. In contrast, for some smooth functionals, shape constrained MLE based plug-in estimators  do not require further bias correction \citep[cf.][]{Jankowski2010, ramu2018}.  The results in \cite{lopuhaa2019} indicate that the squared Hellinger distance is an example of such a smooth functional.
When the underlying density is unimodal, we show that our unimodal density based plug-in estimator   enjoys the same asymptotic guarantees as that of the bias corrected  KDE based estimators \citep{robin2015} under some regularity conditions. 
The simulation studies  in Section~\ref{sec: simulation : measure of discrep} suggest that similar results hold for our log-concave MLE based plug in estimators as well. \textcolor{black}{In fact, our smooth log-concave MLE based estimator shows stable performance across all settings, 
where even the bias-corrected version of the KDE based estimator struggles in some cases.}

In the process, we develop theoretical tools for analyzing the asymptotic behavior of plug-in estimators based on the unimodal density estimator of \citeauthor{birge1997}, which may  of independent interest. We defer the latter analysis to Appendix~\ref{sec: appendix: measure of discrep}. 
 \textcolor{black}{  The methods developed in this paper are implementable using the R package  \texttt{SDNNtests} \citep{SDNNtests}, which is available on GitHub.} 


\subsection{Notations and terminologies}
\label{sec: notation: S2}
Before proceeding further, we introduce some notation that will be used throughout this paper. We consider two independent samples $X_1,\ldots,X_m$ and $Y_1,\ldots, Y_n$ drawn from distributions with densities $\fx$ and $\fy$. We denote the corresponding  distribution functions by $\Fx$ and $\Fy$, respectively. 
The respective empirical distribution functions will be denoted by $\Fmx$ and $\Fmy$. The pooled sample $(X_1,\ldots, X_m,Y_1,\ldots,Y_n)$ has  size $N=m+n$. We denote the corresponding empirical distribution function by  $\mathbb{H}_{N}$.

For $k\geq 1$, we let $||\cdot||_{k}$  denote the usual $L_{k}$ norm, i.e.
 $||\mu||_k=\nolinebreak(\rint |\mu(x)|^k dx)^{1/k}$,
where $\mu$ is a function supported on the real line. Also, we denote
 $||\mu||_{\infty}=\sup_{x\in\RR}\mu(x)$.
For a density $f$,  denote by $\supp(f)$ the set $\{x\ :\ f(x)>0\}$. For a concave function $f:\RR\mapsto\RR$, the domain  $\dom(f)$ will be defined as in \citep[][p. 40]{rockafellar}, that is, $\dom(f)=\{x\in\RR\ :\ f(x)>-\infty\}$.
 For a sequence of measures $\{P_n\}_{n\geq 1}$, we say  $P_n$ converges weakly to  $P$, and write $P_n\to_d P$,  if   $\lim\limits_{n\to\infty}\int \mu d P_n=\int \mu dP$ holds for any bounded continuous function $\mu:\RR\mapsto\RR$.
  For any two sets $A,B\subset\RR$, we denote by $\text{dist}(A,B)$ the quantity
 $\min\{|x-y|\ :\ x\in A, y\in B\}$. 
 We let $\iint(A)$ denote the interior of the set $A$.

\section{Background: HVTN 097 and HVTN 100}
\label{sec: data description}
This section presents an exploratory analysis of the dataset.
  For both trials, we consider the magnitude of IgG binding to the V1V2 region of seven clade C glycoprotein $70$  antigens. The immune responses were measured by an HIV-1 binding antibody multiplex assay (BAMA). \textcolor{black}{Following \cite{RV144_immuno}, \cite{HVTN097}, and \cite{HVTN100_primary}, we use the log-transformed net median fluorescence intensity (MFI) as the measure of immune response for statistical analysis.}

In HVTN 097 and HVTN 100, four injections of HIV vaccines were given at months 0, 1, 3, and 6. In this study, we only consider the responses measured two weeks after the month six vaccination, which is considered to be the peak immune response time point. \textcolor{black}{ We include only those vaccinees in this study who (a) completed  the first four scheduled vaccinations and provided samples at two weeks  after the month six vaccination (known as vaccinated per-protocol participants), and (b) developed a positive immune response  for at least one of the seven clade C V1V2 antigens.  There are  $68$ and $180$  vaccinees in the  HVTN 097 and HVTN 100 trial, respectively, who satisfy the above criteria.   We base our analysis  on the aggregated response, averaged over the seven clade C antigens mentioned above, and refer only to the latter when we say ``immune response".} We let $\Fn$ and $\Fh$ denote the distribution functions corresponding to the aggregated response from the two trials. 

 \begin{figure}[h]
\centering
\begin{subfigure}{.4\textwidth}
 \includegraphics[width=\textwidth, height= 1.5 in]{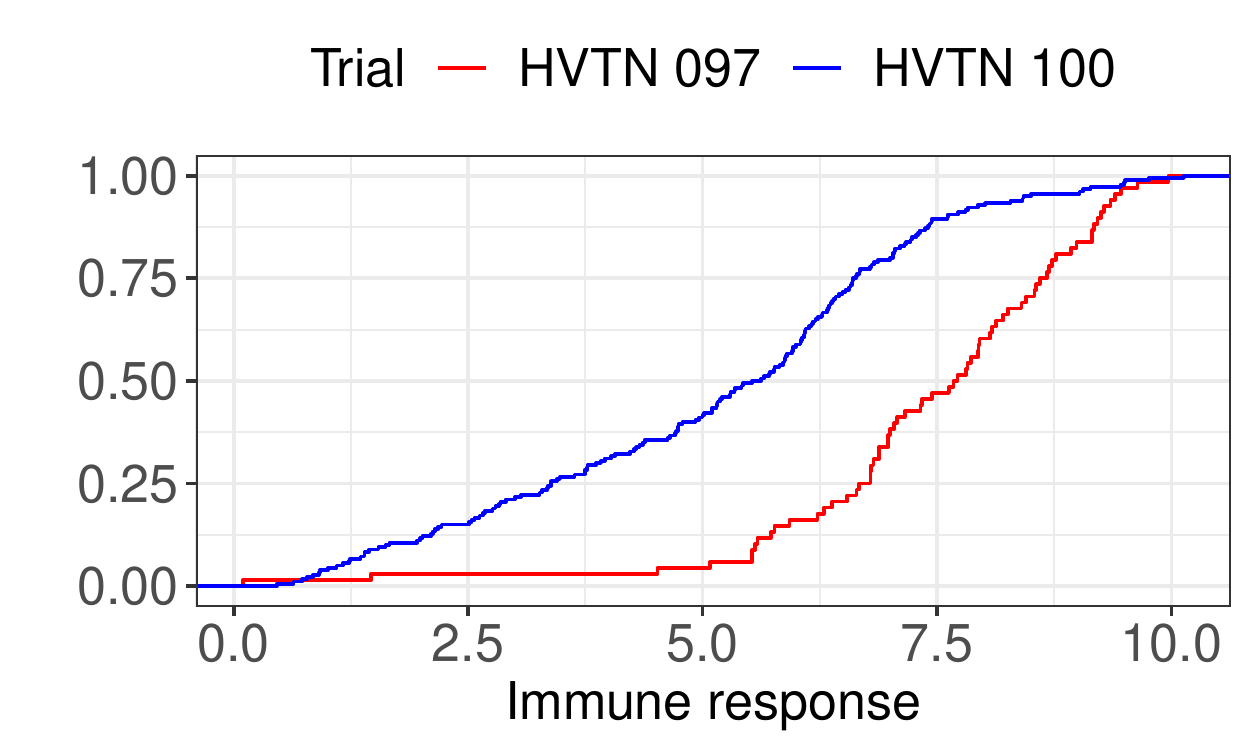}
 \caption{ECDF}
 \label{Fig: ecdf}
\end{subfigure}~
\begin{subfigure}{.4\textwidth}
 \includegraphics[width=\textwidth, height=1.5  in]{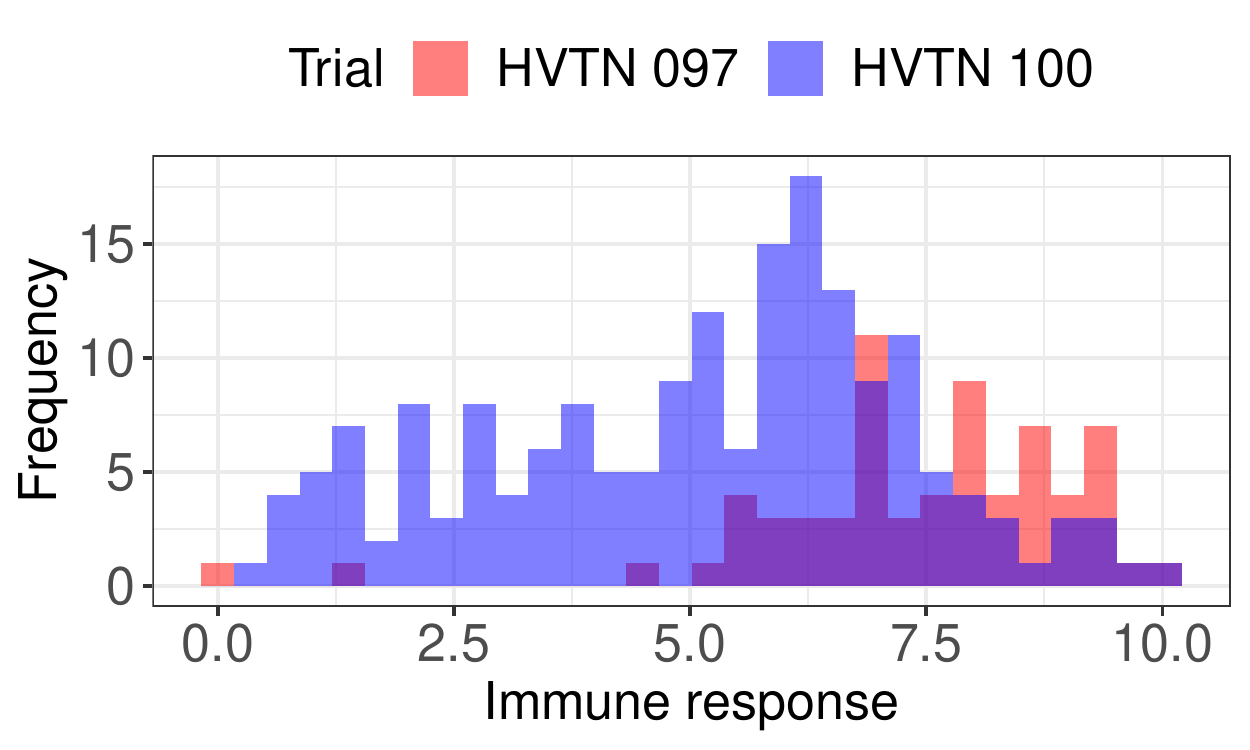}
 \caption{Histogram}
 \label{Fig: Histogram}
\end{subfigure}\\
 \begin{subfigure}{.4\textwidth}
 \includegraphics[width=\textwidth, height= 1.5 in]{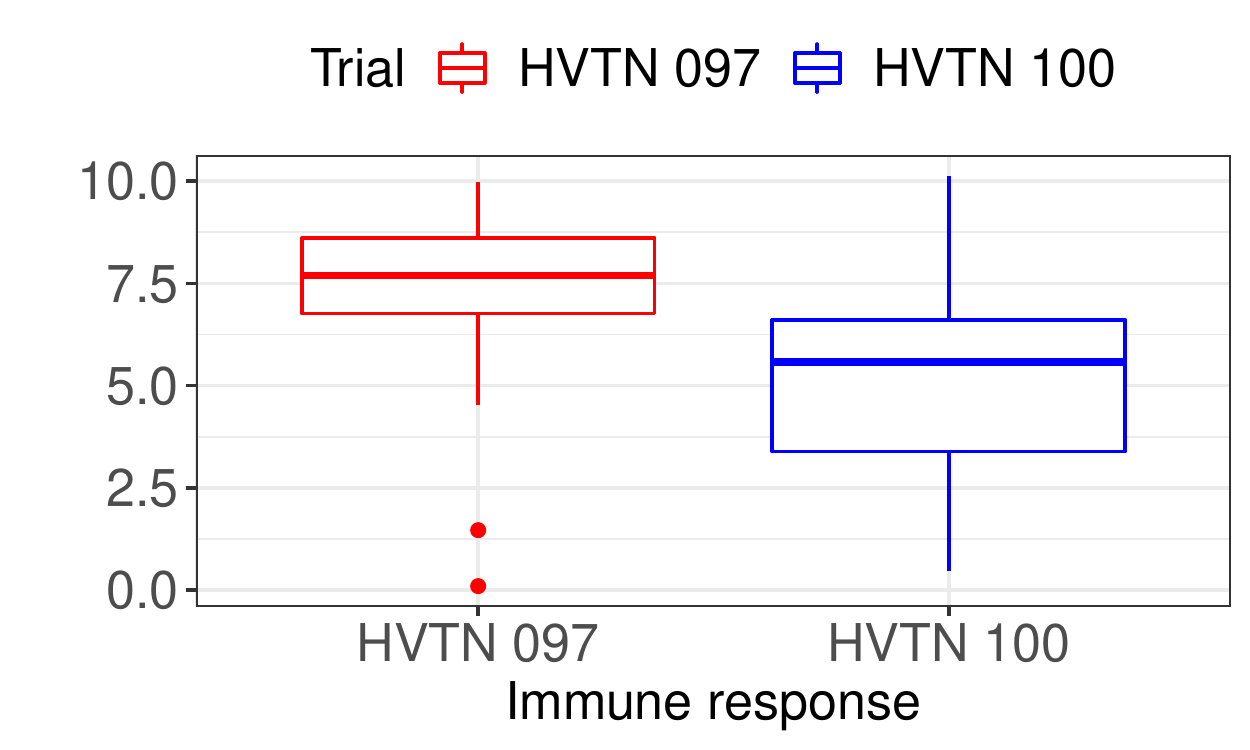}
 \caption{Boxplot}
 \label{Fig: box}
\end{subfigure}~
\begin{subfigure}{.4\textwidth}
 \includegraphics[width=\textwidth, height=1.5 in]{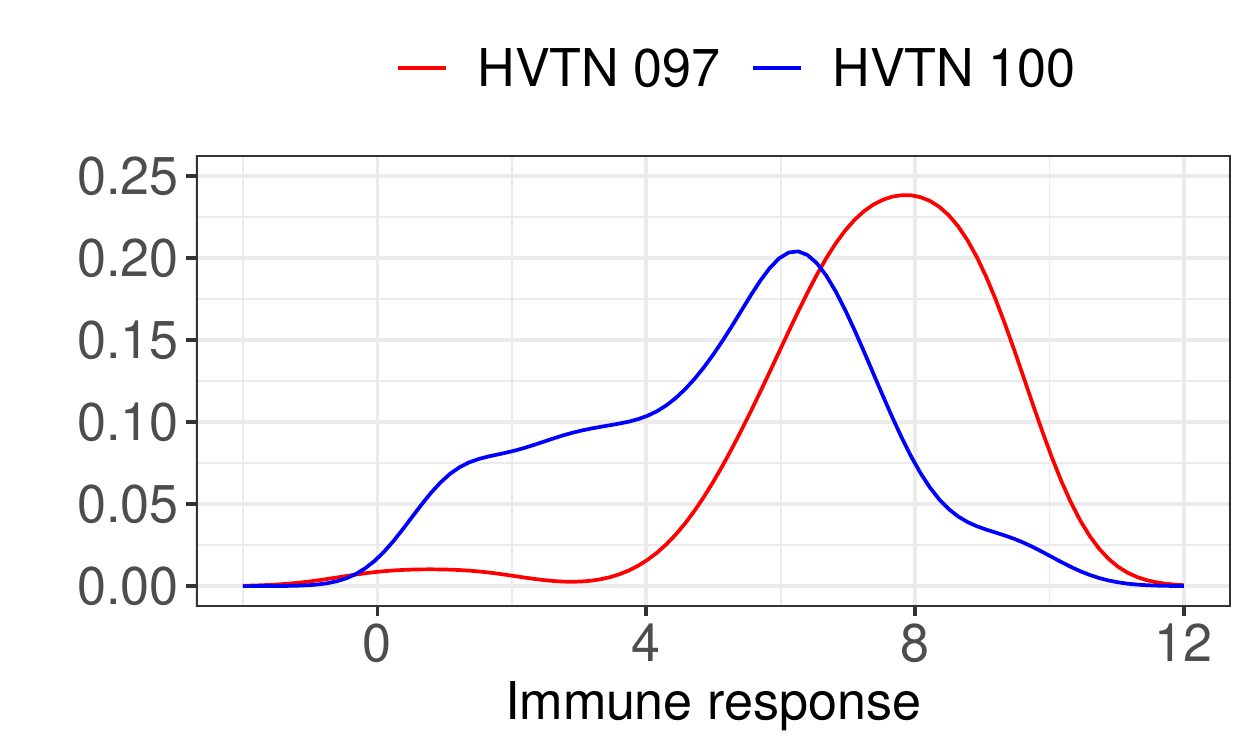}
 \caption{KDE}
 \label{Fig: kde}
\end{subfigure}
    \caption{
      Plots (i), (ii), (iii), (iv) display the empirical distribution functions (ECDF), the histogram, the boxplot, and the KDEs of the average IgG binding responses  corresponding to the HVTN 097 and HVTN 100 regimens.}
    \label{fig: preliminary plots}   
        \end{figure}


 Figure~\ref{fig: preliminary plots} illustrates the empirical CDFs, the histogram, the boxplot, and the KDES of the immune responses from the two trials.
 Figure~\ref{Fig: ecdf}  illustrates  that  $\Fh$ is always greater than $\Fn$   except  at the tails, hinting at the stochastic dominance of the HVTN 097 immune response over the HVTN 100 immune response. The histogram in Figure~\ref{Fig: Histogram}, the boxplot in Figure~\ref{Fig: box}, and the plot of the KDEs in Figure~\ref{Fig: kde}  also suggest that the  HVTN 097 trial induces higher immune response. \cite{HVTN100_primary} also indicated that the magnitude of positive responses in the HVTN 100 trial is lower than that of the RV144 trial, which, on the other hand, is reported to be slightly lower than  the responses  in the HVTN 097 trial \citep{HVTN097}.
 Therefore, it makes sense to posit the null  of non-dominance of $\Fn$ over $\Fh$ against the alternative that $\Fn$ stochastically dominates $\Fh$. Note also that  the comparisons between the two sets of immune responses  is reasonable  because the trials  were  conducted on similar populations and share approximately the same support (cf. Figure~\ref{fig: preliminary plots}). 
 See also Table~\ref{table: HVTN 097 and HVTN 100} in  Appendix~\ref{app: additional tables and figures} for a comparison between the two trials. 
  


\section{Density estimation}
\label{sec: density estimation}
This section  compares different  estimators of $f_{097}$ and $f_{100}$.
Since our study includes some unimodal and log-concave density estimators,  we begin by presenting some observations in support of the shape-restriction assumptions. \\
{
As mentioned previously, the unimodality assumption is not unreasonable owing to the homogeneity of the trial populations. 
The histogram in Figure~\ref{Fig: Histogram} and the KDEs displayed in Figure~\ref{Fig: kde} both support this claim. \textcolor{black}{Although unimodality is a naturally occurring shape constraint, the  class of all unimodal densities is too large to admit an MLE \citep{birge1997}.  The class of log-concave densities is a  subclass of the class of unimodal densities, which is small enough to admit an MLE \citep{exist}, but contains most of the commonly used subexponetial unimodal densities \citep{walther2009}. The  log-concave MLE can be computed efficiently using the R package \texttt{logcondens}. Furthermore, there is a smoothed version of the log-concave MLE, which is also free of tuning parameters \citep{smoothed}, and thus can potentially replace the smoothed unimodal density estimators, which generally depend crucially on external tuning parameters.
In view of the above, many researchers, e.g. \cite{walther2002} and \cite{walther2009}, advocate opting for log-concavity shape constraint in situations where unimodality may seem plausible.}

Although it is difficult to provide visual evidence in favor of the assumption of  log-concavity, the  KDE plot in Figure~\ref{Fig: kde}  does not indicate a departure from log-concavity either.
Using the test of log-concavity in \citeauthor{smoothed}, we test the null of log-concavity against the alternative of violation of log-concavity for $f_{097}$ and $f_{100}$. The corresponding p-values for $f_{097}$ and $f_{100}$ are $0.4890$ and $0.4631$, respectively, which implies that our data does not have enough evidence for rejecting the null of log-concavity. } 
        \begin{figure}[h]
    \centering
    \begin{subfigure}[b]{.3\textwidth}
     \includegraphics[width=\textwidth, height=\textwidth]{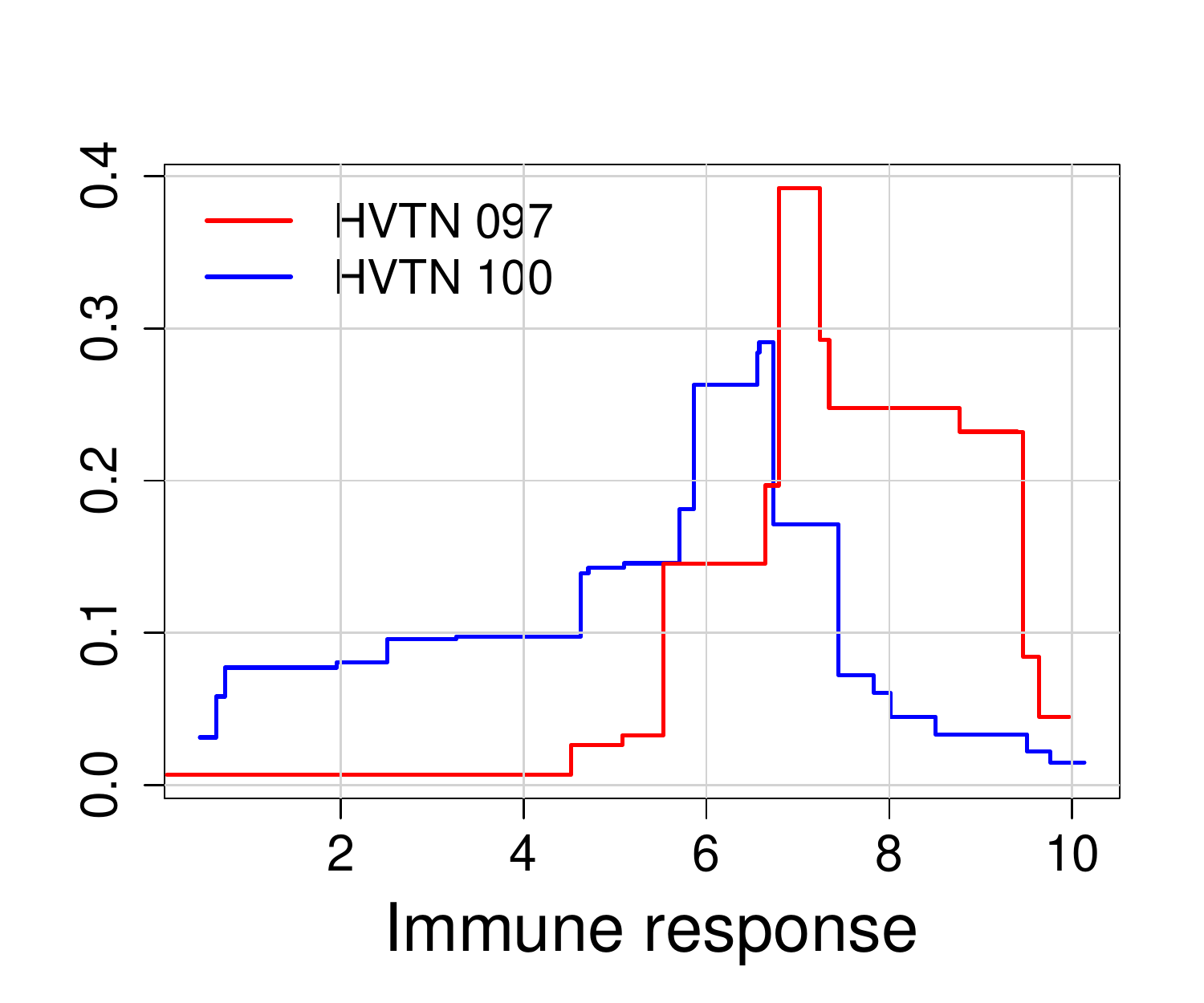}
      \end{subfigure}~
    \begin{subfigure}[b]{.3\textwidth}
        \includegraphics[width=\textwidth, height=\textwidth]{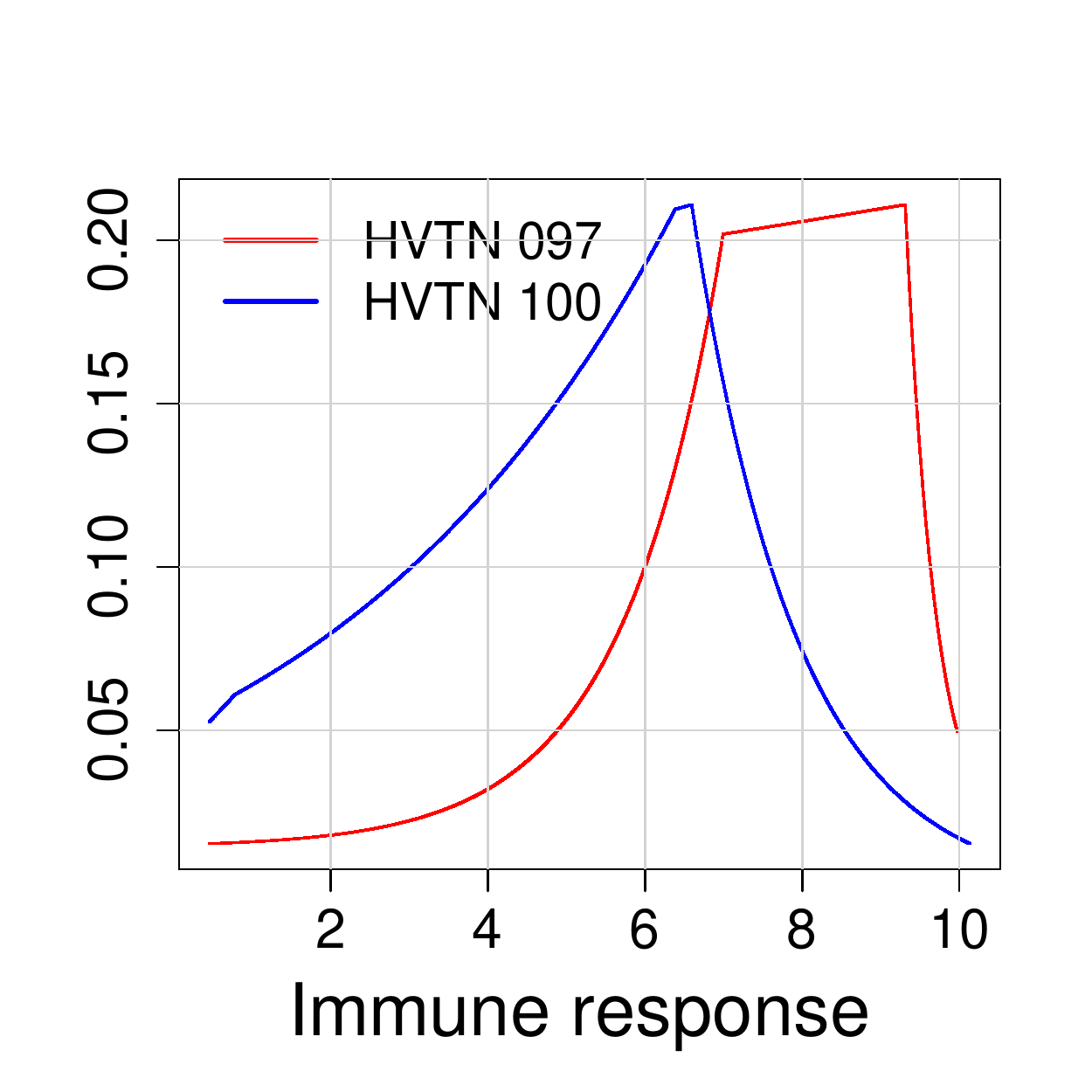}
    \end{subfigure}~
    \begin{subfigure}[b]{.3\textwidth}
     \includegraphics[width=\textwidth, height=\textwidth]{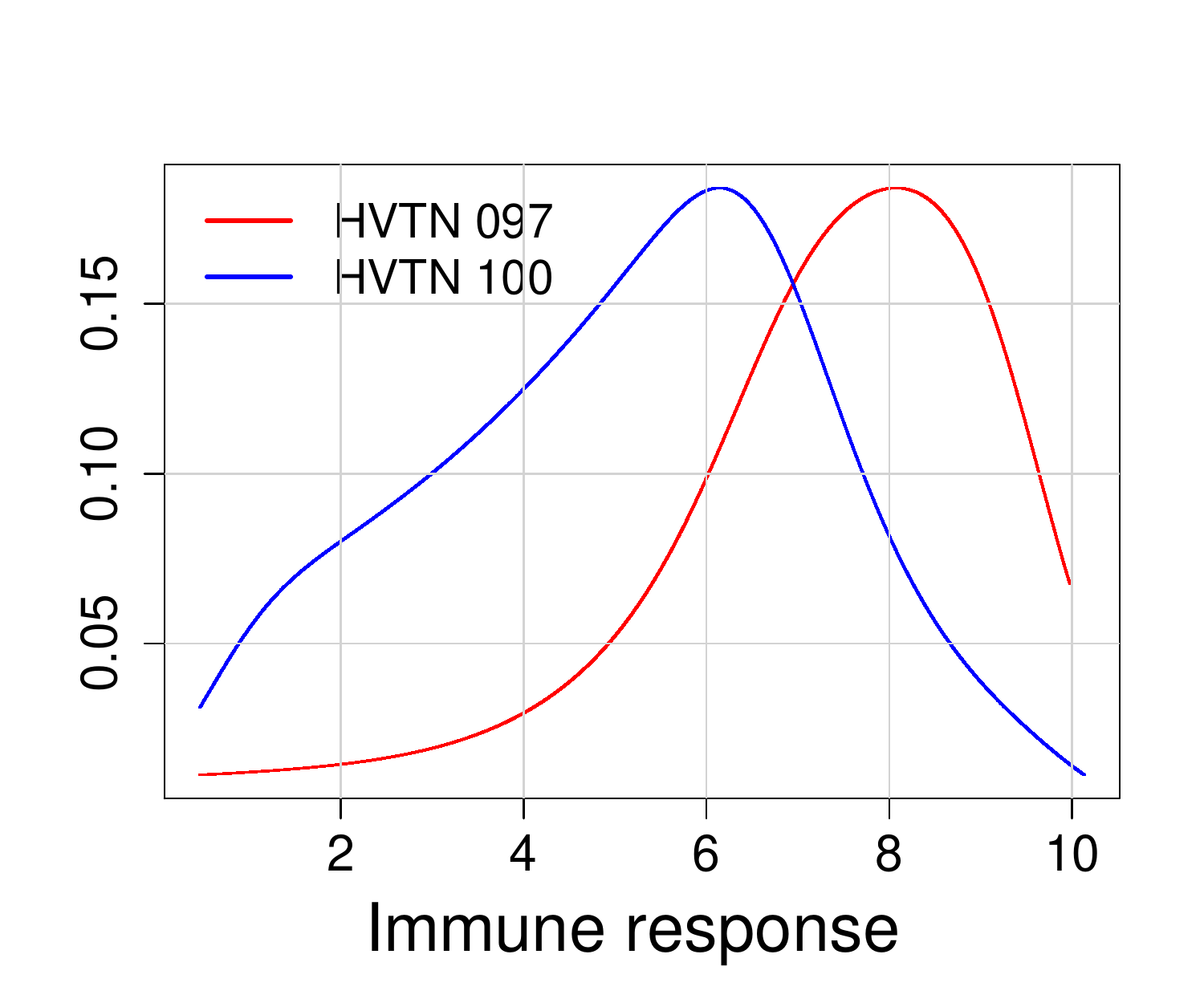}
    \end{subfigure}
    \caption{\birge's estimator (left), log-concave MLE (middle), and the smoothed log-concave MLE (right)  based on the immune responses in the HVTN 097 and the HVTN 100 trial.}   
    \label{fig:log-concave densities}
        \end{figure}
     We briefly describe below the  density estimators that we consider in this section.\\
     \textbf{Grenander-type unimodal  density estimator \citep{birge1997}:}
    \textcolor{black}{
    Although the class of all unimodal densities does not permit an MLE \citep{birge1997}, 
    when the mode of a unimodal density $\fx$ is known, \textcolor{black}{the MLE $\fnxm$ exists \citep{prokasa}, and it is a piecewise constant unimodal density with mode at the true mode.}
    However, the MLE $\fnxm$ is generally not useful due to the lack of knowledge on the location of the true mode.  The estimator $\fnx$ presented in \cite{birge1997} is a piecewise constant estimator of $\fx$, constructed in such a way so that the Kolmogorov-Smirnov distance between the corresponding distribution functions  can be made arbitrarily small, in particular, smaller than a pre-fixed  number $\eta>0$.
    Although this $\eta>0$ is an external parameter, unlike the kernel bandwidth, a smaller $\eta$ always leads to a more accurate estimation \citep{birge1997}, and hence it does not require actual tuning. Indeed, if we choose the parameter $\eta=o(m^{-1})$, then our Lemma~\ref{lemma: weak convergence of dist of unimodal density functions} in Appendix~\ref{sec: appendix: tests} ensures that  the total variation distance between  $\fnx$ and $\fnxm$ is $o(m^{-1/2})$ with probability one, \textcolor{black}{where $\fnxm$ was defined to be the MLE  of $\fx$ had the true mode been known. }
Therefore, we choose $\eta$ to be the inverse of the combined sample size of the two trials.}

   \textbf{Smooth unimodal estimator:} {\textcolor{black}{ There is a substantial body of literature on smooth unimodal density estimators. See for instance, \cite{Eggermont2000}, \cite{mammen2001}, \cite{weightedKDE},  \cite{greedy}, \cite{meyer2012}, \cite{bernstein}, and \cite{AdjustedKDE}, among others.  The smooth unimodal density estimators generally depend crucially on external tuning parameters, as mentioned previously.  Since we opt for shape constraints mainly to avoid tuning parameters and we already have at our disposal the tuning free smooth log-concave MLE estimator, the smooth unimodal density estimators are not particularly attractive to us.
Still, since the current section solely focuses on density estimation, we include some smooth unimodal estimators for comparison here, namely the estimators of \cite{bernstein}, \cite{weightedKDE},  \cite{greedy},  and \cite{AdjustedKDE}}.
   

\cite{bernstein} approximates the unknown unimodal density using Bernstein polynomial.  
We use the condition number approach of \cite{bernstein} to select tuning parameters related to the degree of the polynomial.
\textcolor{black}{ The estimators in \cite{weightedKDE},  \cite{greedy},   and \cite{AdjustedKDE} are kernel based. To compute them, we use the R package \texttt{scdensity} with the default choice of bandwidth.}

}


\textbf{Log-concave  estimators:}{ Our first log-concave estimator is the  MLE among the class of all log-concave densities.   The MLE  is continuous but non-smooth \citep{2009rufi}. The second estimator is a smoothed version of the MLE \citep{2009rufi,smoothed}. The smoothing parameter for the latter is data dependent and has a closed form formula, and hence it does not require external tuning.
Figure~\ref{fig:log-concave densities} displays these two density estimators. \textcolor{black}{For more on the  properties of the log-concave density estimators, see, e.g., \cite{balabdaoui2009}, \cite{dumbreg}, \cite{theory}, and  \cite{dossglobal}. }
}\\
\textbf{Kernel density estimators:}
{ We consider kernel density estimators (KDE) with Gaussian kernel.
The optimal bandwidth was chosen either by the  univariate plug-in selector of \cite{wand1994}, or the univariate least square cross-validation (LSCV) selector of \cite{bowman} and \cite{rudemo} (see Figure~\ref{Fig: kde}). }

 We prefer the estimator $\widehat f_m$ of a density $f$ with the smallest mean integrated squared error
 (MISE), which is given by
 $E\int_\RR (\widehat f_m (x)-\fx(x))^2 dx$. 
Noting minimizing the MISE with respect to $\widehat f_n$ is equivalent to minimizing
 \begin{equation}\label{criterion: Least square cross-validation}
 \text{MISE-err}= E\int_\RR\widehat f_m^2(x)dx-2E\int_\RR\widehat f_m(x)f(x)dx,
 \end{equation}
 we estimate the latter quantity using a ten folds cross-validation. We also estimate the negative log-likelihood 
 \begin{equation}\label{criterion: log-likelihood}
 -l_m=-m^{-1}\si \log \widehat f_m(X_i).
 \end{equation}
 
 \begin{table*}[h]
 \caption{Table of the estimated risks for different density estimators of the aggregated immune response in trials HVTN 097 and HVTN 100}\label{Table: cross validation}
     \begin{tabular}{@{}lllll@{}}
\toprule

  & HVTN 097 &  & HVTN 100 &  \\ 
 
Estimators   & MISE-err & $-l_m$ & MISE-err & $-l_m$ \\ 
 \midrule

Unimodal (\birge's estimator) & -0.172 & 1.662 & -0.126 & 2.071\\

 Unimodal (Bernstein) &  -0.191 & 1.797 & -0.130 & 2.133 \\ 
 
 Unimodal \citep{greedy} & -0.190  & 1.878 & -0.102 & 2.549\\ 
 
 Unimodal \citep{AdjustedKDE} & -0.190 & 1.795 & -0.124 &  2.300  \\ 

 Unimodal \citep{weightedKDE} & -0.190 & 1.965 & -0.124 & 2.288  \\ 
  Smooth log-concave MLE  & -0.196 & 1.758 &  -0.129 & 2.127 \\ 
  
Log-concave MLE  & -0.193  & 1.758 & -0.130 & 2.127\\ 
 KDE (plug-in bandwidth selector) & -0.189 & 1.750 & -0.128 & 2.140\\
 
 KDE (LSCV bandwidth selector) & -0.189 & 1.777 & -0.128 & 2.140\\
 \bottomrule
 \end{tabular} 
 \end{table*}
 
  Based on these risks  (see Table~\ref{Table: cross validation}), our recommended estimators are the log-concave estimators  which exhibit the lowest risk in an overall sense.
  Table~\ref{Table: cross validation}  indicates that \birge's estimator excels in minimizing the $-l_m$ risk but it has higher estimated MISE  when compared to the other estimators. This can be attributed to its spikes  at the mode (see Figure~\ref{fig:log-concave densities}), which contributes  large positive terms to  $l_n$ and MISE. \textcolor{black}{Grenander type unimodal estimators are known to exhibit such ``spike-problem"   at the mode \citep{walther2009}, which is caused by the inconsistency of the density estimator at the mode \citep[for more details, see][]{Woodroofe1993, balabdaoui2009}.}



\section{Test of stochastic dominance}
\label{sec: tests}
To provide an answer to  Q2, we construct tests for the null of non-dominance against that of stochastic dominance using the log-concave MLEs and the unimodal estimator of \birge. We compare the resulting shape-constrained tests with their nonparametric counterparts.

Our shape-restricted methods rely on estimating the densities $\fx$ and $\fy$. We denote the corresponding unimodal estimators of \birge\  by $\fnx$ and $\fny$,  respectively. The construction of \birge's estimators requires a tuning parameter $\eta$, which we set to be $N^{-1}$ where $N=m+n$.  We let $\flx$ and $\fly$ denote the log-concave MLEs of $\fx$ and $\fy$ \citep{2009rufi}, and write $\flxs$ and $\flys$ for their respective smoothed versions \citep{smoothed}. The corresponding distribution functions will be denoted by $\Fnx$, $\Fny$, $\Flx$, $\Fly$, $\Flxs$, and $\Flys$, respectively.


As $m$ and $n$  approach $\infty$, we assume that $ m/N\to\lambda\in(0,1)$.

 Letting $H=\lambda F+(1-\lambda)G$, for $p\in(0,1/2)$,
 we also define the sets 
  \begin{equation}\label{def: Dp}
  D_p(F,G):= D_p=[H^{-1}(p),H^{-1}(1-p)]\quad\text{ and }\quad D_{p,m,n}:=[\Hm^{-1}(p),\Hm^{-1}(1-p)].
  \end{equation}

\subsection{Construction of the tests}
\label{sec: null of non dominance}

Suppose $D\subset\supp(f)\cup\supp(g)$ is compact.
Following \cite{kaur} and \citeauthor{davidson2013}, we  formulate the hypotheses as follows:
 \begin{equation}\label{test: restricted stochastic dominance}
 \Hm_0: \Fx(z)\geq \Fy(z)\text{ for some }z\in D\quad\text{vs.}\quad \Hm_1: \Fx(z)<\Fy(z)\text{ for all }z\in D.
 \end{equation}
 \textcolor{black}{The null configuration $\mathcal H_0$ occurs if $F=G$, or $G$ stochastically dominates $F$, or if $F$ and $G$ touch or cross each other on $D$.}
  Thus our formulation is unable to reject the null when  $\Fx$ and $\Fy$ touch at a point in $D$ even if $\Fx\succ\Fy$.  This limitation seems to be unavoidable because such a configuration $(\Fx,\Fy)$ lies on the common boundary shared by $\{(F,G):F\succeq G\}$ and $\{(F,G):F\npreceq G\}$, and hence can not be discriminated from the null without sacrificing  control over the size of test. Notably, our exploratory analysis (see Fig~\ref{Fig: ecdf}) suggests that it is unlikely that $\Fn$ and $\Fh$ fall in this category.  See Figure~\ref{Plot: non-dominance hypotheses}  for examples of  different scenarios associated with our hypotheses.
 
 \begin{sloppypar}
  Regarding the choice of $D$, we need to ensure that $D$ is inside the combined support of $f$ and $g$ because otherwise,  ${\inf_{z\in D}[ G(z)-F(z)]}$ will always be $0$.
  The set $D_p$ defined in \eqref{def: Dp} satisfies this criterion.
  In practice, we replace this unknown  $D_p$ by $D_{p,m,n}$ defined in \eqref{def: Dp}, which always   utilizes  $100(1-2p)\%$ of the combined data.  
\textcolor{black}{Naturally, if $p$ is too small, rejection of the null will be difficult, where a large $p$ will exclude a large portion of the data, which might be unnecessary. If only some particular interval of the data is of practical interest (e.g. some particular range of immune responses or biomarkers), we suggest setting $D$  to be the smallest  superset of that interval.
  In the absence of such prior knowledge, we suggest choosing the largest $p$ so that $D_{p,m,n}$ excludes the  tail region where empirical distribution functions   overlap. We will return to this issue later in Section~\ref{sec: application to our data}, with a demonstration on our motivating dataset.} 
  \end{sloppypar}

\label{sec: tests: restricted stochastic dominance; design}
%
Now we are in a position to introduce our test statistics.

\textbf{Minimum t-statistic:}{ This statistic was first introduced by  \cite{kaur} in context of second order stochastic dominance, and then extended to the first order by  \citeauthor{davidson2013}. For  distribution functions $F_1$ and $F_2$, this statistic is given by
\begin{equation}\label{definition: T1}
T_{m,n}^{\text{min}}(F_1,F_2)=\inf_{x\in D_{p,m,n}} \dfrac{\slb F_2(x)-F_1(x)\srb }{\sqrt{\dfrac{F_1(x)\slb 1-F_1(x)\srb}{m}+\dfrac{F_2(x)\slb 1-F_2(x)\srb}{n}}}.
\end{equation}

\textcolor{black}{Our tests reject the $\Hm_0$ for large values of $T_{m,n}^{\text{min}}(\Fnx,\Fny)$, $T_{m,n}^{\text{min}}(\Flx,\Fly)$, and $T_{m,n}^{\text{min}}(\Fmx,\Fmy)$. }
The last test-statistic, which is nonparametric, equals the minimum t-statistic of the \citeauthor{kaur} in context of first order stochastic dominance. 
}

\textbf{Two sample empirical process (TSEP) type test statistic:}{Our second test rejects the $\Hm_0$ for large values of $T_{m,n}^{\text{tsep}}(\Fnx,\Fny)$, $T_{m,n}^{\text{tsep}}(\Flx,\Fly)$, or $T_{m,n}^{\text{tsep}}(\Fmx,\Fmy)$, where for distribution functions $F_1$ and $F_2$, $T_{m,n}^{\text{tsep}}$ is  defined by
\begin{equation}\label{definition: T2}
T_{m,n}^{\text{tsep}}(F_1,F_2)=\sqrt{\dfrac{mn}{N}}\inf_{z\in[p,1-p]}\dfrac{F_2(\Hm^{-1}(z))-F_1(\Hm^{-1}(z))}{\sqrt{z(1-z)}}.
\end{equation}
 \cite{LW2013} uses a  test statistic similar to $T_{m,n}^{\text{tsep}}(\Fmx,\Fmy)$ \citep[the second test statistic in Section 2.2 of][]{LW2013} for testing the null of stochastic dominance against  non-dominance. The pivotal distribution of their test statistic  is completely different from ours because they based their critical values on the configuration $F=G$, which is very different from what we will consider. We are not aware of any existing test which uses   $T_{m,n}^{\text{tsep}}(\Fmx,\Fmy)$  for testing  non-dominance against stochastic dominance. }


\textbf{Wilcoxon rank sum (WRS) type test statistic:}
 Wilcoxon rank sum (WRS) test is widely used for   comparing two vaccines \citep[cf.][]{miladinovic2014} although the WRS  test is  actually designed  for  testing location shift. It is  a popular choice for testing the null $F=G$ against the alternative $F=G(\cdot-\delta)$ for $\delta>0$ \citep{lee1976}.
The WRS test is the most powerful nonparametric test for testing the following hypotheses  \citep[cf. Example 25.46 of][]{vdv}:
 \begin{equation}\label{test: WRS}
           \Hm_0^a: \rint \Fy(z)d\Fx(z)\geq 1/2 \quad vs \quad \Hm_1^a: \rint \Fy(z)d\Fx(z)< 1/2.
           \end{equation} 
Although the WRS test is not  designed to test the null of non-dominance, we include this test to demonstrate its failure to control the type I error at some null configurations.

 { The  one-sided WRS test rejects $\Hm_0$ for large values of $T_{m,n}^{\text{wrs}}(\Fmx,\Fmy)$, where, for distribution functions $F_1$ and $F_2$, 
\begin{equation}\label{def: test statistic: T3}
T_{m,n}^{\text{wrs}}(F_1,F_2)=\sqrt{\dfrac{12mn}{N+1}}\lb\rint F_2(x)dF_1(x)-1/2\rb.
\end{equation}
$T_{m,n}^{\text{wrs}}(\Fmx,\Fmy)$ is the Mann-Whitney form of the two-sample WRS statistic.
The corresponding shape-constraint versions are given by $T_{m,n}^{\text{wrs}}(\Fnx,\Fny)$ and $T_{m,n}^{\text{wrs}}(\Flx,\Fly)$.}
 

We excluded tests based on the smoothed log-concave MLE because rigorous asymptotic analysis of the corresponding tests is out of the scope of the present paper.
 However,  our empirical study in Section \ref{sec:simulation:1}
 includes  minimum t-test and TSEP test based on the smoothed log-concave MLE. 
 Our simulations indicate that the asymptotic critical values of the  tests based on $T_{m,n}^{\text{min}}(\Flx,\Fly)$ and $T_{m,n}^{\text{tsep}}(\Flx,\Fly)$ are valid for the corresponding smoothed log-concave  tests. Our simulations also indicate that the finite sample performance of the tests based on the log-concave MLE and the smoothed log-concave MLE are quite similar. We leave the rigorous  analysis of the tests based on the smoothed log-concave MLE for future study.
 
  \begin{remark}
  Since the nonparametric tests use the empirical distribution function, shape constrained methods  do not gain any advantage in terms of tuning parameters. Also, we will see that the nonparametric and shape constrained tests  are asymptotically equivalent. For moderate sized samples, however, our simulations in Section~\ref{sec:simulation:1} show that the shape-constrained tests exhibit better performance. 
  \end{remark}

  \begin{remark}
   \textcolor{black}{ \citeauthor{davidson2013} proposed an  empirical likelihood ratio approach to test $\Hm_0$ vs $\Hm_1$. We did not appeal to this nonparametric approach in this paper because this approach does not extend easily to shape-constrained scenarios.}
   Regardless, we point out that \citeauthor{davidson2013} showed that their empirical likelihood ratio test  is  asymptotically equivalent to the minimum t-test.
  \end{remark}

\textcolor{black}{
 In the sequel, we may use the terms ``log-concave" or ``unimodal"  to refer to the test statistics based on the log-concave or unimodal density estimators. For example, we may refer to $T_{m,n}^{\text{min}}(\Fnx,\Fny)$ and $T_{m,n}^{\text{min}}(\Flx,\Fly)$ as the unimodal minimum t-statistic and the log-concave minimum t-statistic, respectively. 
 Also,  unless otherwise specified,
 the terms ``null" and ``alternative" will refer to the $\Hm_0$ and $\Hm_1$ defined in \eqref{test: restricted stochastic dominance}, respectively. 
 }
  \begin{figure}[h]
 \includegraphics[ width=\textwidth]{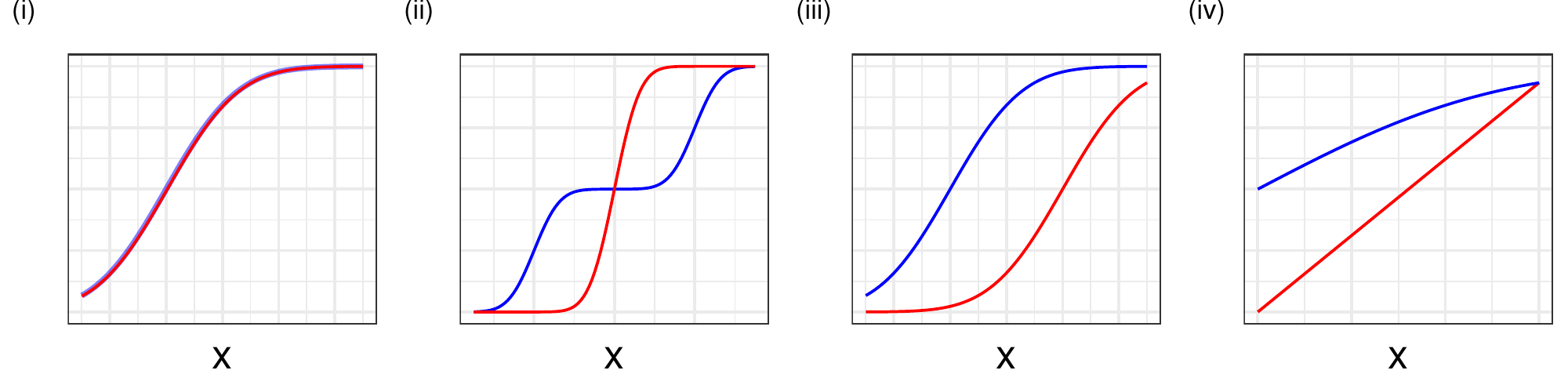} 
 \caption[foo bar]{
 This figure displays plots of two distribution functions $F$ (red) and $G$ (blue). The range of $x$ in these plots correspond to  $D_p:= D_p(F,G)$.
 (i) $F=G$ on $D_p$. This is a null configuration.
 (ii) $F$ and $G$ cross each other on $D_p$. This is also a null configuration.
 (iii) $F$ strictly stochastically dominates $G$. In fact $G(x)>F(x)$ for all $x\in D_p$. This is an alternative configuration.
 (iv) $F$ and $G$ touch each other at the endpoint of $D_p$. This is a null configuration.
 }\label{Plot: non-dominance hypotheses}
 \end{figure}
 
\subsection{Asymptotic distribution }
\label{section: asymptotic}

In this section, we explore the asymptotic distribution of  our test statistics. 
We show that the minimum t-test and the two sample empirical process (TSEP) test  asymptotically control type I error  for each null configuration and they are asymptotically consistent against each $(F,G)\in \Hm_1$.  We also show that, with the exception of the test based on log-concave MLE, the WRS type tests control the type I error for distributions in $\Hm_0^{a}$ and  are consistent against $(F,G)\in \Hm_1^{a}$.
We first prove the asymptotic results on the nonparametric test statistics.
 Then, we show that the shape-constrained test statistics are equivalent to their nonparametric counterparts up to a $o_p(N^{-1/2})$ term, which implies that the same critical values can be used for them.
 
 \textcolor{black}{We exclude the log-concave MLE based WRS statistic  from  our discussion because we are unable to infer on its asymptotic limit. The difficulty arises due to our inability to track the asymptotic behaviour of  $\sqrt{m}||\Flx-\Fmx||_{\infty}$. }In the remainder of this section, by  ``shape-constrained test statistics",  we will therefore refer to   $T_{m,n}^{\text{min}}(\Fnx,\Fny)$,  $T_{m,n}^{\text{tsep}}(\Fnx,\Fny)$, $T_{m,n}^{\text{wrs}}(\Fnx,\Fny)$, $T_{m,n}^{\text{min}}(\Flx,\Fly)$, and $T_{m,n}^{\text{tsep}}(\Flx,\Fly)$ only.  For the sake of clarity, in Table \ref{table: summarize test}, we summarize the current state of results on the different tests discussed in this paper.
 




 
   \begin{table*}[h]
   \centering
  \begin{tabular}{lcc}
\toprule
Test & Asymptotic results & Empirical results \\ 
\midrule
\multicolumn{3}{c}{Nonparametric}\\

\midrule

Minimum t-test & \makecell{previously known\\\citep{kaur}} & included in this paper\\
TSEP & we derived & $---''---$\\
WRS & \makecell{previously known\\ \citep{dwass1956}} & $---''---$\\
\midrule
\multicolumn{3}{c}{Unimodal}\\
\midrule
Minimum t-test & we derived & included in this paper\\
TSEP & $---''---$ & $---''---$\\
WRS & $---''---$ & $---''---$\\
\midrule
\multicolumn{3}{c}{Log-concave}\\
\midrule
Minimum t-test & we derived & included in this paper\\
TSEP & $---''---$ & $---''---$\\
WRS & unknown & $---''---$\\
\midrule
\multicolumn{3}{c}{Smoothed log-concave}\\
\midrule
Minimum t-test & unknown & included in this paper\\
TSEP & $---''---$ & $---''---$\\
WRS & $---''---$ & not included\\
\bottomrule
\end{tabular} 
\caption{Table summarizing  different tests used in this paper and current results on them. }
\label{table: summarize test}
  \end{table*}





 Before going into further details, we state a technical condition that will be required by all of our Theorems.
 \begin{cond}{N}\label{assump: Nonparametric}
 $F$ and $G$ are continuous. Also, $F$ and $G$ have densities $f$ and $g$, respectively, such that $\supp(f)\cup\supp(g)$ contains an open neighborhood of $D_p$.
 \end{cond}
The first requirement of Condition \ref{assump: Nonparametric}, i.e., the continuity of $F$ and $G$, is necessary for the weak convergence of the empirical processes to Brownian bridges. The second requirement ensures that $\text{dist}(D_{p,m,n},D_p)$ approaches $0$ as $m,n\to\infty$ with probability one. These assumptions are likely to be satisfied by our immune response data provided $p$ is not too small (see Figure~\ref{Fig: ecdf} and Figure \ref{Fig: Histogram}).
For the rest of the paper, we restrict our attention to $F$ and $G$ that satisfy Condition~\ref{assump: Nonparametric}.

 \subsubsection{Asymptotic critical values of the nonparametric tests}
\label{section: LFC}

We begin our discussion with the minimum t-statistic and the TSEP statistic.
Our first objective is to identify the  null  configurations that lead to the highest asymptotic type I error. \textcolor{black}{Here we remind the readers that $(F,G)$ is a null configuration if there exists $x\in D_p$ so that $G(x)\leq F(x)$.}
One may guess that the interesting cases appear on the boundary of $\Hm_0$. However, to formally discuss the boundary of  $\Hm_0$, we need to equip it with a suitable topology. To formalize our discussion, we consider the space $\mathcal F$ of all continuous distribution functions on $\RR$, and equip it with the uniform metric
${d(F,F')=\sup_{x\in \RR}|F(x)-F'(x)|.}$
Consider the product space $\mathcal F\times \mathcal F$ with the  metric
\[d_2\slb (F,G),(F',G')\srb =\max\lbs d(F,F'),d(G,G')\rbs.\]

By an abuse of notation, we  denote by $\Hm_0$ and $\Hm_1$ the set of all combinations $(F,G)\in\mathcal F\times\mathcal F$   that satisfy the hypotheses $\Hm_0$ and $\Hm_1$, respectively. For $i=0,1$, we denote the closure of $\Hm_i$ in $\mathcal F\times \mathcal F$ by $\text{cl}(\Hm_i)$. Then the boundary of $\Hm_i$ is given by $\text{cl}(\Hm_i)\setminus \iint (\Hm_i)$.
The following lemma characterizes $\bd(\Hm_0)$ and $\iint(\Hm_0)$.

\textcolor{black}{
\begin{lemma}\label{Lemma: geometry of H0}
$\Hm_0$ is a closed subset of $\mathcal F\times\mathcal F$ with boundary
\[\bd (\Hm_0)=\lbs(F,G)\in\mathcal F\times\mathcal F\ :\ \sup_{x\in D_p}\slb F(x)-G(x)\srb=0 \rbs.\]
Moreover, $\bd(\Hm_0)=\bd(\Hm_1)$. Also, the interior of $\Hm_0$ is given by
\[\iint(\Hm_0)=\lbs(F,G)\in\mathcal F\times\mathcal F\ :\ F(z)>G(z)\text{ for some }z\in D_p \rbs.\]
\end{lemma}
}

\textcolor{black}{
 Figure~\ref{Plot: non-dominance hypotheses}(ii) 
 gives an example of an $(F,G)$ pair in the interior of $\Hm_0$.
 The following lemma entails that the minimum t-statistic and the TSEP statistic are asymptotically degenerate  on $\iint(\Hm_0)$. 
   \begin{lemma}\label{thm: dist based: null distribution: A}
 Suppose  $(F,G)\in\iint(\Hm_0)$ satisfies Condition~\ref{assump: Nonparametric},  $m/N\to\lambda$, and  $\supp(f)\cup\supp(g)$ contains an open neighborhood of $D_p$. Then, 
\[T_{m,n}^{\text{min}}(\Fmx,\Fmy)\to_p -\infty\quad\text{and}\quad  T_{m,n}^{\text{tsep}}(\Fmx,\Fmy)\to_p -\infty.\]
\end{lemma}
}

\textcolor{black}{
The proof of Lemma~\ref{thm: dist based: null distribution: A} for the minimum t-statistic can be found in \cite{whang2019}  \citep[see also][]{davidson2013}. However, we include it in Appendix~\ref{thm: dist based: null distribution: A} for the sake of completeness.
Lemma~\ref{thm: dist based: null distribution: A} indicates that non-trivial type I errors can originate only at the boundary of $\Hm_0$, which is a subset of $\Hm_0$ because the latter is a closed set (see Lemma~\ref{Lemma: geometry of H0}). 
Lemma~\ref{Lemma: geometry of H0} also implies that  $\text{bd}(\Hm_0)$ consists of all those $F$ and $G$ that touch each other on $D_p$. \textcolor{black}{To concretize this idea, we define the contact set $C_p$ by
\begin{equation*}
    C_p=\{ x\in D_p\ :\ F(x)=G(x) \}.
\end{equation*}
Note that if $(F,G)\in\text{bd}(\Hm_0)$, then
 $C_p\neq\emptyset$, where $F=G$ on $C_p$, and $F<G$ on $ D_p\setminus C_p$.}  Let us also define
 \begin{equation}\label{def: B3}
H(C_p)=\{t\in[p,1-p]\ |\ H^{-1}(t)\in C_p\}.
\end{equation}
Theorem~\ref{thm: dist based: null distribution: B} shows that the asymptotic  distribution of the test statistics on $\bd(\Hm_0)$ crucially depends on this contact set $C_p$ and $H(C_p)$. The proof of  Theorem~\ref{thm: dist based: null distribution: B} is given in Appendix~\ref{sec:proof: main thm}.
}
\begin{theorem}\label{thm: dist based: null distribution: B}
 Suppose  $(F,G)\in\bd(\Hm_0)$,  $m/N\to\lambda$,  and  $F$ and $G$ have continuous densities $f$ and $g$ satisfying
 \[\inf_{x\in D_p}\min\{f(x),g(x)\}>0,\] 
 where $D_p=[H^{-1}(p),H^{-1}(1-p)]$.
   Let $\mathbb{U}$ denote a standard Brownian bridge.
   Then under the stated conditions, the following assertions hold:
   \begin{itemize}
       \item[A.]  \begin{align}\label{eq:asT1mnlim}
 T_{m,n}^{\text{min}}(\Fmx,\Fmy)\to_d\inf_{x\in C_p}\dfrac{\mathbb{U}\circ F(x)}{\sqrt{F(x)(1-F(x))}},
\end{align}
where $ C_p=\{ x\in D_p\ :\ F(x)=G(x) \}$.
\item[B.]
\[T_{m,n}^{\text{tsep}}(\Fmx,\Fmy)\to_d \inf_{t\in H(C_p)}\sqrt{\dfrac{\lambda}{1-\lambda}}\dfrac{\mathbb{L}_0(t)}{\sqrt{t(1-t)}},\]
where $\mathbb{L}_0$ is the centred Gaussian process given by \eqref{def: L0} of Appendix~\ref{sec: appendix: tests}, and $H(C_p)$ is as in \eqref{def: B3}.
\item[C.] In particular, if $C_p\subset\iint(D_p)$, then 
\[T_{m,n}^{\text{tsep}}(\Fmx,\Fmy)\to_d \inf_{t\in H(C_p)}\dfrac{\mathbb{U}(t)}{\sqrt{t(1-t)}}.\]
   \end{itemize}
  
 \end{theorem}
 
 \textcolor{black}{The Gaussian process $\mathbb L_0$, which is defined in \eqref{def: L0} of Appendix~\ref{sec: appendix: tests}, depends on $F$ and $G$. We postpone further discussion on the form of $\mathbb L_0$ till Appendix~\ref{sec: appendix: tests}. Next we discuss the implication of Theorem \ref{thm: dist based: null distribution: B} on the minimum t-test. Then we will discuss the case of the TSEP test.}
 
 \textcolor{black}{
 \textbf{Asymptotic critical value of minimum t-test:}
  Theorem~\ref{thm: dist based: null distribution: B} reveals an interesting fact: the length of $C_p$ imposes a stochastic ordering among the limiting laws of the minimum t-statistic for the boundary configurations.
 To elaborate, let us consider $(F_1,G_1)$ and $(F_2,G_2)\in\bd(\Hm_0)$, with respective contact sets $C_{p}^1$ and $C_p^2$. Then, 
   under the conditions of Theorem~\ref{thm: dist based: null distribution: B}, 
  \[T_{m,n}^{\text{min}}(\Fmx,\Fmy)\xrightarrow{F_i,G_i}_d\inf_{x\in C^i_p}\dfrac{\mathbb{U}\circ F(x)}{\sqrt{F(x)(1-F(x))}},\quad\text{ for }i=1,2.\]
If the contact sets satisfy the ordering $C_{p}^1\subset C_{p}^2$, then
   \[\inf_{z\in C^2_p}\dfrac{\mathbb{U}\circ F(z)}{\sqrt{F(z)(1-F(z))}}\preceq\inf_{z\in C^1_p}\dfrac{\mathbb{U}\circ F(z)}{\sqrt{F(z)(1-F(z))}},\]
   implying that the limiting law of the minimum t-statistic under $(F_1,G_1)$ stochastically dominates that under $(F_2,G_2)$.
   The extreme cases for $C_p^1$ are the singleton sets $\{x\}$, where $x\in D_p$. In this case, the asymptotic distribution of both test  statistics is standard Gaussian.
   Therefore, we set the critical value of our tests to be $z_{\alpha}$, the $(1-\alpha)$-th quantile of the standard Gaussian distribution.
   }
   \textcolor{black}{
   The  class of boundary configurations with a singleton contact set is referred to as ``the least favorable class" (LFC) \citep{davidson2013}.
    Figure~\ref{figure: boundary of Hknot}  illustrates the difference between an LFC and an ordinary non-LFC  boundary combination.
   Theorem 2 of \citeauthor{davidson2013} shows that   there is no null configuration under which  the law of the minimum t-statistic  strictly stochastically dominates that of the LFC configuration. This result, which holds for any $m$ and $n$, implies that, among the null configurations, the LFC configurations lead to the greatest dominance of $F$ over $G$. 
       The above finding, in conjunction with our  Theorem~\ref{thm: dist based: null distribution: B}, imply that our tests, whose critical values are based on the LFC class, is likely to have asymptotic size $\alpha$. This being a stronger assertion than the asymptotic control of type I error can be an interesting topic for further investigation. 
    }
    \begin{figure}[h]\label{Fig: B}
\begin{subfigure}[b]{.5\textwidth}
\centering
\includegraphics[height=1 in, width=\textwidth]{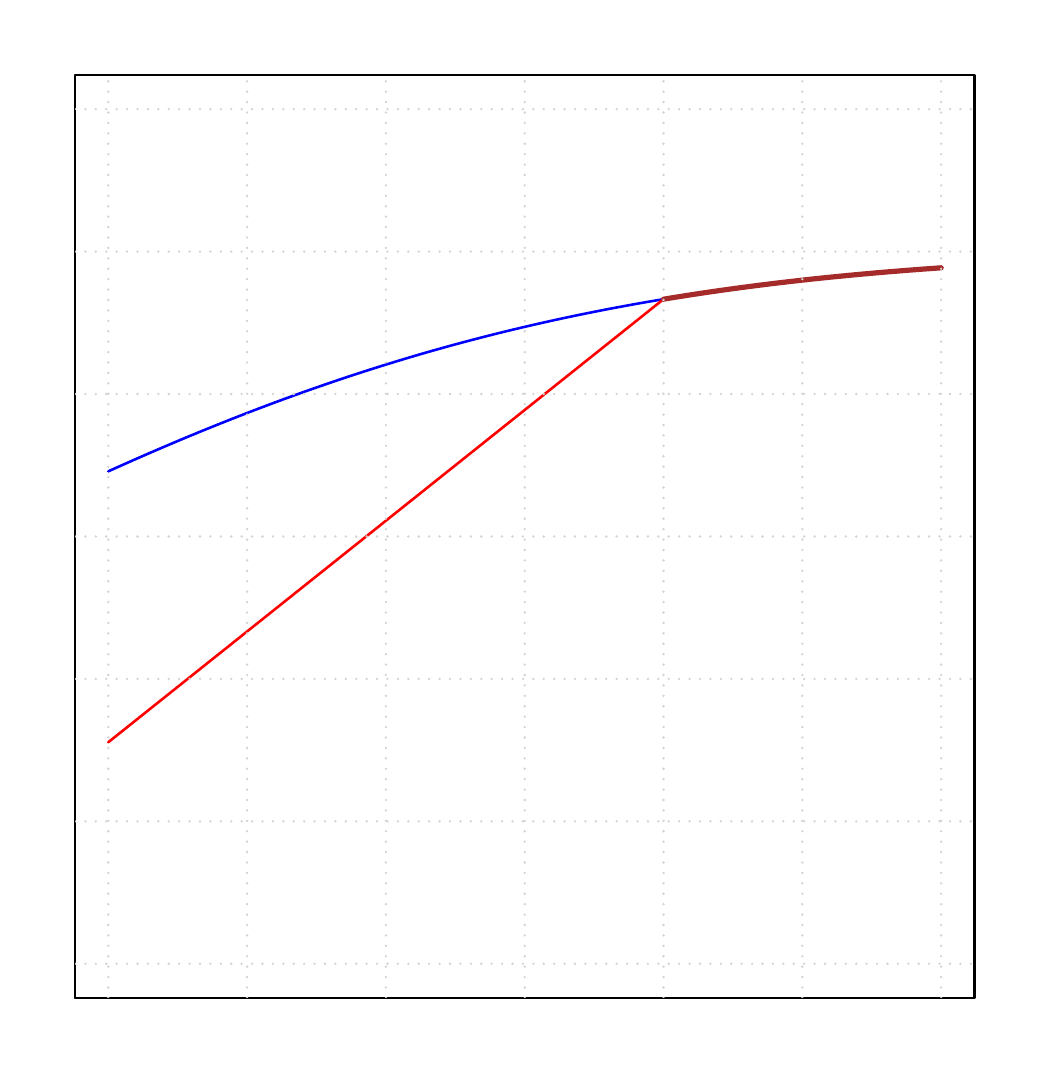}
\caption{$(F,G)\in\bd(\Hm_0)$ but not LFC}
\label{Figure: touch: region}
\end{subfigure}~
\begin{subfigure}[b]{.5\textwidth}
\centering
\includegraphics[height=1  in, width=\textwidth]{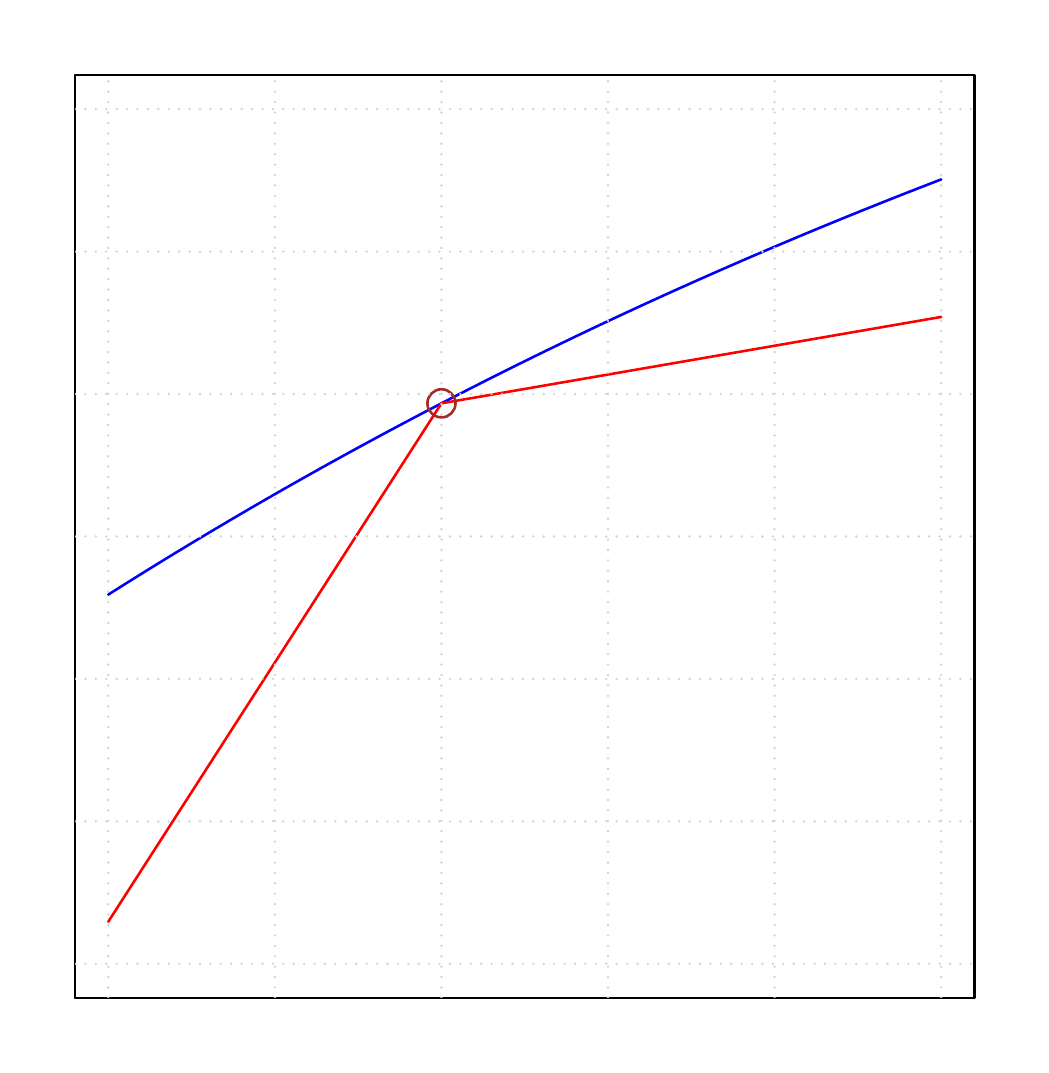}
\caption{$(F,G)$ is LFC}
\label{Figure: touch: LFC}
\end{subfigure}
\caption{Illustration of $F$ (red) and $G$ (blue) on $\bd(\Hm_0)$. (i) The contact set $C_p$ is an interval (ii) The contact set is singleton; so this is an LFC configuration.}\label{figure: boundary of Hknot}
\end{figure}
    \begin{remark}
 The asymptotic behavior of  the nonparametric  minimum t-statistic has been previously studied \citep{davidson2013,  kaur}. However,  previous studies  focus only on the asymptotic behaviour of the minimum t-statistic at the LFC configurations and the interior of $\Hm_0$, whereas our results show that there are other  classes of  boundary configurations with non-vanishing type I error.
Although existing results are enough for the purpose of constructing critical values, our new results provide a more complete understanding of the scenario.
\end{remark}

    \textcolor{black}{
    \textbf{Asymptotic critical value of the TSEP test:}  The asymptotic distribution of the TSEP statistic under $\Hm_0$ also exhibits a monotonocity property similar to the minimum t-test. To elaborate, suppose the pairs $(F_1,G_1)$ and $(F_2,G_2)$ have respective contact sets $C_p^1$ and $C_p^2$ satisfying $C_p^1\subset C_p^2$. Then
    \[\inf_{t\in H(C_p^2)}\frac{\mathbb L_0(t)}{\sqrt{t(1-t)}}\preceq \inf_{t\in H(C_p^1)}\frac{\mathbb L_0(t)}{\sqrt{t(1-t)}}.\]
Therefore, similar to the case of the minimum t-test,  the configurations with singleton $H(C_p)$ constitute the class of  LFC configurations for the TSEP test. Under Condition \ref{assump: Nonparametric}, $H$ is continuous and strictly increasing on $C_p$. Therefore, $H(C_p)$ is singletone if and only if $C_p$ is singletone, i.e., $C_p$ is of the form $\{x\}$, where $x\in D_p$. To find the critical value, it suffices to study the asymptotics in such LFC cases. If $x\in\iint(D_p)$, the  TSEP test statistic weakly converges to a standard Gaussian distribution  by part C of  Theorem~\ref{thm: dist based: null distribution: B}. Till this point, there has been no difference between the TSEP and the minimum t-test statistic. If, however, $x\in\text{bd}(D_p)$ (see Figure~\ref{Fig: Cp, etc.}iv), the asymptotic distribution of the TSEP test statistic can be different. To this end, first we state a lemma, and then using this lemma, we explain the asymptotics of the TSEP test when $C_p=\{x\}\subset \{p,1-p\}$.
    }
    
    \textcolor{black}{\begin{lemma}
\label{Dist of L0: at bd}
Suppose $F$ and $G$ are as in Theorem \ref{thm: dist based: null distribution: B} and $H(C_p)=\{t\}$ for some $t\in[p,1-p]$. Then
\[\sqrt{\frac{\lambda}{1-\lambda}}\dfrac{\mathbb L_0(t)}{\sqrt{t(1-t)}}\sim N(0,\sigma_{TSEP}^2),\] 
where
\begin{equation}
    \label{def: sigma TSEP}
    \sigma_{TSEP}^2=\frac{\lambda (f\circ H^{-1}(t))^2+(1-\lambda)(g\circ H^{-1}(t))^2}{\slb\lambda f\circ H^{-1}(t)+(1-\lambda)g\circ H^{-1}(t)\srb^2}.
\end{equation}
Moreover, the following assertions also hold:
\begin{itemize}
    \item [A.]  $\sigma_{TSEP}^2=1$ if and only if $f\circ H^{-1}(t)=g\circ H^{-1}(t)$. Otherwise,
\begin{equation}
    \label{L0 var: bounds}
   1< \sigma_{TSEP}^2\leq \max\{(1-\lambda)^{-1}, \lambda^{-1}\}.
\end{equation}
\item[B.] Given any  $\e\in(0,2(1-\lambda)/\lambda)$, we can find  a constant $C_\lambda>0$, depending only on $\lambda>0$, so that whenever  $g(H^{-1}(t))/f(H^{-1}(t))<C_\lambda$, then $\sigma^2_{TSEP}>\lambda^{-1}-\e.$
\item[C.]  Given any  $\e\in(0,2\lambda/(1-\lambda))$, we can find  a constant $C'_\lambda>0$, depending only on $\lambda>0$, so that whenever $f(H^{-1}(t))/g(H^{-1}(t))<C'_\lambda$, then $\sigma^2_{TSEP}>(1-\lambda)^{-1}-\e.$
\end{itemize}
\end{lemma}}

\textcolor{black}{Lemma~\ref{Dist of L0: at bd} has some interesting consequences. First, if 
$C_p=\{H^{-1}(p)\}$ or $\{H^{-1}(1-p)\}$, then using Theorem~\ref{thm: dist based: null distribution: B}B and Lemma~\ref{Dist of L0: at bd} one can show that $T_{m,n}^{\text{tsep}}$ converges weakly to a centred Gaussian distribution with variance $\sigma^2_{TSEP}$. Part A of Lemma~\ref{Dist of L0: at bd}  implies that if $F$ and $G$ touch at $C_p$, i.e. if $f=g$ at the point of contact, then $\sigma^2_{TSEP}$ is still one. Hence, the TSEP test statistic is asymptotically standard Gaussian for this case. However, if $F$ and $G$ cross at the point of contact instead of touching, i.e. if $f\neq g$ at the point of contact, then $\sigma^2_{TSEP}>1$. The precise value of $\sigma^2_{TSEP}$ is given by \eqref{def: sigma TSEP}. Moreover, the value of $\sigma^2_{TSEP}$ increases as the value of $f$ and $g$ diverges at the point of contact. On one hand,   if $f$ is much larger than $g$, then part B of Lemma~\ref{Dist of L0: at bd} implies $\sigma^2_{TSEP}$ is close to $\lambda^{-1}$. On the other hand, if $g$ is much larger than $f$, then  part C of Lemma~\ref{Dist of L0: at bd} indicates that $\sigma^2_{TSEP}$ is close to $(1-\lambda)^{-1}$. These bounds are tight because, by part B of Lemma~\ref{Dist of L0: at bd}, $\sigma^2_{TSEP}$
can not be larger than $\max\{\lambda^{-1},(1-\lambda)^{-1}\}$ under the set up of Theorem~\ref{thm: dist based: null distribution: B}. 
}

\textcolor{black}{The above discussion leads to the following conclusion for the TSEP test. If we want to control the asymptotic type I error of the TSEP test at all null configurations, then we should use the critical value $C_{m,n}z_{\alpha}$ where $C_{m,n}=\max\{\sqrt{N/m},\sqrt{N/n}\}$. There is a caveat, however.
To see this, we begin by noting that $C_{m,n}\geq \sqrt 2$, with equality holding only when $m=n$.   For example, for our motivating dataset, $C_{m,n}\approx 1.9$. However,  if  $m/N\to 0$ or $n/N\to 0$,   then $C_{m,n}\to\infty$. Thus if $m$ and $n$ are not close to being equal,  $C_{m,n}$ can be a large quantity. Therefore, using $C_{m,n}z_\alpha$ as critical value yields a conservative test. Hence, we will call the corresponding TSEP tests the conservative TSEP tests. Our simulations indicate that   conservative TSEP tests have poorer power  compared to the minimum t-test even when $m$ and $n$ are  equal, and their power keeps degrading  as $C_{m,n}$ increases. 
}

\textcolor{black}{
If we use instead use the critical value $z_\alpha$, then we  control the asymptotic type I error at the null configurations with $C_p\subset\iint(D_p)$ or those with $f=g$ at $C_p$.
This only excludes the null cases where $F$ and $G$ may cross at $\text{bd}(D_p)$ (see Figure \ref{Plot: non-dominance hypotheses} iv). These type of configurations  can be considered  pathological cases.  Moreover, our simulations in  Section \ref{sec:simulation:1} (see case b)
show that even when $F$ and $G$ cross at $\text{bd}(D_p)$, the TSEP tests with critical value $z_{\alpha}$ control the type I error. The TSEP tests with critical value $z_\alpha$ have decent power and their overall performance  is comparable with the minimum t-tests. In view of the above, we recommend using the asymptotic critical value $z_\alpha$ when using the TSEP test.
}

    Our final result on the minimum t-test and the TSEP test
  establishes their asymptotic consistency.
     \begin{theorem}\label{thm:critical values}
Suppose  $(F,G)\in \Hm_1$ satisfy Condition~\ref{assump: Nonparametric}. Then if  $m/N\to\lambda$, then
\[\lim_{m,n\to\infty} T_{m,n}^{\text{min}}(\Fmx,\Fmy)\to_p \infty\quad\text{and}\quad T_{m,n}^{\text{tsep}}(\Fmx,\Fmy)\to_p \infty.\]
\end{theorem}

\textcolor{black}{
The asymptotic distribution of the WRS statistic is well-established in the literature. Suppose  $\Fx$ and $\Fy$ are continuous distribution functions.  In that case, it is  well known  that, when $\Fx=\Fy$, the WRS statistic is asymptotically distributed as a standard gaussian random variable, i.e. $T_{m,n}^{\text{wrs}}(\Fmx,\Fmy)\to_d N(0,1)$ \citep{dwass1956}. If $(\Fx,\Fy)\in H^a_0$ satisfies $\rint G(z)dF(z)>1/2$, however, $ T_{m,n}^{\text{wrs}}(\Fmx,\Fmy)\to_p -\infty$, whereas for $(F,G)\in \Hm_1^{a}$, we have $ T_{m,n}^{\text{wrs}}(\Fmx,\Fmy)\to_p \infty$.}

 \subsubsection{Asymptotic critical values of the shape-constrained tests} 
We will show that under some additional conditions, the difference between the nonparametric and the  shape-constrained  test statistics  is $o_p(1)$, which automatically implies that the shape-constrained tests enjoy the same asymptotic properties as the nonparametric tests.


\textcolor{black}{
For the unimodal case, the additional condition is a curvature condition, which requires $\fx$ and $\fy$ to be nowhere flat within their respective domains. 
\begin{cond}{A}\label{Cond A}
For the density $\mu$, the Lebesgue measure of the set $\{\mu'=0, \mu>0\}$ is $0$, where $\mu'$ is the derivative of $\mu$.
\end{cond}
}

  
 \textcolor{black}{
 For densities satisfying Condition \ref{Cond A}, $\sqrt{m}(\widehat F_m^0-F)$ almost surely weakly converges to $\mathbb V\circ F$, where $\mathbb V$ is a Brownian bridge. In this case, it can be shown that \citep{kiefer}
$\sqrt{m}\|\Fnx^0-\Fmx\|_\infty=o_p(1)$. Our Lemma 
\ref{lemma: weak convergence of dist of unimodal density functions} in
Appendix \ref{app: shpe constrained tests} states that  $\sqrt{m}\|\widehat F_m^0-\Fnx\|_\infty=o_p(1)$, which implies $\sqrt{m}\|\Fnx-\Fmx\|_\infty=o_p(1)$ in this case. However, under the violation of
 Condition \ref{Cond A}, the process $\sqrt{m}(\widehat F_m^0-F)$ no longer  converges to $\mathbb V\circ F$ weakly.  The densities that 
violate Condition \ref{Cond A} form the boundary of the class of unimodal densities. For these densities, $\sqrt{m}\|\widehat F_m^0-\Fmx\|_\infty$, and hence,  $\sqrt{m}\|\Fnx-\Fmx\|_\infty$ is non-negligible.  The limiting process of $\sqrt{m}(\widehat F_m^0-F)$ in this case is slightly convoluted, and we refer to \cite{beare2017, carolan1999, carolan2002} for more details on the limiting process.  Just to give an example, if $f\sim U[0,1]$, then the limiting process is the least concave majorant of a Brownian bridge \citep{carolan1999}.  Condition \ref{Cond A} is thus required to ensure that neither $\fx$ nor  $\fy$ is one of these problematic boundary densities.
}
 
 
\begin{lemma}\label{theorem: convergence: Ti: unimodal }
Suppose that $\fx$ and $\fy$ are unimodal densities satisfying Condition~\ref{Cond A}. Further suppose that $\fx$ and $\fy$ are bounded away from $0$ on an open set containing  $D_p$, and $m,n$ satisfy  $m/N\to\lambda$. Then, 
\[|T_i(\Fmx,\Fmy)-T_i(\Fnx,\Fny)|\as 0,\quad\text{ for }i=1,2,3.\]
\end{lemma}

\textcolor{black}{
In case  of the log-concave test statistics, however, we require a smoothness condition as well as a curvature condition. We will quantify smoothness via a H\"older condition.  For a compact set $K\subset\RR$, a function $h$ is said to be in the H\"older class $\mathcal{H}^{\beta,L}(K)$
 with exponent $\beta\in[1,2]$ and constant $L>0$ if, for  all $x,y\in K$, $|h(x)-h(y)|\leq L|x-y|$ if $\beta=1$ and
 $ |h'(x)-h'(y)|\leq  L|x-y|^{\beta-1}\quad\text{if }\beta>1.$}\\
  We say that a density $\mu$ ($\mu=\fx$ or $\fy$) satisfies Condition~B1 if the following holds.
 \begin{cond}{B1}\label{Cond: B1}
There exists $\beta\in[1,2]$, $L>0$, and a compact $K\subset\RR$ such that the density $\mu$ satisfies $\log \mu\in\mathcal{H}^{\beta,L}(K)$.
\end{cond}

In addition to Condition~\ref{Cond: B1}, we also require $\fx$ and $\fy$ to satisfy a curvature condition. 



\begin{cond}{B2}\label{Cond: B2}
 $\mu$ is a log-concave density with log-density $\phi=\log\mu$. Suppose $K\subset\dom(\phi)$ is compact. Then there exists $C>0$ such that  all $x,y\in K$ satisfying $x<y$ obeys
 
 {\centering
  $ \displaystyle
    \begin{aligned}
    \phi'(x)-\phi'(y)\geq C(y-x),
    \end{aligned}
  $ 
\par}
 where $\phi'$ is the left derivative or the right derivative of $\phi$.
\end{cond}
\textcolor{black}{
Note that, since $\phi$ is concave, its left and right derivatives always exist.  If $\phi'$ is differentiable on $K$, Condition~\ref{Cond: B2} reads as $\phi''(x)\leq -C$ for $x\in K$. 
 \textcolor{black}{Conditions of  type \ref{Cond: B1} and \ref{Cond: B2} also appear in  \citeauthor{2009rufi}.}
 }

\begin{lemma}\label{theorem: convergence: Ti: log-concave }
Suppose that $\fx$ and $\fy$ are log-concave densities satisfying Conditions \ref{Cond: B1} and \ref{Cond: B2}. Suppose, further, $\fx$ and $\fy$ are bounded away from $0$ on an open set containing $D_p$, and  $m/N\to\lambda$.
 Then it follows that
$|T_i(\Fmx,\Fmy)-T_i(\Flx,\Fly)|=o_p(1)$ for $i=1,2$.
\end{lemma}



 
 \begin{remark}
  All results of Section~\ref{section: asymptotic} hold if we replace $D_p$ by a compact set $D$ as long as there is a  $p>0$ so that $D\subset D_p$ and $D_p$ satisfies the conditions of the theorems stated in Section~\ref{section: asymptotic}.
 \end{remark}

\subsection{Simulations} \label{sec:simulation:1}

This section compares the performance of the shape-constrained tests designed in Section~\ref{sec: null of non dominance} with their nonparametric counterparts.  We let $m=n=100$, which is reflective of the sample sizes anticipated in many phase 1b or phase 2 vaccine trials --- for example, our motivating dataset has $m+n =248$. For all TSEP tests in this section, we use the critical value $z_{\alpha}$. See Appendix \ref{app: conservative test simulations}
 for the simulations with TSEP tests that  have  critical value $C_{m,n}z_{\alpha}$. For the TSEP and the minimum t-test, we also include the tests based on the  smoothed log-concave MLE. We will refer to the corresponding test as the smoothed log-concave test. Although we do not have any theoretical result for this test, we   use the asymptotic critical value $z_{\alpha}$. 


By the design of our hypotheses,  $\Hm_0$ encompasses a broad number of cases ranging from $G\succ F$, $F=G$ to cases where $\Fx$ and $\Fy$ touch or  cross each other. We develop  simulation schemes so that we can explore a wide range of scenarios. 
Our simulation schemes involves a parameter $\gamma$ varying over the range $[0,1]$. Here $\gamma$ quantifies the difference between the data generating distribution functions $\Fx_{\gamma}$  and $\Fy_{\gamma}$. We evaluate the power $\nu(\gamma)$ at a grid of equally spaced points in $[0,1]$. 


For our simulation study, we consider the following cases:
\begin{compactitem}
\item[(a)] $F_{\gamma}\sim N(\gamma,1)$, and $G_{\gamma}\sim N(0,1)$.
\item[(b)] $F_{\gamma}\sim N(3\gamma,1)$, and $G_{\gamma}\sim N(0.5,2)$.
\item[(c)]$F_{\gamma}\sim Gamma(2,0.1+0.4\gamma)$ and $G_{\gamma}\sim Gamma(1,0.5)$,  \textcolor{black}{where $Gamma(a,b)$ is a Gamma random variable with shape parameter $a$ and scale parameter $b$.}
\item[(d)] $F_{\gamma}\sim Gamma(2,1)$  and $G_{\gamma}\sim Pareto(0.5+ 2\gamma,1)$. Here $Pareto(a,b)$ is the Pareto distribution function with shape parameter $a$ and scale parameter $b$. 
\item[(e)] $F_{\gamma}\sim N(0,1)$ and $G_{\gamma}\sim N(2\gamma+ 4,1)/2+N(2\gamma-2,1)/2$.
\end{compactitem}

\begin{figure}[h]
\includegraphics[height=1.5 in, width=\textwidth]{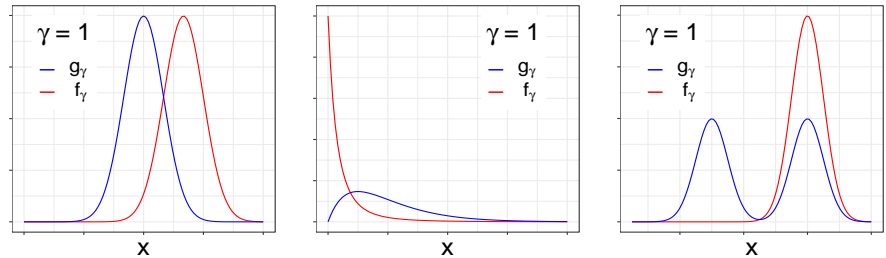}
\caption{Plots of the densities $f_{\gamma}$ and $g_{\gamma}$ corresponding to  the cases (a), (d), and (e) for $\gamma=1$. 
}\label{Figure: 3 density: S2}
\end{figure}
\textcolor{black}{
Case (a) corresponds to  the traditional setting of null of equity against a shift alternative.    The last four cases cover $\Hm_0$ combinations that include crossing. In those cases, $\Fx_{\gamma}$ and $\Fy_{\gamma}$ cross each other  on $\iint(D_p)$ at $\gamma=0$.  As $\gamma$ increases, however, $\Fx_{\gamma}$ and $\Fy_{\gamma}$ eventually touch (cases c, d, and e), or cross (case b) each other at $\text{bd}(D_p)$,  generating a LFC configuration. We denote the corresponding $\gamma$  by $\gamma^*$. Finally, at $\gamma=1$,  $\Fx_{\gamma}$  strictly dominates $\Fy_{\gamma}$ in the sense of $\Hm_1$. Figures~\ref{Figure: cross dist: a}, \ref{Figure: cross dist: d}, \ref{Figure: cross dist: e}, and \ref{Figure: cross dist: b}  in Appendix~\ref{app: additional tables and figures}  display the plots of $\Fx_{\gamma}$ and $\Fy_{\gamma}$ for several values of $\gamma$ in cases (b), (c), (d) and (e), respectively.  }
\textcolor{black}{
 Figure~\ref{Figure: 3 density: S2}  illustrates the densities in cases (a), (d), and (e) for $\gamma=1$.
 Cases (d) and (e) are chosen to reflect violations of the shape constraints that we consider. The Pareto density in (d) violates the log-concavity assumption and the normal mixture in (e) violates  both the unimodality  and the log-concavity assumptions. All other densities satisfy both shape constraints. }
 
%



%
We evaluate the properties of the tests under consideration using 10,000 Monte Carlo replicates. 
\textcolor{black}{We set the $p$ in $D_{p,m,n}$ to be $0.05$, where $D_{p,m,n}$ was defined to be the set $[\mathbb H_N^{-1}(p),\mathbb H_N^{-1}(1-p)]$.} Also, the level of significance is $0.05$ for all our tests. 
The tests based on the empirical cumulative distribution function will be referred to as NP (nonparametric) tests. For brevity, we will refer to the nonparametric, unimodal, log-concave and smoothed log-concave tests by NP, UM,  LC, and smoothed LC tests, respectively. 


{\color{black}
Figure~\ref{fig: power} displays the power curves for the minimum t-test and the TSEP test. 
In terms of power, the LC and smootheed LC tests generally outperform the UM tests, which generally outperform the NP tests. In their simulation study, \cite{davidson2013} reported  the NP minimum t-test to be conservative, which aligns with our observation. Also, the shape-constrained minimum t-tests have slightly higher power than the shape-constrained TSEP tests, although the difference is not always  significant. For the NP tests, the rejection rates of the minimum t-tests and the TSEP tests are almost identical. 
 
 Except for in case (d), where $G_{\gamma}$ is Pareto, all tests control the type I error in all null set-ups, including the LFC configuration where $\gamma=\gamma^*$. Since LFC configurations constitute the boundary of $\Hm_0$, this observation indicates that our tests have size $0.05$ for all cases except case (d). Although Pareto density violates only the log-concavity assumption, apparently no test has the correct size, albeit the NP and UM tests performing the best in terms of size. \textcolor{black}{Also, for case (d), the smoothed LC test  has lower type I error than LC test.} Surprisingly, in case (e), where both log-concavity and unimodality are violated, the shape-constrained tests have overall better performance than the NP tests although all tests exhibit poor power in this case.

 All the WRS type tests, including the LC WRS test (whose asymptotic behavior is yet unknown),  exhibit a much larger size than 0.05 in all cases except the null of equity type case (a), which is the ideal scenario for WRS type tests. Figure~\ref{Figure: WRS}  illustrates the power curve $\nu(\gamma)$ in cases (b) and (c), where all three WRS tests exhibit very high type I error at several null configurations. Our findings are consistent with the observation in \cite{LW2012} that the WRS test can be  misleading for testing the null of non-dominance.

 \textcolor{black}{ 
 In summary,   under a correctly specified model, the shape-constrained minimum t-test and TSEP test outperform their nonparametric counterparts, with the log-concave minimum t-test having the best power. When the shape constraints are violated, the nonparametric tests are not distinguishably better than the unimodal tests. On the other hand, \textcolor{black}{The WRS tests do not control type I error for  null configurations with crossing, as expected. Another important takeaway from this section is that  with  the critical value $z_{\alpha}$, the smoothed log-concave tests seem to perform as well as the log-concave tests  under log-concavity. }
 }

  \textcolor{black}{
  The current article uses asymptotic critical values for performing the above-mentioned tests, but bootstrap critical values could also be an option. Bootstrap  for the shape-constrained tests, however, is not straightforward. Generation of observations from shape-constrained LFC configuration poses some challenges, which involves solving non-trivial optimization problems. Further discussion in this direction is out of the scope of the present paper. Therefore, we leave  bootstrap  tests for future research.}
  
  }

%
%
%

\subsection{Application to HVTN 097 and HVTN 100 data}
\label{sec: application to our data}
 


{\color{black}
Our first task is to select a $p$ for choosing $D_{p,m,n}$. Figure~\ref{fig: preliminary plots}  and some inspection show that  the empirical distribution functions are very close on the sets $(-\infty,\Hm^{-1}(0.072)]$ and  $[\Hm^{-1}(0.975),\infty)$ in that they either cross or touch each other on these regions.  Therefore, this is the problematic region we wish to exclude  from our $D_{p,m,n}$, because clearly there is not enough evidence of any dominance in this region.
    Therefore, we set the $p$ in $D_{p,m,n}$ to be $0.075$. \textcolor{black}{We remark that it is ideal to choose $p$ in a systematic way without looking at the data. However,  constructing a rigorous procedure for choosing $p$ is out of the scope of the present paper, and we leave it for future research.} Table~\ref{table: p-values}, which tabulates the p-values of the tests, displays  that all  tests  reject the null at the level of significance $0.05$. The highest p-value is observed for the  NP TSEP test, which is approximately $0.042$.}

  \begin{table*}[h]
   \centering
  \begin{tabular}{lcccc}
\toprule
Tests & Nonparametric & Unimodal & Log-concave & Smoothed log-concave\\ 
\midrule
Minimum t-test & 0.015 & 0.007 & 0.007 & 0.001 \\
TSEP & 0.042 & 0.037 & 0.035 & 0.030\\
\bottomrule
\end{tabular} 
\caption{Table of the p-values of different tests applied on our data. Here   the $p$ in $D_{p,m,n}$ is set to be $0.075$. The critical value was $z_\alpha$ for all the tests. }
\label{table: p-values}
  \end{table*}   

   It is natural to ask if we can estimate the power of our tests
   at $(\Fn,\Fh)$.
   Because  $(\Fn,\Fh)$ is unavailable, we analyze the power in a neighborhood of $(\hFns,\hFhs)$ instead, where  $\hFns$ and $\hFhs$ correspond to the distributions of the smoothed log-concave MLE \citep{smoothed} estimators of $f_{100}$ and $f_{097}$, respectively. 

Let us denote the smoothed log-concave MLE of the pooled sample by $\tilde{f}_{m,n}^{0}$. 
    Letting  $\tilde{F}_{m,n}^0$ denote the corresponding distribution function, we consider the mixture distributions
\begin{equation}\label{def: power curve}
\hFns(\gamma)=(1-\gamma) \tilde{F}_{m,n}^0+\gamma\hFns,\quad\text{ and }\quad\hFhs(\gamma)=(1-\gamma) \tilde{F}_{m,n}^0+\gamma\hFhs,
\end{equation}
where $\gamma\in[0,1]$. Note that, similar to  case (a) in Section \ref{sec:simulation:1}, here also $\gamma$ quantifies the departure of the configuration $(\hFns(\gamma),\hFhs(\gamma))$ from the null of equality of distributions. Also, the distance between $\hFns(\gamma)$ and $\hFhs(\gamma)$ increases as $\gamma$ approaches $1$.   
We denote the  densities of $\hFns(\gamma)$ and $\hFhs(\gamma)$ by $\hfns(\gamma)$ and $\hfhs(\gamma)$, respectively.
%

   Now observe that when $\gamma\in(0,1)$,  the mixture densities may not be log-concave or even unimodal. Hence, we compute the log-concave projections \citep{dumbreg}  of $\hfns(\gamma)$ and $\hfhs(\gamma)$, respectively. Log-concave projection of a density $f$ is the log-concave density   closest to  $f$ in Kullback-Leibler (KL) distance. Since the log-concave projection of any arbritrary density is not directly computable, we adopt a two step approach to approximate the log-concave projections.  
  In the first step, we simulate $1000$ observations from each of $\hfns(\gamma)$ and $\hfhs(\gamma)$. In the second step, we 
  calculate the smoothed log-concave MLE density estimators of \citeauthor{smoothed} based on the simulated samples in the last step.  The resulting densities are  the approximate smoothed log-concave projections of $\hfns(\gamma)$ and $\hfhs(\gamma)$. 
  Finally, we generate two samples of size $68$ and $180$ from the projected densities using the methods in  \cite{lcsoft} and the R package \textit{logcondens}, and replicate this process 10,000 times. 


   
 Figure~\ref{fig:power: ourdata} entails that all the tests exhibit decent power for higher values of $\gamma$. The LC minimum t-test  exhibits the highest power, which is unsurprising since the underlying data is generated from log-concave densities. 
 Also, since $\hFns(\gamma)=\hFhs(\gamma)$ at $\gamma=0$, the power curves resemble that of case (a) in our simulation schemes. 
\section{Measures of discrepancy}
 \label{sec: measure of discrep}
 

This section describes an approach for quantifying the difference between $\fx$ and $\fy$ via estimates of the squared Hellinger distance $\hd^2(\fx,\fy)$, which provides complementary insights to the tests of stochastic dominance presented in Section \ref{sec: tests}. 
 \textcolor{black}{The Hellinger distance is an example of $f$-divergence. These divergences are widely used to measure the similarity or dissimilarity between two probability measures. Compared to other commonly used $f$-divergences such as the KL divergence or the $\chi^2$ divergence \citep{chi2}, the Hellinger distance is appealing due to its symmetry in its arguments, which is a desirable property for a measure of discrepancy. Another crucial advantage of the Hellinger distance  is that its value is finite for every pair of densities \citep{gibbs2002}. In contrast, the KL and $\chi^2$ divergences can both be infinite if the densities under consideration do not share the same support. It is not always reasonable to assume that the underlying densities of responses collected from different vaccine trials will have the same support. Therefore,  the Hellinger distance appeals to us more than  the KL divergence or the $\chi^2$ divergence. The Hellinger distance has also seen successful application in various disciplines ranging from machine learning \citep{cieslak2009, gonzalez, gonzalez2010}, to ecology  \citep{rao1995}, to fraud detection \citep{yamanishi2004}.  Finally, we choose to work with the squared Hellinger distance instead of the Hellinger distance because due to its simpler form, the squared version easily lends itself to efficient estimation procedures. Regardless, a 95\% confidence interval for the Hellinger distance can always be constructed using that of  its squared version.
 }

{\color{black}
  We estimate the squared Hellinger distance between $\fx$ and $\fy$ using the same density estimators  involved in the construction of tests of stochastic dominance, that is, the log-concave MLE \citep{2009rufi} and its smooth version \citep{smoothed}, or  \birge's estimator. Recalling the definition of the density estimators  $\fnx$, $\fny$, $\flx$, $\fly$, $\flxs$, and $\flys$ from Section \ref{sec: null of non dominance}, we propose the plug-in estimators   $\hd^2(\fnx,\fny)$, $\hd^2(\flx,\fly)$, and $\hd^2(\flxs,\flys)$ for the purpose of estimating $\hd^2(\fx,\fy)$.  We refer to the resulting estimators as the ``unimodal", ``log-concave", and the ``smoothed log-concave"   estimator, respectively. }


\subsection{Asymptotic properties of the unimodal estimator:}
\label{subsec: measure of discrep: unimodal }
 We will show that, under some regularity conditions,  $\hd^2(\fnx,\fny)$ is a $\sqrt{N}$-consistent estimator of $\hd^2(\fx,\fy)$ with asymptotic variance
 \begin{equation}\label{def: sigma f g}
    \sigma_{f,g}^2=\frac{2\hd^2(f,g)-\hd^4(f,g)}{4\lambda(1-\lambda)}.
\end{equation}
Letting $b,B>0$, we denote by $\mP(b,B)$ the class of densities that are bounded below and above by $b$ and $B$ on their support, that is,
  \begin{equation}\label{def: bounded density class}
  \mP(b,B)=\bigg\{f\in\mP\ :\ b\leq f(x)\leq B,\ \text{ for }x\in\supp(f)\bigg\}.
  \end{equation}
  \textcolor{black}{
We will assume that $\fx,\fy\in\mP(b,B)$. Simulations suggest that this  condition may not be necessary. However, this  condition is required for  technical reasons in our proof. Such technical condition is quite common in the literature, and has appeared in the analysis of plug-in  estimators \citep{robin2015} and functionals of Grenander estimators \citep{ ramu2018}.} 

 \begin{theorem}\label{thm: Hellinger distance}
 Suppose  $\fx$ and $\fy$ are unimodal densities in $\mP(b,B)$, where  $b$, $B>0$. Further suppose that $\fx$ and $\fy$ satisfy condition~\ref{Cond A} and  $m/N\to\lambda$. 
 Then,
 \begin{equation}\label{intheorem: statement: hellinger}
     \sqrt{N}[\hd^2(\fnx,\fny)-\hd^2(f,g)]\to_d N(0,\s^2_{f,g}),
 \end{equation}
 where $\s^2_{f,g}$ is as defined in \eqref{def: sigma f g}.
 \end{theorem}
Note that, since $\sigma_{f,g}^2$ can be consistently estimated plugging in the estimator $\hd^2(\fnx,\fny)$, a Wald type confidence interval is readily available for $\hd^2(f,g)$.
Also, the asymptotic variance $\s^2_{f,g}$ equals the lower bound of the asymptotic variance on  a regular estimator of the squared Hellinger distance under the nonparametric model \citep{robin2015}. See \cite{vdv, birge1995} for more detail on the lower bound and related theory.

\textcolor{black}{
We now give a high-level explanation of the idea behind   Theorem~\ref{thm: Hellinger distance}. It can be shown that the squared Hellinger distance is a smooth  functional of the underlying distribution functions in some suitable sense. Specifically,  we will show that it allows a first order Von Mises expansion  \citep{vonmises1}, which has the same essence as the Taylor series expansion \citep{robin2015}. On the other hand,  for $\Fnx^0$,  the unimodal MLE  of  $\Fx$  based on the true mode, we show that $\sqrt{m}(\Fnx^0-\Fx)$ converges weakly to a Brownian process almost surely under Condition~\ref{Cond A}.   It can then be shown via a delta-method type argument, applied on the squared Hellinger distance functional, that \eqref{intheorem: statement: hellinger} holds for the unimodal MLEs of $\Fx$ and $\Fy$. The final step in proving Theorem~\ref{thm: Hellinger distance} is  showing that the squared Hellinger distance between  \birge's estimator and  the MLE is small.}
 

  
  \subsection{Asymptotic properties of the log-concave estimators:}
     Recall that our log-concave estimators of the squared Hellinger distance are given by $\hd^2(\flx,\fly)$ and $\hd^2(\flxs,\flys)$, where we remind the reader that $\flx$, $\fly$ are the log-concave MLEs and $\flxs$, $\flys$ are the smoothed log-concave MLEs. Lemma~\ref{lemma: measure of discrep: log-concave} implies that
      these estimators are strongly consistent for $\hd^2(\fx,\fy)$ provided the shape constraint holds.
  \begin{lemma}\label{lemma: measure of discrep: log-concave}
Suppose that the densities $\fx$ and $\fy$ are log-concave and continuous.  Then, as $m,n\to\infty$, $\hd^2(\flx,\fly)\as \hd^2(\fx,\fy)$ and $\hd^2(\flxs,\flys)\as \hd^2(\fx,\fy)$.
  \end{lemma}

 \textcolor{black}{ Our simulations  indicate that both the log-concave and the smoothed log-concave plug-in estimators are $\sqrt N$-consistent with asymptotic variance $\sigma^2_{f,g}$ when $f$ and $g$ are continuous log-concave densities. Therefore, in  our empirical study,  we include Wald type confidence intervals based on  these  log-concave plug-in estimators as well. Our simulations in Section \ref{sec: simulation : measure of discrep}  indicate that under the violation of the continuity assumption, the  log-concave plug-in estimator may still remain $\sqrt{N}$-consistent, but the $\sqrt{N}$-consistency of the 
   the smoothed log-concave plug-in estimator may  fail to hold. }
  
  \textcolor{black}{
It may be possible to analyze the $\sqrt N$-consistency of the log-concave estimators working along the lines of \cite{Kulikov2006,groeneboom1984,GROENEBOOMP1989}, which pertain to the shape restriction of monotonicity. However, we leave this investigation for future research because a detailed treatment of the $\sqrt N$-consistency of the log-concave estimators is out of scope of the present paper. We remark in passing that  proving the  $\sqrt N$-consistency of the log-concave plug-in estimator may be easier  for some special cases, e.g., when the logarithm of $f$ and $g$ are linear or piece-wise affine, by using the results of \cite{kim2018}. We do not pursue this direction because  it is unlikely that, for our data, $f$ and $g$ belong to such restricted classes.
 }

\begin{remark}\label{remark: model misspecification: log-concave}
 The case of model misspecification is of natural interest in the study of shape-constrained estimators. Suppose  $\fx$ and $\fy$ are not log-concave, but they are bounded continuous densities with finite first moments.
 Then it follows that there exist unique log-concave densities $\fx^*$ and $\fy^*$, which are almost sure limits of $\flx$ and $\fly$ in both uniform and $L_1$ metric \citep[Theorem 4][]{theory}. Here $\fx^*$ and $\fy^*$ are also the log-concave projections of $\fx$ and $\fy$, respectively, in the sense of \cite{dumbreg}. Using the same arguments as in the proof of Lemma~\ref{lemma: measure of discrep: log-concave}, it can be shown that the log-concave estimator of the squared Hellinger distance converges almost surely to $\hd^2(\fx^*,\fy^*)$.
 If $\fx$ and $\fy$ additionally have finite second moments, a similar phenomenon takes place for the smoothed log-concave MLEs as well. In this case, however,  $\flxs$ and $\flys$  converge uniformly (and also in $L_1$) to  different limits $\fx^{**}$ and $\fy^{**}$, which can be interprated as the respective smoothed versions of the log-concave projections $\fx^*$ and $\fy^*$  \citep[cf. Theorem 1,][]{smoothed}. In this case also, the smoothed log-concave estimator of the squared  Hellinger distance converges almost surely to $\hd^2(\fx^{**},\fy^{**})$.
In summary, if the log-concavity assumption is violated, the log-concave plug-in estimators  converge to a different limit, whose distance from $\hd^2(\fx,\fy)$ depends on the departure of $\fx$ and $\fy$ from log-concavity.
\end{remark}

  \textbf{KDE based plug-in estimators:}{
  The natural non-parametric comparators of the shape constrained plug-in estimators are the KDE based plug-in estimators.
   Though the latter is simple to implement, it can have a bias of order $||\fnxk-\fx||_2+||\fnyk-\fy||_2$ \citep[cf. Section 2 of][]{robins2009}, where  $\fnxk$ and $\fnyk$ are the KDEs of \fx\ and \fy, respectively. The above bias decreases to zero at a rate slower than $N^{-1/2}$  \citep{kderate}, thereby leading to a suboptimal performance. See Section 5 of \cite{robin2015} for more discussion on the disadvantages of the KDE based na\"{i}ve plug-in estimators. }
  
One can improve the  na\"{i}ve plug-in estimator, however, using a one-step Newton-Raphson procedure \citep[cf.][]{vdv,pfanzagl2013}, which leads to a bias corrected plug-in estimator. 
In context of the Hellinger distance, the bias-corrected estimator $ (\hd^2)^*(\fnxk,\fnyk)$ takes the form \citep{robin2015}
  \begin{align*}
   & \hd^2(\fnxk,\fnyk) +\rint \ps_f(x;\fnxk,\fnyk)d\Fmx(x)+\rint \ps_g(y;\fnxk,{\fnyk\vphantom{\fnxk}})d\Fmy(y) \\
 &= 1-\dfrac{1}{2}\lb\rint \sqrt{\fnyk(x)/\fnxk(x)}d\Fmx(x)+\rint \sqrt{\fnxk(y)/\fnyk(y)}d\Fmy(y)\rb,
\end{align*}  
where $\psi_f$ and $\psi_g$ are the influence functions corresponding to the functional $(f,g)\mapsto \hd^2(f,g)$. We will formally introduce the influence functions in Appendix~\ref{sec: appendix: measure of discrep}. 
Note that learning the form of the bias-corrected estimator thus requires the explicit computation of the influence functions $\psi_f$ and $\psi_g$. From a broader prospective, each time one tries to estimate a functional of the underlying distributions using a bias-corrected plug-in estimator,  they have to carry out some extra analytical calculations that depend on the functional of interest.

  In contrast, Theorem~\ref{thm: Hellinger distance} shows that when the shape constraint is satisfied, the unimodal estimator of the squared  Hellinger distance does not require any bias correction for $\sqrt N$-consistency. Our analysis in Appendix~\ref{sec: appendix: measure of discrep} indicates that this is not an artifact of the Hellinger distance, but rather  the result of  \birge's estimator's proximity to the unimodal MLE, i.e. the Grenander estimator based on the true mode. 
 \textcolor{black}{In fact, Theorem \ref{theorem: divergence: general} in Appendix~\ref{sec: appendix: measure of discrep} shows that, if  $T$ is a smooth functional of the distribution functions $\Fx$, then  $T(\fnxm)$ is $\sqrt m$-consistent for $T(\Fx)$ under mild conditions. Because the total variation distance between $\fnx$ and $\fnxm$ is $o_p(m^{-1/2})$ (see our Lemma~\ref{lemma: weak convergence of dist of unimodal density functions} in Appendix~\ref{sec: appendix: tests}), it is then natural to expect that  $T(\fnx)$ would be $\sqrt m$-consistent  if $T$ is sufficiently smooth. }

Although we do not have any such theoretical evidence for  the log-concave estimators,  our simulations and the simulations of \cite{cule2010} on  plug-in estimators based on log-concave MLEs (see Figure 18 therein)  indicate that these estimators do not have any $O_p(N^{-1/2})$ bias term either.
This allows users of correctly specified shape-constrained estimators to avoid analytic calculation of the influence functions entirely. 

  In our upcoming simulation study, we use both $\hd^2(\fnxk,\fnyk)$ and  $(\hd^2)^*(\fnxk,\fnyk)$ as comparators, where the KDEs are based on the Gaussian kernel. We refer to these estimators as the KDE estimator and the bias-corrected KDE estimator, respectively. As in Section \ref{sec: density estimation}, the  kernel bandwidth is chosen using the univariate least square cross-validation (LSCV) selector of \cite{bowman} and \cite{rudemo}. \textcolor{black}{ Generally, kernel-based bias corrected estimators satisfy $\sqrt N$ consistency results of the type \eqref{intheorem: statement: hellinger} \citep[cf. Theorem 6,][]{robin2015}. Therefore, we use the bias-corrected KDE based confidence intervals to benchmark the performance of our shape constrained estimators. However, we do not report any confidence interval based on the na\"{i}ve plug-in  estimator because our simulations indicate that usually its coverage is much less than the nominal level. For the sake of clarity, in Table \ref{table: summarize plug-in}, we summarize the current state of results on the above-mentioned estimators of the squared Hellinger distance.
  }
   \begin{table*}[h]
   \centering
  \begin{tabular}{lcccc}
\toprule
Estimator & Consistency & $\sqrt N$-consistency & Showed in & \makecell[vl]{Confidence\\ interval$^*$}\\ 
\midrule
Na\"{i}ve KDE & yes & no & \makecell{ \cite{robin2015}} & no\\
\makecell{Bias\\ corrected\\ KDE} & yes & yes & \cite{robin2015} & yes\\
Unimodal & yes & yes & current paper & yes\\
Log-concave & yes & \makecell{unknown, but\\suggested by\\  simulations} & current paper & yes\\
\makecell{Smoothed \\log-concave} & yes & $---''---$ & current paper & yes\\
\bottomrule
\end{tabular} 
\caption{Table summarizing the current  results on different plug-in estimators of the squared Hellinger distance.  The consistency and $\sqrt N$-consistency of the shape constrained estimators are proved under the correct specification of the shape constraints. Empirical studies are provided on all the estimators listed above.\\
* Indicates whether we included confidence intervals for the corresponding estimator in the empirical study of Section \ref{sec: simulation : measure of discrep}.}
\label{table: summarize plug-in}
  \end{table*}          
            
\subsection{Simulations:}
\label{sec: simulation : measure of discrep}

\begin{figure}[h]
\includegraphics[ width=\textwidth]{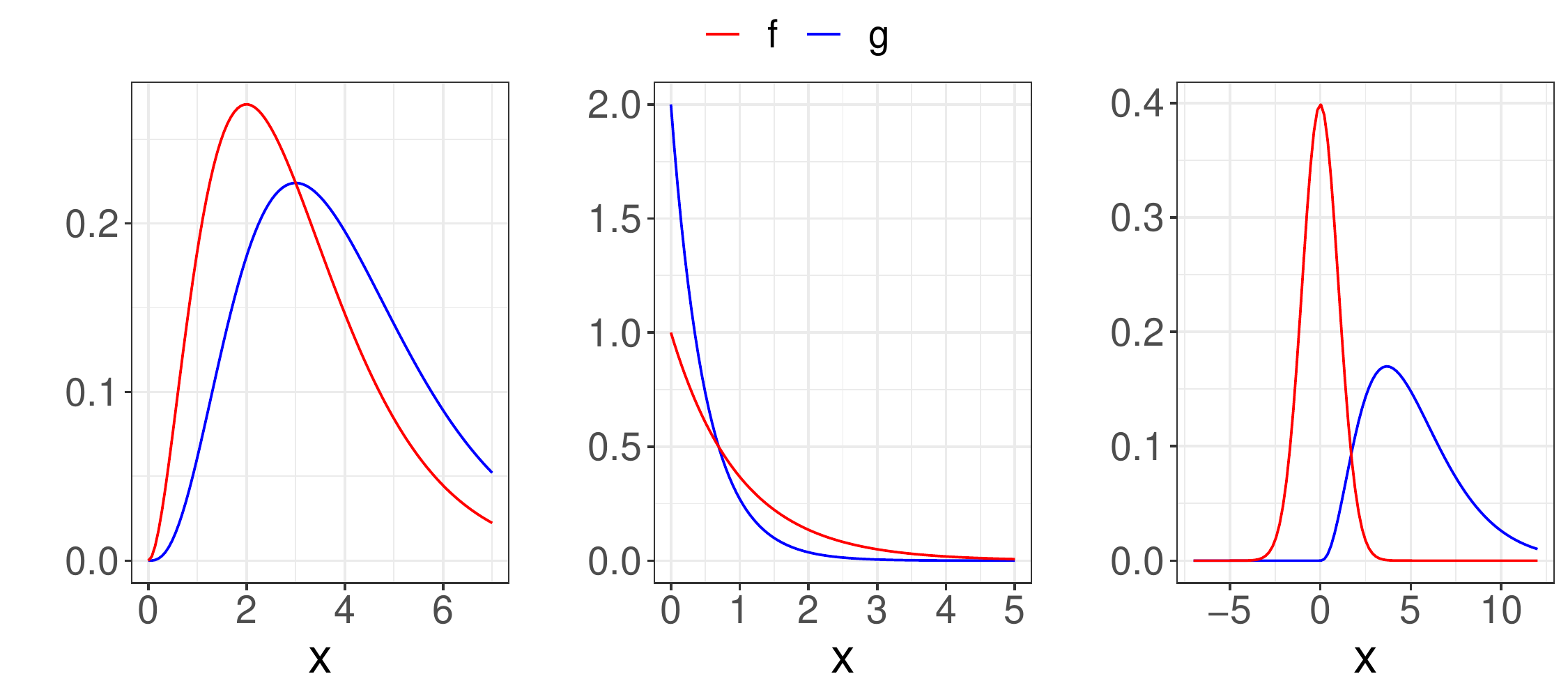}
\caption{Plots of the densities $f$ and $g$ for the cases (a) (left), (e) (middle), and (f) (right)  in the simulation scheme of Section~\ref{sec: simulation : measure of discrep} }
\label{Figure: density: s3}
\end{figure}

\begin{figure}[h]
\includegraphics[width=\textwidth]{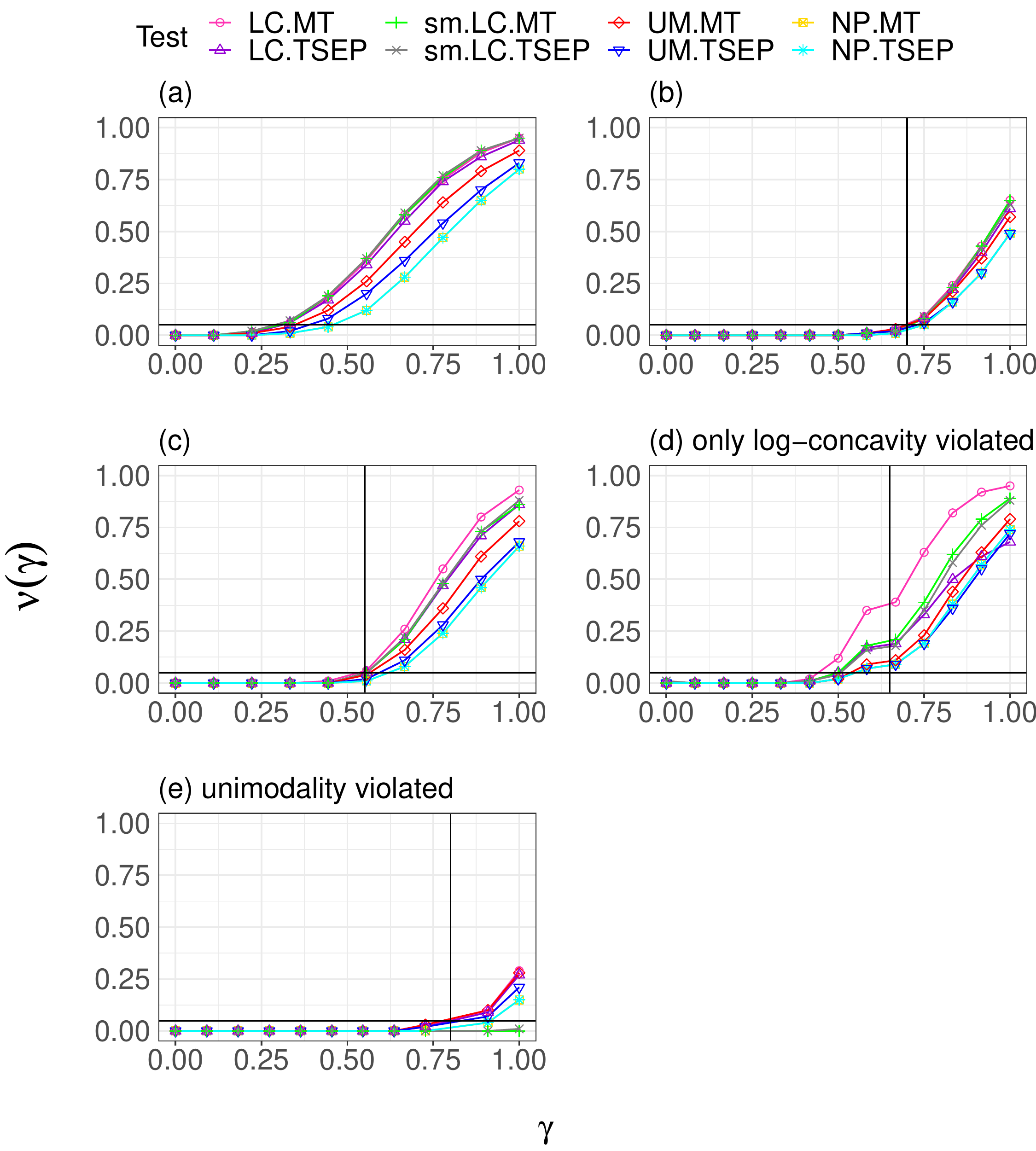}
 \caption{ Plot of estimated power $\nu(\gamma)$ vs $\gamma$: the labels (a)-(e) correspond to simulation schemes (a)-(e). Here MT and TSEP correspond to the minimum t-test and the TSEP test, respectively. The standard deviation of the $\nu(\gamma)$ estimate in each case is less than $0.005$.
  The black  horizontal line corresponds to the level of the test, $\alpha=0.05$. For cases (b)-(e), the black vertical line   represents the LFC configuration $\gamma^*$, taking value 0.70 (b), 0.55 (c), 0.65 (d), and 0.80 (e). }
 \label{fig: power}
\end{figure}

\begin{figure}[h]
\includegraphics[ width=\textwidth]{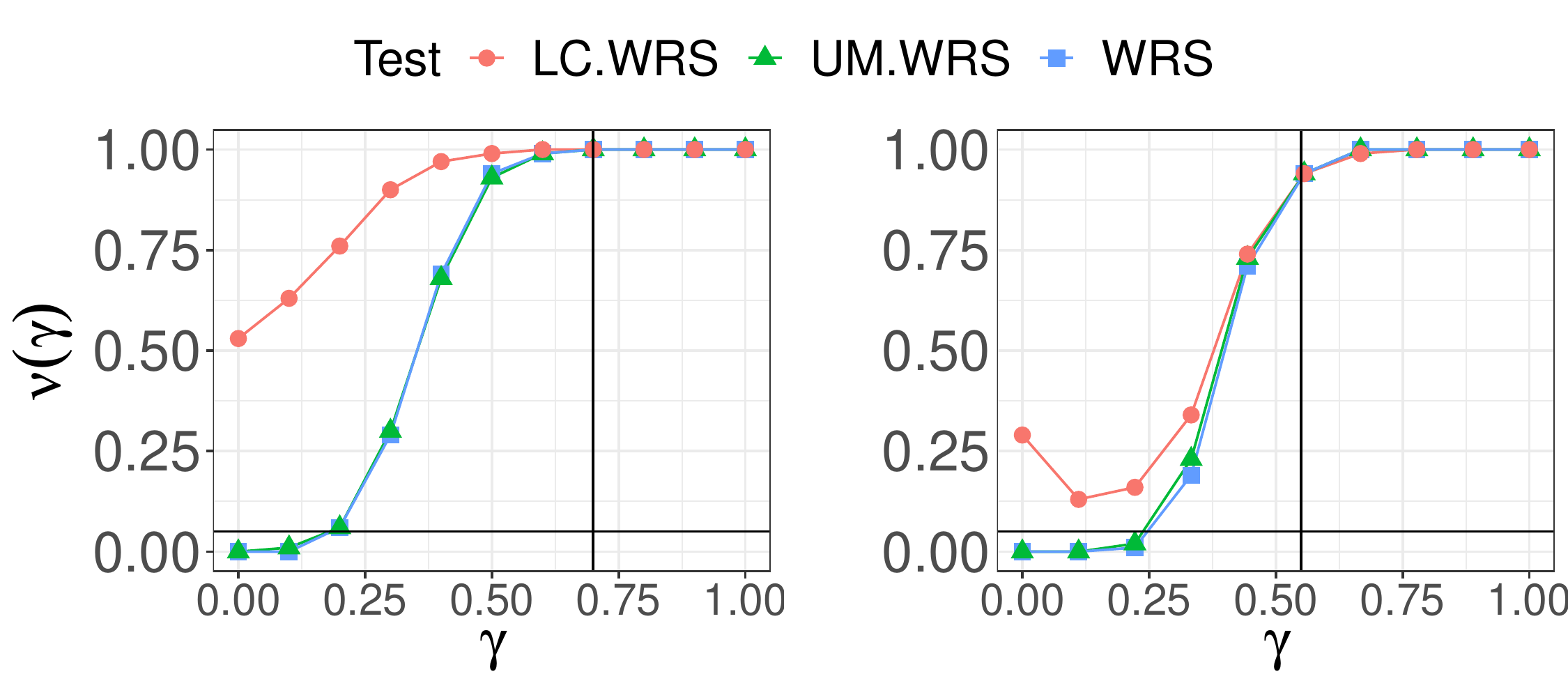}
\caption{Estimated power curves of WRS type tests for case (b) (left) and case (c) (right).  Here LC.WRS and UM.WRS correspond to the tests based on $T_{m,n}^{\text{wrs}}(\Flx,\Fly)$ and $T_{m,n}^{\text{wrs}}(\Fnx,\Fny)$, respectively.
 The  standard errors of the estimated powers are less than 0.005. The black  horizontal line corresponds to the level of the test, $\alpha=0.05$. The black vertical lines represent the LFC configuration $\gamma^*$, which are 0.70 (b) and 0.55 (c).}\label{Figure: WRS}
\end{figure}
 
  \begin{figure}[h]
\centering
    {\includegraphics[width=\textwidth]{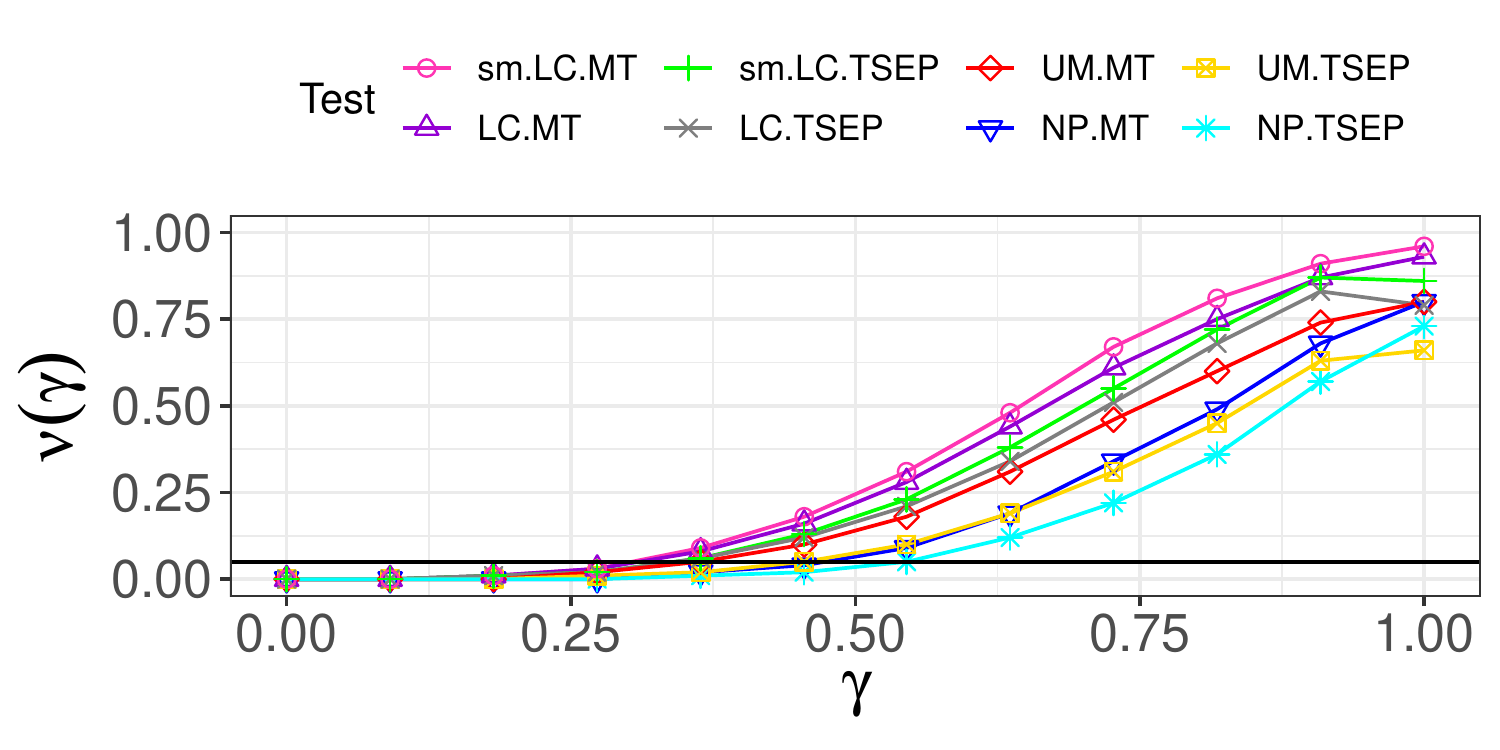}
    }
    \caption{Estimated power curve applied on our datasets: Here MT and TSEP correspond to the minimum t-test and the TSEP test, respectively. The black  horizontal line corresponds to the level of the test, which is $0.05$. The standard deviation of the estimated power does not exceed $0.005$ in any case.}
\label{fig:power: ourdata}
    \end{figure}

To compare the performance of different estimators of the squared Hellinger distance, we consider the following  combinations of $f$ and $g$:
\begin{compactitem}
\item[(a)] $f\sim Gamma(4,1)$ and $\fy\sim Gamma(3,1)$.
\item[(b)] $f\sim N(1,1)$ and $\fy\sim N(0,1)$, which corresponds to case (a) in Section \ref{sec:simulation:1} with $\gamma=1$.
\item[(c)] $f=\hfns$ and $g=\hfhs$, where $\hfns$ and $\hfhs$ are the smoothed log-concave MLEs of $\fn$ and $\fh$, respectively (see Figure~\ref{fig:log-concave densities}).
\item[(d)] $f\sim\text{Exp}(1)$ and $g\sim\text{Exp}(2)$ where $\text{Exp}(t)$ is the exponential distribution with rate $t$.
\item[(e)] $(f,g)$ corresponds to case (e) in Section \ref{sec:simulation:1} with $\gamma=1$.  
\item[(f)] $f\sim N(0,1)$, and $g\sim Gamma(3.61,1.41)$.
\end{compactitem}
The plots of the above schemes can be found in Figure~\ref{Figure: 3 density: S2}  and Figure~\ref{Figure: density: s3}.
. In cases (a), (b), (c), (d), and (f),   both shape constraints are satisfied,  where in case (e), both shape constraints are violated. \textcolor{black}{We will refer to case (a), (b), (c), (d), and (f), therefore, as the ``correctly specified'' cases and case (e),  as the ``misspecified" case.} 
 In case (d), the logarithm of the densities are linear on their respective supports. This is also the only case where the densities are discontinuous (see Figure~\ref{Figure: density: s3}; discontinuity at zero).
In case (f), $g$ was chosen so as to resemble the density estimate of  $\fh$ (see  Figure~\ref{fig:log-concave densities} and Figure~\ref{Figure: density: s3}). This is the only correctly specified case where the densities are known to come from different families, and the densities also have very different shapes.

We generate two samples of same size from each simulation setting.  We vary the common sample size $n$ from $50$ to $500$ in increments of $50$. We do not consider larger values of $n$ because in our motivating phase 1b and phase 2 vaccine trial applications,  the sample sizes are generally no larger than 500. We consider 10,000 Monte Carlo replications for each sample size.

\begin{figure}[h]\centering
\centering
\includegraphics[width=0.95\textwidth]{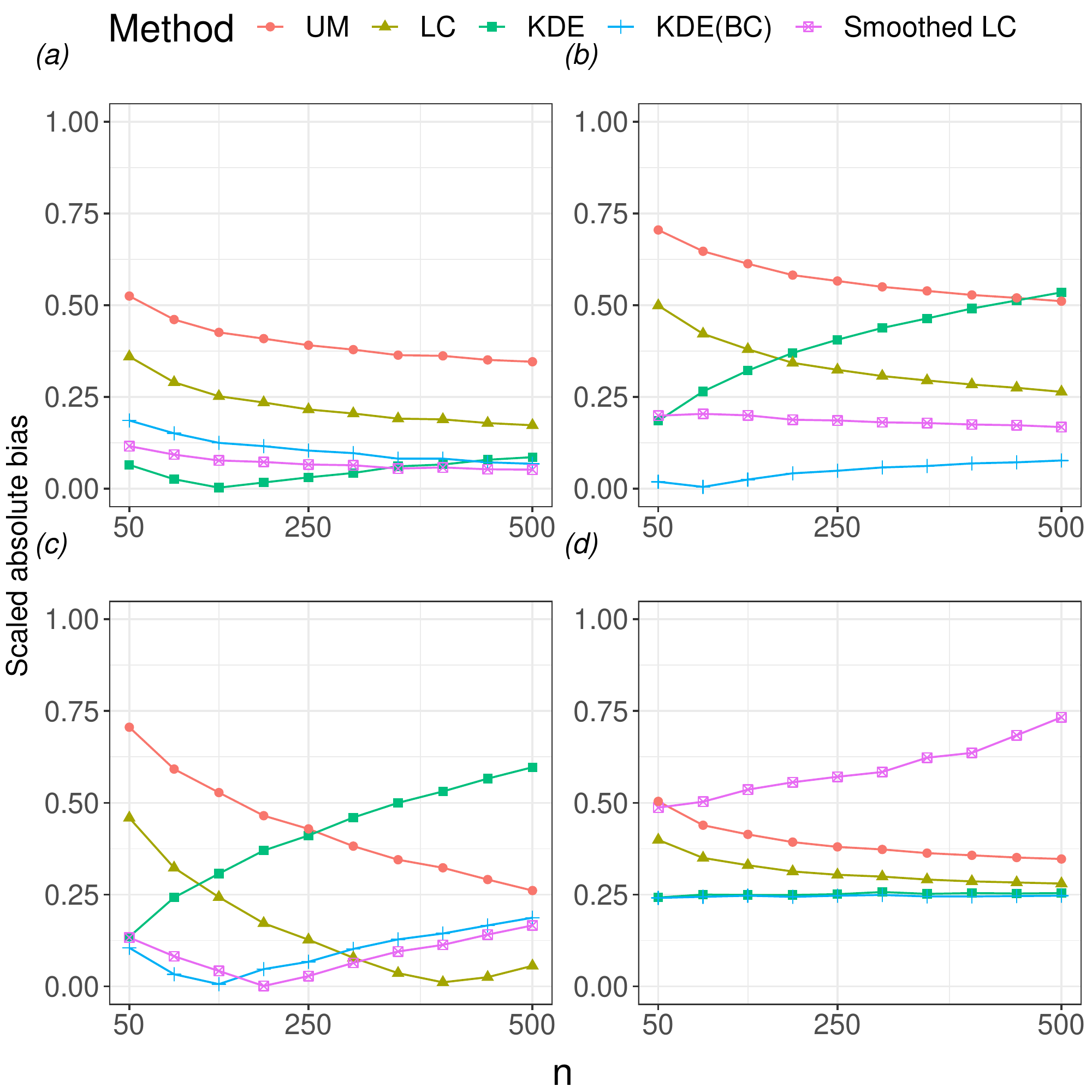}
\caption{Plots of the absolute values of the bias (scaled by $\sqn$) for  cases (a)--(d).  Here ``KDE (BC)" and ``smoothed LC" stand for the bias corrected KDE estimator and the smoothed log-concave estimator, respectively.
For each estimator, either the standard error is less than 0.006 or the relative standard error is less than 2\%.  }
\label{Figure: bias 1}
\end{figure}
\begin{figure}
\centering
\includegraphics[width=0.95\textwidth]{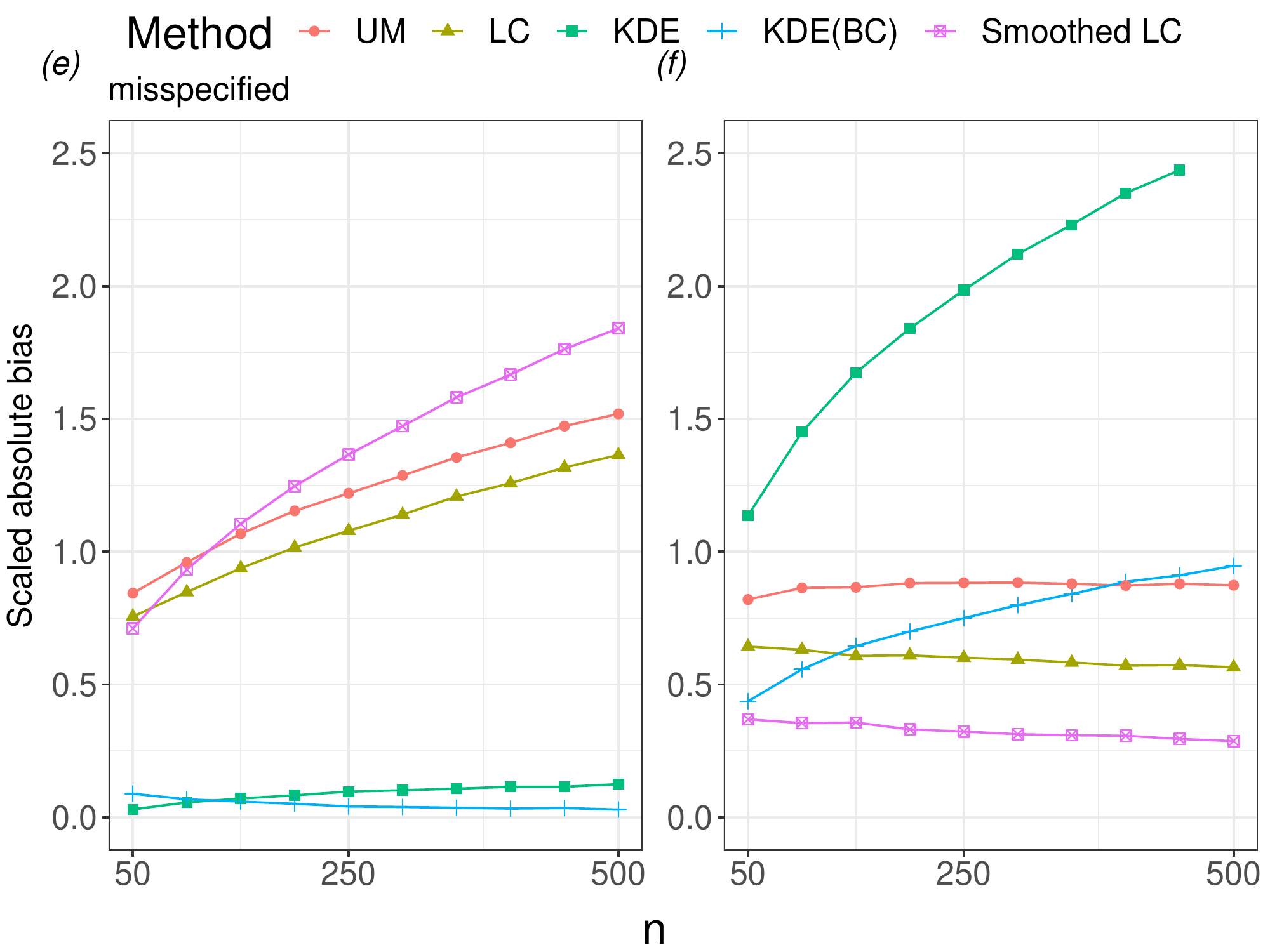}
\caption{Plots of the bias (scaled by $n$) for cases (e) and (f).  Here ``KDE (BC)" and ``smoothed LC" stand for the bias corrected KDE estimator and the smoothed log-concave estimator, respectively.
For each estimator, either the standard error is less than 0.006 or the relative standard error is less than 2\%. 
}\label{Figure: bias 2}
\end{figure}

\begin{figure}[h]\centering
\begin{subfigure}{\textwidth}\centering
\includegraphics[width=\textwidth]{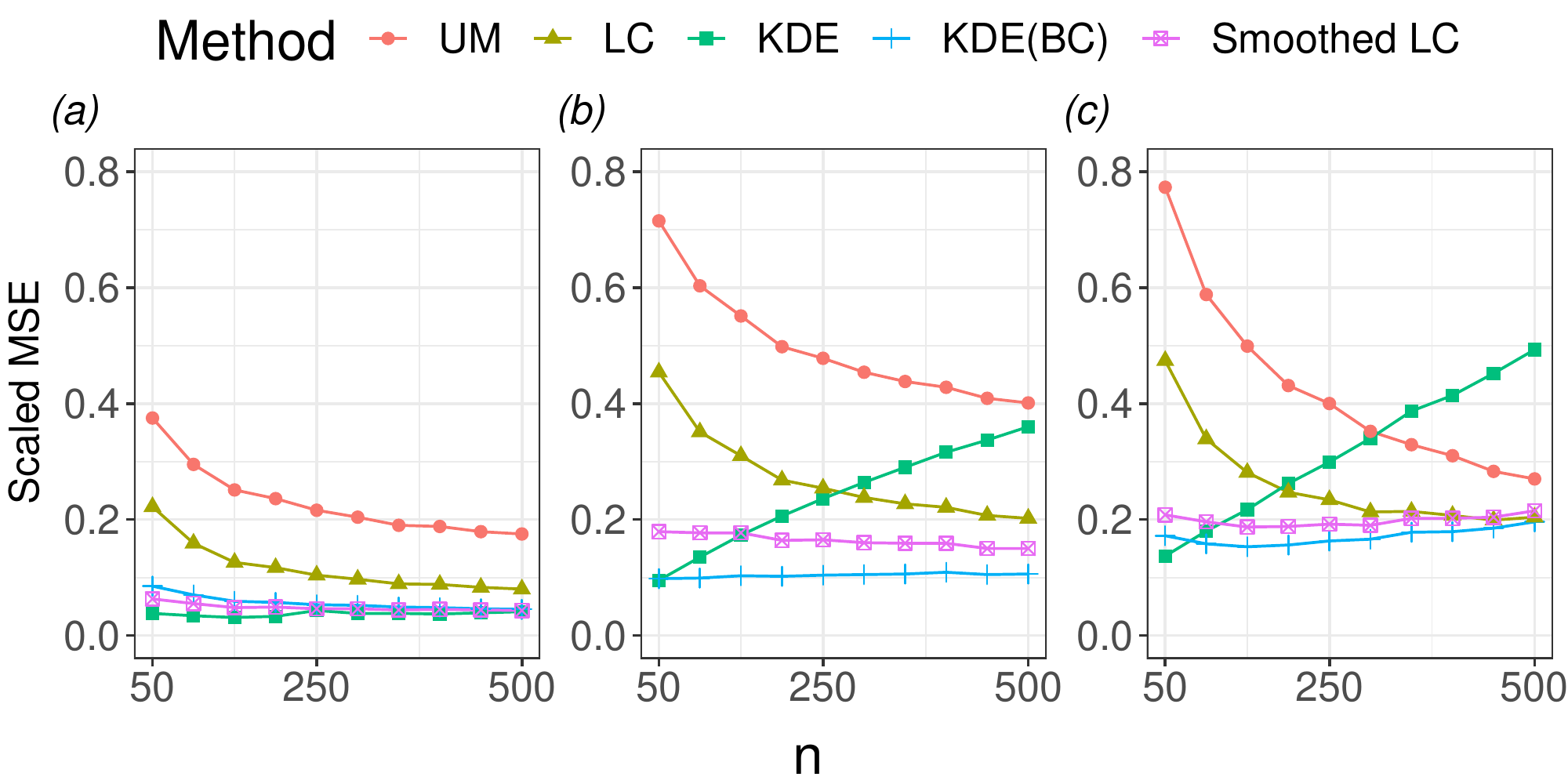}
\caption{Plots of the  MSE (scaled by $n$) for  cases (a)--(c). }
\label{Figure: MSE 1}
\end{subfigure}\\
\begin{subfigure}{\textwidth}\centering
\includegraphics[width=\textwidth]{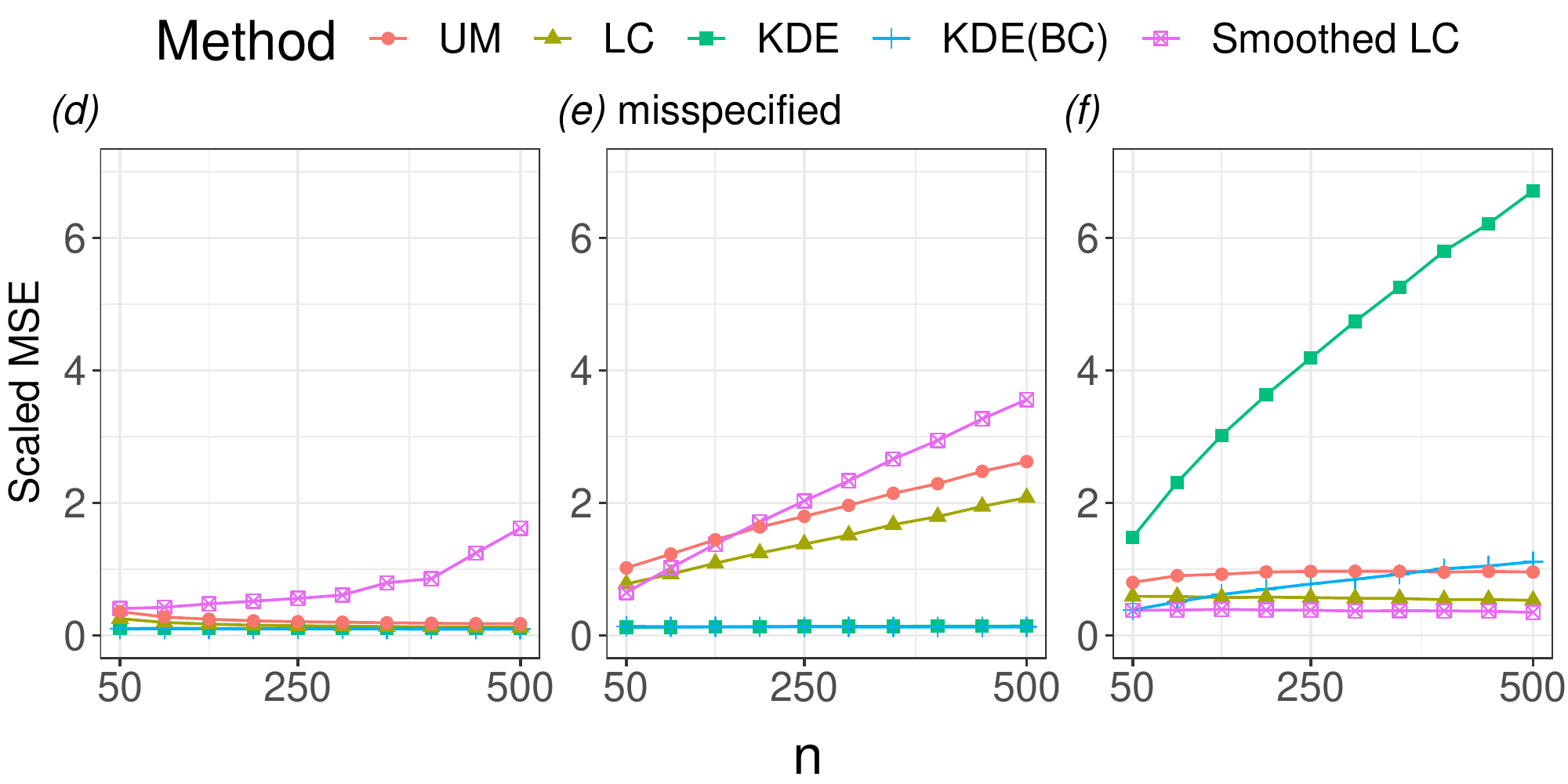}
\caption{Plots of the MSE (scaled by $n$) for cases (d)--(f). }
\label{Figure: MSE 2}
\end{subfigure}
\caption{Plot of MSE: Here ``KDE (BC)" and ``smoothed LC" stand for the bias corrected KDE estimator and the smoothed log-concave estimator, respectively.
For each estimator, either the standard error is less than 0.006 or the relative standard error is less than 2\%. 
}\label{Fig:  MSE}
\end{figure}

\begin{figure}[h]\centering
\begin{subfigure}{\textwidth}\centering
\includegraphics[width=\textwidth]{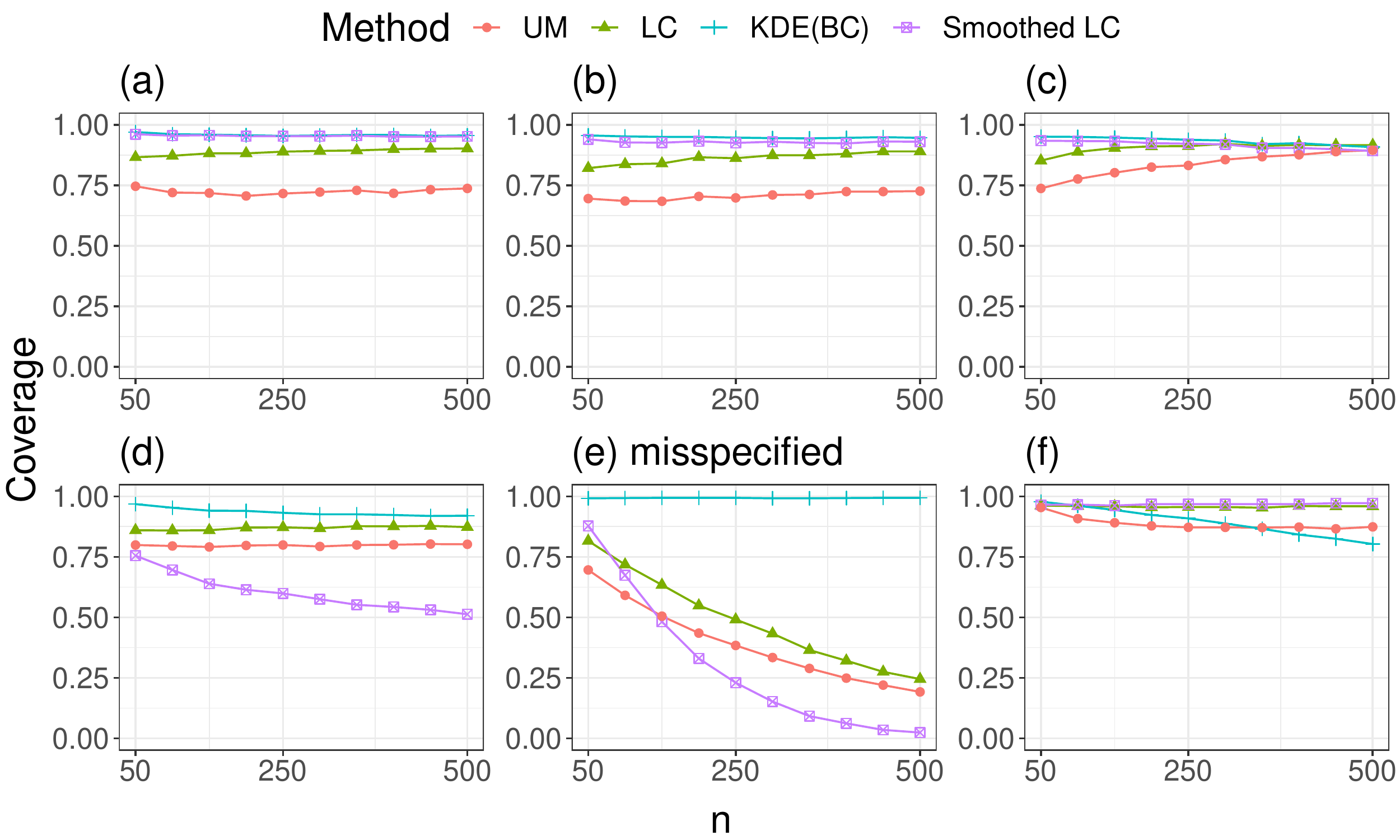}
\caption{Coverage probability of the 95$\%$ confidence intervals for  cases (a)--(f). }
\label{Figure: coverage}
\end{subfigure}\\
\begin{subfigure}{\textwidth}\centering
\includegraphics[width=\textwidth]{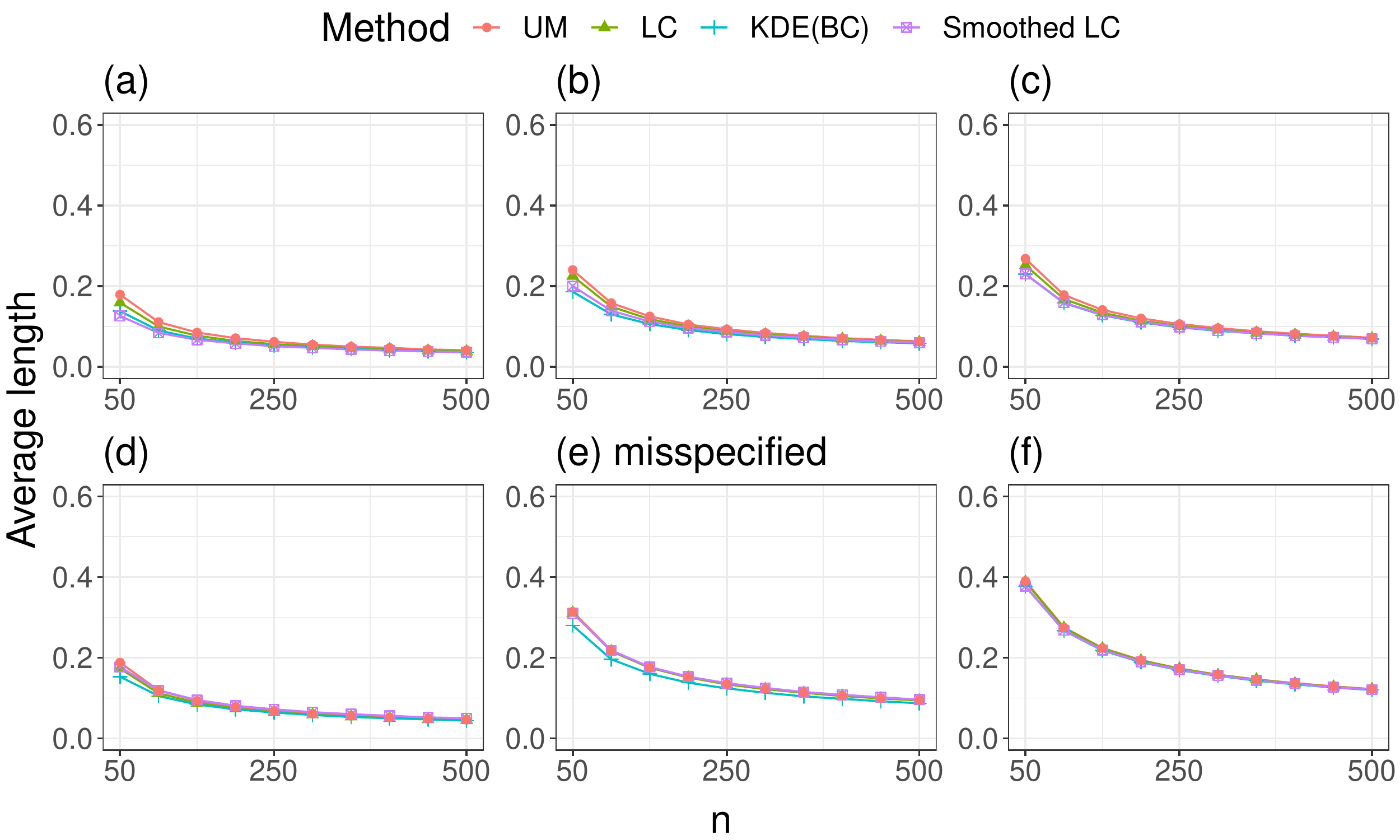}
\caption{Average length of the 95$\%$ confidence intervals for cases (a)--(f). }
\label{Figure: length}
\end{subfigure}
\caption{Here ``KDE (BC)" and ``smoothed LC" stand for the bias corrected KDE estimator and the smoothed log-concave estimator, respectively. 
In all the cases, the standard error is less than 0.005, and hence, not plotted. 
}\label{Fig: cov-length}
\end{figure}

Figures~\ref{Figure: bias 1} and \ref{Figure: bias 2}  plot $\sqn$ times the  absolute bias,  and Figure  \ref{Fig:  MSE} plots  $n$ times the  mean squared error (MSE). Figure~\ref{Figure: coverage} and Figure~\ref{Figure: length} display, respectively, the coverage and the average length of the confidence intervals, both of  which are estimated using  10,000 Monte Carlo samples. 
In the misspecified case (e),  the performance of the shape-constrained estimators deteriorate sharply with $n$, which is unsurprising due to the violation of the shape constraints in this case. \textcolor{black}{In case (d), where the densities are exponential, the scaled bias of the smoothed log-concave plug-in estimator increases with $n$, which implies $\sqrt N$-consistency does not hold for this estimator in this case.    Closer inspection reveals that although it is not $\sqrt N$-consistent, the smoothed log-concave plug-in estimator is  still consistent in case (d).  We found out that the smoothed log-concave density estimator fails to approximate the exponential densities near zero, their point of discontinuity (see Figure~\ref{fig:smooth exp} in Appendix \ref{app: additional tables and figures}). It is worth mentioning that although the curvature Condition~\ref{Cond: B2}  is violated in case (d), it does not affect the performance of the  log-concave plug-in estimator.  }

For all other cases,  the smoothed log-concave estimator exhibits the best performance among the shape-constrained estimators. In all cases, the unimodal estimator underperforms. Importantly, Figure~\ref{Figure: coverage} indicates that the unimodal estimator would require  sample size larger than 500 for the Wald type confidence interval to be valid, where for the other confidence intervals, this sample size is sufficient  for the asymptotics to kick in.

 The KDE-based plug-in estimator experiences an increase in the bias and the MSE  with $n$, which agrees with  our previous discussion on  KDE based plug-in estimators. 
Its bias corrected version performs comparably to the  shape-constrained estimators in all cases except case (f). In   case (f), however, the  bias corrected estimator yields a confidence interval with poor coverage, which is probably  due to the large bias incurred by the original  KDE-based plug-in estimator in this case (see panel (f) of  Figure~\ref{Figure: bias 2}). Case (f) is a case where the underlying densities differ by both shape and scale. We suspect that the cross-validated bandwidth for the KDE estimator does not work in this case even with  bias correction.  Finally, the average length of the confidence intervals do not vary noticeably  across different methods. 

 
 \textcolor{black}{
 In summary, the   log-concave plug-in estimator exhibits reliable performance when  the  shape restriction holds. The smoothed log-concave plug-in estimator may also require the underlying densities to be continuous, but otherwise it performs comparably with the bias-corrected KDE-based plug-in estimator. The latter performs well in many settings, but it is not always reliable. The lacking  performance of the KDE based methods in some settings is probably due to the variable nature of the optimal bandwidth under different settings.
 }

  \subsection{Application to HVTN 097 and HVTN 100 data:}

Table \ref{table: hell estimate} tabulates the point estimates of the squared Hellinger distance (between $\fh$ and $\fn$) and the corresponding 95\% intervals. \textcolor{black}{Table \ref{table: hell estimate} also displays the 95\% confidence intervals of the Hellinger distance, which are obtained by taking square root of the upper and lower bounds of the previous confidence intervals.}

 \begin{table*}[h]
   \centering
  \begin{tabular}{lccccc}
\toprule
Estimator & \makecell{Na\"{i}ve\\ KDE} & \makecell{Bias\\-corrected\\ KDE} & UM & LC & \makecell{Smoothed\\ LC}\\ 
\midrule
\makecell[vl]{Point\\ estimate of\\ $\hd^2(\fn,\fh)$} & 0.16 & 0.19 & 0.18 & 0.15 & 0.21 \\
\makecell[vl]{95\% CI for\\  $\hd^2(\fn,\fh)$}  & \makecell[vl]{not\\available} & $(0.110,0.274)$ & $(0.128, 0.300)$ & $(0.098,0.256)$ & $(0.079,0.228)$ \\
\makecell[vl]{95\% CI for\\  $\hd(\fn,\fh)$} & $--''--$ & $(0.332, 0.523)$ & $(0.358, 0.548)$ & $(0.313, 0.506)$ & $(0.281,0.447)$\\
\bottomrule
\end{tabular} 
\caption{Table of the point estimates of  $\hd^2(\fn,\fh)$, and 95\% confidence intervals of  $\hd^2(\fn,\fh)$ and  $\hd(\fn,\fh)$.  }
\label{table: hell estimate}
  \end{table*}  
  To give the reader some perspective, if $f_0$ is a $N(0,1)$ distribution and $f_\mu$ is a $N(\mu,1)$ distribution, then, when $\mu$ is equal to 1.00, 1.25, 1.50, 1.75, and 2.00, $\hd(f_0,f_\mu)$ is equal to 0.346, 0.424, 0.500, 0.566, and 0.624, respectively. \textcolor{black}{Also,  in Figure~\ref{fig:hell-plot}, we display some density pairs $(f,g)$ satisfying $F\succeq G$ with   Hellinger distance in the range $0.35-0.60$, which is  similar to our data.}
\begin{figure}
    \centering
    \includegraphics[width=\textwidth]{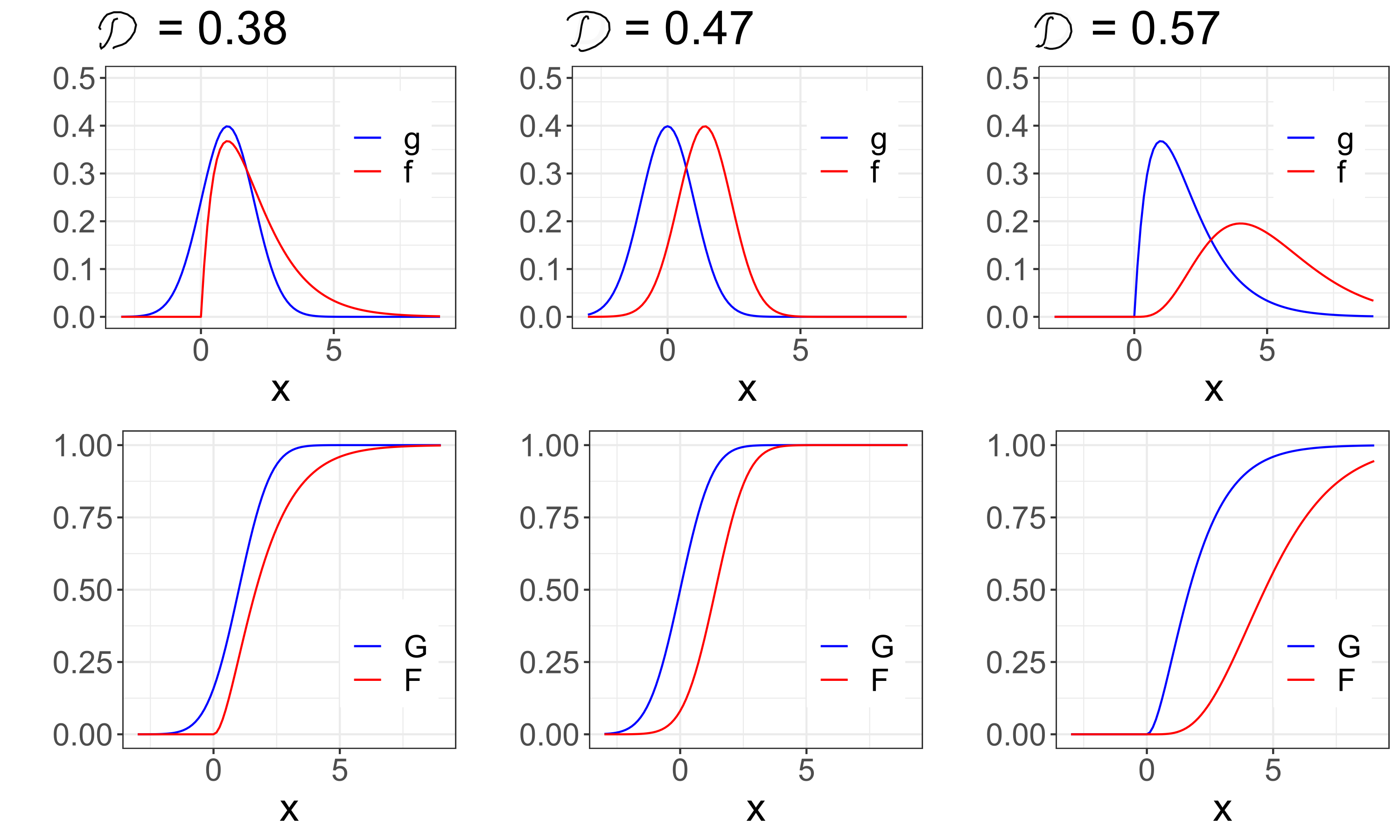}
    \caption{Plot of different density pairs and their distribution functions. The Hellinger distance $(\mathcal D)$ is given on top of each pair. Left: $f\sim Gamma(25,1)$, $g\sim N(0.5,1)$; middle: $f\sim N(1.40, 1),$ $g\sim N(0,1)$; right: $f\sim Gamma(5,1)$, $g\sim Gamma(2,1)$.}
    \label{fig:hell-plot}
\end{figure}

\section{Discussion}

The first contribution of our work is a novel analysis of the data from the HVTN 097 and HVTN 100 trials. All of our tests  reject the null of non-dominance   in favor of the strict stochastic dominance of $\Fn$ over $\Fh$. To provide further insight into the discrepancy between the two   IgG binding response distributions, we estimated the squared Hellinger distance between the corresponding densities, which turns out to be approximately 0.20 (95\% CI 0.10-0.30). 
We remark that our findings are consistent with those of \cite{HVTN100_primary}, who found that the average magnitude of IgG binding to V1V2 antigens observed in the HVTN 100 trial is lower than that in the RV 144 trial. Although the latter used the same regimen as HVTN 097, it was conducted in a different population (Thailand).  Our findings indicate that the difference in the magnitude of IgG binding response between HVTN 100 and RV 144 regimen may be attributable to the HIV clade difference rather than to the difference in populations.


The outcome of our tests become meaningful when viewed  against the lack of efficacy observed
 in the phase 2b/3 HVTN 702 trial, which evaluated the safety and efficacy of the HVTN 100 regimen in South Africa. 
A possible hypothesis for  why no efficacy was observed when the HVTN 100 regimen was evaluated in this trial, whereas efficacy was observed when the HVTN 097 regimen was evaluated in the RV144 trial, is that the HVTN 100 regimen leads to a lower magnitude of  IgG binding to the V1V2 region. This possibility is supported by the observation made by \cite{RV144_immuno} regarding the negative correlation between rate of infection and the magnitude of IgG binding to V1V2 region. This hypothesis can be tested when the immune profile of the participants in HVTN 702 trial becomes available.

  
  Another contribution of our work  relates to density estimation in the context of vaccine trials. Based on a cross-validated analysis of the HVTN 097 and HVTN 100 data, we believe that the log-concave density estimators of \cite{2009rufi} and \cite{smoothed}  may yield improved density estimation in  vaccine studies. In future work, it would be worth further validating this claim on other vaccine trial datasets.
  
 We also made several methodological contributions. In Section \ref{sec: tests}, we introduce three novel shape-constrained tests. These tests have the desirable asymptotic properties of nonparametric tests, and simulations illustrate that the shape-constrained tests have better overall performance than the nonparametric tests. Moreover, even under the violation of the shape-constraints, their performance is not much worse than the nonparametric tests. 
     We also introduce shape-constrained plug-in estimators of the squared Hellinger distance  and provide asymptotic consistency and distributional results. Our simulations suggest that, when the shape constraint is  satisfied, the log-concave plug-in estimators exhibit overall lower MSE and absolute bias than the KDE based plug-in estimator. In fact, they perform comparably with the bias-corrected KDE based estimator. 
    However, unlike the bias-corrected KDE plug-in estimator, the log-concave plug-in estimators require neither selecting a tuning parameter nor carrying out the analytic calculations needed to derive the bias-correction term. Therefore, the log-concave plug-in estimators may be preferred in settings where this shape constraint is plausible. 
\section{Acknowledgements}
This work was supported by the National Institutes of Health (NIH) through award numbers DP2-LM013340 and 5UM1AI068635-09. The content is solely the responsibility of the authors and does not necessarily represent the official views of the NIH.
GlaxoSmithKline Biologicals SA was provided the opportunity to review a preliminary version of this manuscript, but the authors are solely responsible for final content and interpretation.


\bibliographystyle{natbib}
\bibliography{location_estimation_1}
\FloatBarrier
\appendix

  \section{Proofs for Section~\ref{sec: tests}}
   \label{sec: appendix: tests}
  
 Before proceeding any further, we introduce some new notations. We let $l^{\infty}$ denote the collection of all bounded functions on $\RR$, which we equip with the uniform metric $\|\cdot\|_{\infty}$. For any function $\mu:\RR\mapsto\RR$, we define the norm $\|\cdot\|_{a}^{b}$ by
 $\|\mu\|_{a}^{b}=\sup_{x\in[a,b]}|\mu(x)|$.
 Also, we denote the boundary of a set $A$ by $\text{bd}(A)$.
 
 
  Suppose $(\Omega,\iint(\Hm_0),\mathcal{P})$ is the common probability space corresponding to the $X_i$'s and $Y_j$'s. For the rest of this section, we let $\to_p$ and $\as$ correspond to this probability space. Since $F$ and $G$ are continuous, using the construction of Section 1.1 of \cite{shorack1984} \citep[see also p.93 of][]{shorack}, we can show that there exist two independent Brownian bridges $\mathbb{V}_1$ and $\mathbb{V}_2$ on  $(\Omega,\iint(\Hm_0),\mathcal{P})$ such  that 
    \begin{equation}\label{intheorem: convg of empirical process : F}
 \norm{\sqrt{m}(\Fmx-\Fx)-\mathbb{V}_1\circ \Fx}_{\infty}\as 0\quad\text{ as }m\to\infty,
 \end{equation}
 and
 \begin{equation}\label{intheorem: convg of empirical process : G}
 \norm{\sqn(\Fmy-\Fy)-\mathbb{V}_2\circ\Fy}_{\infty}\as 0\quad\text{ as }n\to\infty.
 \end{equation}
 Let us denote
\begin{equation}\label{def: U}
    \mathbb{U} =\lambda ^{1/2}\mathbb{V}_2-(1-\lambda )^{1/2} \mathbb{V}_1,
\end{equation}
where $\lambda $ is so that  $m/N\to\lambda$. Note that $\mathbb{U}$ is also distributed as a Brownian bridge. We will show that the asymptotic distributions of our test statistics depend on $\mathbb{U}$. Also, for ours $F$ and $G$,
\[\|\Fmx-\Fx\|_\infty\as 0,\quad \|\Fmy-\Fy\|_\infty\as 0,\quad\|\mathbb H_N-H\|_\infty\as 0.\]

In the sequel, we will also use the following fact on the convergence of the quantiles of $\mathbb H_N$ and $D_{p,m,n}$, often without mentioning. As a corollary to Fact~\ref{fact: quantile}A, it follows that  $\text{dist}(D_{p,m,n},D_p)\as 0$ under Condition \ref{assump: Nonparametric}.
 \begin{fact}
 \label{fact: quantile}
  Suppose Condition \ref{assump: Nonparametric} holds and $U$ is the open set such that $D_p\subset U\subset \supp(f)\cup \supp(g)$. Further suppose $p'<p$ such that $D_{p'}\subset U$. Then
  \begin{itemize}
      \item[A.] For any $t\in[p',1-p']$, $\mathbb H_N^{-1}(t)\as H^{-1}(t)$.
      \item[B.] $D_{p,m,n}\subset D_{p'}$ for all sufficiently large $m$ and $n$ with probability one.
      \item[C.] For any $q>p$,  $D_{q}\subset D_{p,m,n}$ for all sufficiently large $m$ and $n$ with probability one.
  \end{itemize}
    \end{fact}
    
  \begin{proof}[Proof of Fact~\ref{fact: quantile}]
  Since $h=\lambda f+(1-\lambda)g>0$ on $U$, it follows that $H^{-1}$ is continuous on $U$ \citep[cf. Proposition A.7, pp. 98 of][]{bobkovbig}. Also,  under Condition \ref{assump: Nonparametric}, $\|\mathbb H_N-H\|_\infty\as 0$. Since $D_{p'}\subset U$, $H^{-1}(t)\in U$ for any $t\in [p',1-p']$. Hence, $\mathbb H_N^{-1}(t)\as H^{-1}(t)$ by Lemma A5, pp. 96 of \cite{bobkovbig}. Hence part (A) of Fact~\ref{fact: quantile} is proved.

  Part (A) of Fact~\ref{fact: quantile} implies
  $\mathbb H_N^{-1}(p)\as H^{-1}(p)$ and $\mathbb H_N^{-1}(1-p)\as H^{-1}(1-p)$.  To prove part B and C of Fact~\ref{fact: quantile}, therefore, it suffices to show that $H^{-1}(p')<H^{-1}(p)$, $H^{-1}(1-p')>H^{-1}(1-p)$, $H^{-1}(p)<H^{-1}(q)$, and $H^{-1}(1-q)<H^{-1}(1-p)$. Since $q<p<p'$ and  $D_{q}\subset D_p\subset D_{p'}\subset U$, it is enough to show that $H^{-1}$ is continuous and strictly increasing on $U$ . The latter holds if $H$ s continuous and strictly increasing on $U$. The proof now  follows  since we already showed that $H$ has a positive density on $U$.
 \end{proof}
{\color{black}
 \subsubsection{Proof of Lemma~\ref{Lemma: geometry of H0}}
 To prove Lemma~\ref{Lemma: geometry of H0}, we will use an alternative definition of $\Hm_0$ and $\Hm_1$.
Note that continuous distribution functions $(F,G)\in \Hm_0$ if and only if
 $\sup_{x\in D_p}[F(x)-G(x)]\geq 0$.
 Because $D_p$ is compact, such a pair is in $\Hm_1$ if and only if
 $\sup_{x\in D_p}[F(x)-G(x)]< 0$.
 We will prove the current lemma in some steps.\\
 \textbf{Step 1: closure of $\Hm_0$}
 Suppose $(F_n,G_n)\in \Hm_0$ converges to $(F,G)\in \mathcal F\times \mathcal F$ with respect to the metric $d_2$. Then
 \begin{align*}
 \MoveEqLeft \sup_{x\in D_p}(F(x)-G(x))
 \geq  -\sup_{x\in D_p}(F_n(x)-F(x))\\
 &\ + \sup_{x\in D_p}(F_n(x)-G_n(x))- \sup_{x\in D_p}(G(x)-G_n(x)),
\end{align*}
which is bounded below by $ -\|F_n-F\|_\infty -\|G_n-G\|_\infty$.
Taking $n\to\infty$ yields $\sup_{x\in D_p}(F(x)-G(x))\geq 0$. Therefore,
\[ \text{cl}(\Hm_0)\subset \lbs (F,G)\in\mathcal F\times \mathcal F: \sup_{x\in D_p}(F(x)-G(x))\geq 0\rbs=\Hm_0.\]
Hence, $\text{cl}(\Hm_0)=\Hm_0$, i.e. $\Hm_0$ is closed in $\mathcal F\times \mathcal F$.

\textbf{Step 2: interior and boundary of $\Hm_0$}\\
If $(F,G)\in \Hm_0$, either
$\sup_{x\in D_p}(F(x)-G(x))>0$ or $\sup_{x\in D_p}(F(x)-G(x))=0$. 
Let us denote 
\[\mathcal A=\{(F,G)\in\mathcal F\times \mathcal F: \sup_{x\in D_p}(F(x)-G(x))=0\}.\]
By Lemma~\ref{Lemma: geometry: boundary lemma}, $\mathcal A\subset \text{bd}(\Hm_0)$ because this set is not a part of the interior.  Hence,
\[\iint(\Hm_0)\subset \{(F,G)\in \mathcal F\times \mathcal F: \sup_{x\in D_p}(F(x)-G(x))>0\}.\]
Therefore, to prove
\begin{equation}\label{inlemma: geometry: H0 interior}
     \iint(\Hm_0)=\{(F,G)\in \mathcal F\times \mathcal F: \sup_{x\in D_p}(F(x)-G(x))>0\},
\end{equation}
it suffices to prove
\begin{equation}\label{inlemma: geometry: H0 interior: inclusion}
  \{(F,G)\in \mathcal F\times \mathcal F: \sup_{x\in D_p}(F(x)-G(x))>0\}\subset \iint(\Hm_0).
\end{equation}

Suppose $(F,G)\in \Hm_0$ satisfies $\sup_{x\in D_p}(F(x)-G(x))>\e$ for some $\e>0$. Consider any $(\tilde F,\tilde G)$ such that $d_2((\tilde F,\tilde G),(F,G))<\e/3$, which means $\|\tilde F-F\|_\infty<\e/3$ and $\|\tilde G-G\|_\infty<\e/3$.
Since $F-G$ is continuous and $D_p$ is compact, the supremum of $F-G$ over $D_p$ is attained at some $z\in D_p$. We have
 \begin{align*}
  \MoveEqLeft  \tilde F(z)-\tilde G(z) 
  \geq  F(z)-G(z)-\|\tilde F-F\|_\infty-\|\tilde G-G\|_\infty,
 \end{align*}
 which is greater than $\e/3$. Thus $\sup_{x\in D_p}(\tilde F(x)-\tilde G(x))>0$ which implies $(F,G)\in\iint(\Hm_0)$. Therefore, we have established \eqref{inlemma: geometry: H0 interior}.
Because $\text{cl}(\Hm_0)=\Hm_0$, \eqref{inlemma: geometry: H0 interior} also implies $\text{bd}(\Hm_0)=\mathcal A$.

  
 \textbf{Step 3:  boundary of $\Hm_1$} \\
 Our first step is to show
 \begin{align}\label{inlemma: geometry: H1  closure: inclusion}
    \text{cl}(\Hm_1)\subset \lbs (F,G)\in \mathcal F\times\mathcal F:\sup_{x\in D_p}(F(x)-G(x))\leq 0\rbs.
\end{align}
 To this end, consider $(F_n,G_n)\in \Hm_1$ converging to $(F,G)\in \mathcal F\times\mathcal F$ in $d_2$. Then 
\begin{align*}
 \sup_{x\in D_p}(F(x)-G(x))
 \leq  \sup_{x\in D_p}(F(x)-F_n(x))+ \sup_{x\in D_p}(F_n(x)-G_n(x))+ \sup_{x\in D_p}(G_n(x)-G(x)),
\end{align*}
which is bounded above by $ \|F_n-F\|_\infty +\|G_n-G\|_\infty$.
Letting $n\to\infty$, we obtain 
$\sup_{x\in D_p}(F(x)-G(x))\leq 0$ implying $(F,G)\in \Hm_1$. Hence, \eqref{inlemma: geometry: H1  closure: inclusion} holds, which implies $\text{cl}(\Hm_1)\subset \mathcal A\cup  \Hm_1$. Using Lemma \ref{Lemma: geometry: boundary lemma} we obtain
\[ \text{cl}(\Hm_1)\subset\mathcal A\cup  \Hm_1\subset \text{bd}(\Hm_1)\cup \Hm_1=\text{cl}(\Hm_1).\]
Hence, the inclusion in \eqref{inlemma: geometry: H1  closure: inclusion} is actually an equality and $\mathcal A\cup  \Hm_1= \text{bd}(\Hm_1)\cup \Hm_1$. The proof will be complete if we can show that $\iint(\Hm_1)=\Hm_1$, because then $\text{bd}(\Hm_1)=\mathcal A$ follows. To that end, consider $(F,G)\in \Hm_1$. Suppose $\sup_{x\in D_p}(F(x)-G(x))<-\e$. Let $(\tilde F,\tilde G)\in\FF$ be such that $d_2((\tilde F,\tilde G),(F,G))<\e/3$. Then 
\begin{align*}
  \MoveEqLeft \sup_{x\in D_p}(\tilde F(x)-\tilde G(x)) 
  \leq  \sup_{x\in D_p}( F(x)- G(x))+\|\tilde F-F\|_\infty+\|\tilde G-G\|_\infty,
\end{align*}
which is less than $-\e/3$. Therefore, $(\tilde F,\tilde G)\in \Hm_1$, which completes the proof.
  \hfill $\Box$



 \begin{lemma}\label{Lemma: geometry: boundary lemma}
 Under the set-up of Lemma~\ref{Lemma: geometry of H0},
 \[\mathcal A:=\lbs (F,G)\in\mathcal F\times \mathcal F: \sup_{x\in D_p}(F(x)-G(x))=0\rbs\subset \text{bd}(\Hm_1).\]
 \end{lemma}
 
 \begin{proof}
 Consider $(F,G)\in\mathcal A$. Because $\mathcal A\cap \Hm_1=\emptyset$, it suffices to show that $(F,G)$ is a limit point of $\Hm_1$. We will show that given $\delta>0$, there exists $(\tilde F,\tilde G)\in \Hm_1$ so that
 $d_2((F,G),(\tilde F,\tilde G))<\delta$. 
 
 Since $F-G$ is continuous, its supremum over $D_p$ is attained. Thus, there exists $C_p\subset  D_p$ so that $F=G$ on $C_p$. Because $ p\leq \lambda F+(1-\lambda)G\leq 1-p$ on $D_p$,  $\inf_{z\in C_p}F(z)\geq p$ and $\sup_{z\in C_p}F(z)\leq 1-p$. Suppose $a'=\inf C_p$ and $\beta'=\sup C_p$. Since $F$ is continuous, there exists closed interval $L_p=[a,\beta]\supset C_p$ such that $L_p\subset D_p$ and  $p/2<F<1-p/2$ on $[a,\beta]$. We choose $a$ and $\beta$ so that additionally the followings hold:
 \begin{enumerate}
     \item If $a'>H^{-1}(p)$ then $a<a'$. Thus $a=a'$ only if $a'=H^{-1}(p)$, in which case, the only choice for $a$ is $a'$ because $[a',\beta']\subset [a,\beta]\subset D_p$.
     \item If $\beta'<H^{-1}(1-p)$ then $\beta'<\beta$. When $\beta'=H^{-1}(1-p)$, the only choice for $\beta$ is $\beta'$ because $[a',\beta']\subset [a,\beta]\subset D_p$.
 \end{enumerate}

 Letting $\e'=\min(\delta/2,p/4)$,  $t_1=F^{-1}(1-\e')$ and $t_2=t_1+1$,
 we define 
 \begin{equation}\label{inlemma: geometry: tilde F}
     \tilde F(x)=\begin{cases}
     \max(0,F(x)-\e') & x< a\\
     F(x)-\e' & x\in [a,t_1]\\
     (x-t_1)\frac{F(t_2)-\tilde F(t_1)}{t_2-t_1}+\tilde F(t_1) & x\in (t_1,t_2]\\
     F(x) & x>t_2
     \end{cases}
 \end{equation}
 That $\tilde F$ is a distribution function is clear from the definition. We claim that (a) $\tilde F$ is continuous, (b) $\|\tilde F-F\|_\infty < \delta$, and finally
 (c) $(\tilde F,G)\in \Hm_1$. Taking $\tilde G=G$, the proof of the current lemma follows. Hence, it remains to prove Claim (a), (b), and (c).

  First, we prove  Claim (a). Because  $F$ is continuous, $\max(0,F(x)-\e')$ is continuous in $x$ on $(-\infty, a)$. BSince $ F(a)>p/2>\e'$, continuity of $F$ also implies $\tilde F(a-)=F(a)-\e'=F(a+)$. Therefore, $\tilde F$ is continuous on $(-\infty,a]$.
  Therefore  \eqref{inlemma: geometry: tilde F} implies $\tilde F$ is continuous on $(-\infty,a]$. The continuity of $F(x)-\e'$ on $(a,t_1]$ follows because $F$ is continuous.  $\tilde F$ is linear on $(t_1,t_2]$, and equals $F$ on $(t_2,\infty)$ Moreover, the left and right limits of $\tilde F$ agree at $t_1$ and $t_2$. Therefore $\tilde F$ is continuous.
  
   To prove Claim (b), it suffices to show that $|\tilde F(x)-F(x)|\leq \e'$ for any $x\in(-\infty,a)\cup (t_1,t_2]$. Note that if $x<a$ and $F(x)\leq \e'$, then $\tilde F(x)=0$. Thus, $\tilde F(x)-F(x)\leq \e'$ in this case. If $x<a$ but $F(x)>\e'$, $\tilde F(x)-F(x)=\e'$. Therefore, 
   $|\tilde F(x)-F(x)|\leq \e'$ for $x\in(-\infty,a)$. 
   On the other hand, since $t_1=F^{-1}(1-\e')$, $F(x)$ is greater than $1-\e'$ on $(t_1,t_2]$. Also for $x\in [t_1,t_2]$,
   \begin{align*}
     \tilde F(x)\geq   \tilde F(t_1)= F(t_1)-\e'= F(F^{-1}(1-\e'))-\e'=1-2\e'.
   \end{align*}
     Thus on $[t_1,t_2]$, $|\tilde F(x)-F(x)|$ is bounded by $2\e'$, 
   which is not greater than $\delta$ by our choice of $\e'$. Hence, we have shown that $\|\tilde F-F\|_\infty<\delta$.
   
   To prove Claim (c), first note that $\tilde F\leq F$ on $\RR$.  Let us partition $D_p=[H^{-1}(p),a)\cup [a,\beta]\cup (\beta, H^{-1}(1-p)]$. 
   If $a=H^{-1}(p)$, then we define $[H^{-1}(p),a)$ to  be the empty set. Similarly, $(\beta, H^{-1}(1-p)]$ is non-empty only if $\beta< H^{-1}(1-p)$. Note that $\beta\leq t_1$, implying  $[a,\beta]\subset[a,t_1]$. Therefore, by \eqref{inlemma: geometry: tilde F},
  \[\sup_{x\in [a,\beta]}(\tilde F(x)-G(x))= \sup_{x\in [a,\beta]}(F(x)-G(x))-\e'\leq \sup_{x\in D_p}(F(x)-G(x)) -\e',\]
  which equals $-\e'$. Now let us consider the set 
  $[H^{-1}(p),a)$. Of course if $a=H^{-1}(p)$, there is nothing to prove. So suppose $a>H^{-1}(p)$. Because $C_p\cap [H^{-1}(p),a)=\emptyset$, there is no $x\in [H^{-1}(p),a)$ such that $F(x)-G(x)=0$. Therefore, if $\sup_{x\in [H^{-1}(p),a)}(F(x)-G(x))=0$, the set $[H^{-1}(p),a]$ must contain a limit point of $C_p$. However,  $\text{cl}(C_p)\subset [a',\beta']$ where $a<a'$ because $a>H^{-1}(p)$. Therefore, $[H^{-1}(p),a]$ can not contain any limit point of $C_p$ either. Thus we must have
  $\sup_{x\in [H^{-1}(p),a)}(F(x)-G(x))<0$. Therefore, $\sup_{x\in [H^{-1}(p),a)}(\tilde F(x)-G(x))<0$ as well because $\tilde F\leq F$.
  On the other hand, using $H^{-1}(1-p)<t_1$, we have
  \[\sup_{x\in (\beta, H^{-1}(1-p)]}(\tilde F(x)-G(x))=\sup_{x\in (\beta, H^{-1}(1-p)]}( F(x)-G(x))-\e\leq  \sup_{x\in D_p}(F(x)-G(x)) -\e',\]
  which equals $-\e'$. Combining the above pieces, we obtain $\sup_{x\in D_p}(\tilde F(x)-G(x))<0$, which completes the proof 
  of $(\tilde F,G)\in \Hm_1$.
 
\end{proof}
}

 \subsubsection*{Proof of Lemma~\ref{thm: dist based: null distribution: A}}
 \label{secpf: Lemma null dist A}
   Since $F,G\in\iint(\Hm_0)$,  there exists $x'_0\in D_p$ and $\delta>0$ such that $G(x'_0)-F(x'_0)<-3\delta$. Because $F$ and $G$ are continuous, we can find $x_0\in \iint(D_p)$  such that $G(x_0)-F(x_0)<-2\delta$. 
 Hence, \eqref{intheorem: type 1 error} indicates that with probability one,   $\Fmy(x_0)-\Fmx(x_0) <-\delta$ for all sufficiently large $m$ and $n$, 
which leads to
\[\dfrac{\Fmy(x_0)-\Fmx(x_0)}{\sqrt{\dfrac{\Fmx(x_0)(1-\Fmx(x_0))}{m}+\dfrac{\Fmy(x_0)(1-\Fmy(x_0))}{n}}} \leq \dfrac{-\delta}{\sqrt{1/m+1/n}},\]
which approaches $-\infty$ as $m,n\to\infty$.
 Here we used the fact that
 \[\dfrac{\Fmx(x_0)(1-\Fmx(x_0))}{m}+\dfrac{\Fmy(x_0)(1-\Fmy(x_0))}{n}\leq \frac{1}{m}+\frac{1}{n}.\]
Because $x_0\in\iint (D_p)$, Fact \ref{fact: quantile} implies  $x_0\in D_{p,m,n}$ almost surely.
 Therefore, the proof follows for the minimum t-statistic. 
 
 Now we will prove the result for for $T_{m,n}^{\text{tsep}}$. Let us denote $q=H(x_0)$ where $H$ was denoted to be $\lambda F+(1-\lambda)G$. Since either $f>0$ or $g>0$ on $D_p$ under Condition~\ref{assump: Nonparametric}, it follows that $h>0$ on $D_p$. Thus $H^{-1}(q)=x_0$ \citep[see Lemma A.3.7, pp.94,][]{bobkovbig}. 
 Hence,
 \[G(H^{-1}(q))-F(H^{-1}(q))<-2\delta.\]
 Since $H^{-1}(q)\in\iint(D_p)$, it follows that $q\in[p,1-p]$. Fact \ref{fact: quantile} also implies that $\mathbb H_N^{-1}(q)\as H^{-1}(q)$. Combined with the fact that $F$ and $G$ are continuous, we obtain that 
 \[\limsup_{m,n}\{G(\mathbb H_N^{-1}(q))-F(\mathbb H_N^{-1}(q))\}<-\delta\]
 with probability one. Since $\|\Fmx-\Fx\|_\infty$ and $\|\Fmy-\Fy\|_\infty$ converges to zero almost surely, the above implies
 \[\limsup_{m,n}\{\Fmy(\mathbb H_N^{-1}(q))-\Fmx(\mathbb H_N^{-1}(q))\}<-\delta/2\]
 almost surely.
 Note that
 \[\limsup_{m,n} T_{m,n}^{\text{tsep}}\leq \limsup_{m,n}\sqrt{\frac{mn}{N}}\frac{G(\mathbb H_N^{-1}(q))-F(\mathbb H_N^{-1}(q))}{q(1-q)}\leq -\lim_{m,n}\sqrt{\frac{mn}{N}}\frac{\delta}{2q(1-q)},\]
 which equals $-\infty$ because $m/N\to\lambda$.
 Hence, the proof follows.

  \subsubsection{Proof of Theorem~\ref{thm: dist based: null distribution: B} }
  \label{sec:proof: main thm}
 \begin{figure}
 \centering
  \begin{tikzpicture}[scale=.8]
\draw [-,thick, blue, ultra thick] (0,1) to [out=60,in=180] (3,3)
        to  (6,3) to [out=0,in=-135] (9,6) ; 
        \node[above]at (0,1.4){\large $G$};
 \draw [-,thick, red, ultra thick] (0,-.5) to [out=70,in=180] (5,3)
        to  (7,3) to [out=0,in=-135] (9,5) ; 
           \node[above]at (0,-.1){\large $F$};
   \draw[<->,thick,black] (4.8,3.3) to (6.2,3.3);
   \node[above ] at (5.5,3.3) {\large $C_{p,m,n}$};
   \draw[-,thick,black, dotted, thick] (4.8,3.6) to (4.8,3);
   \draw[-,thick,black, dotted, thick] (6.2,3.6) to (6.2,3);
    \draw[<->,thick, dashed,black] (4,5) to (7,5);
      \node[above ] at (5.5,5) {\large $C_{p,m,n}(\sigma_\e)$};
     \draw[-,thick, dotted,black] (4,5.5) to (4,2.5);
      \draw[-, thick, dotted,black] (7,5.5) to (7,2.5);
      \draw[<->,dotted,ultra thick, black] (2,2) to (8,2);
        \node[above ] at (5,2) {\large $D_{p,m,n}$};
      \draw[-,dotted, thick, black](2,1)--(2,4);
        \draw[-,dotted, thick, black](8,1)--(8,4);
\end{tikzpicture}
\caption{An illustration  of $D_{p,m,n}$, $C_{p,m,n}$, and  $C_{p,m,n}(\sigma_\e)$ as in the Proof of Theorem~\ref{thm: dist based: null distribution: B} for $(F,G)\in\mathcal B$.}
\label{Fig: Cp, etc.}
 \end{figure}
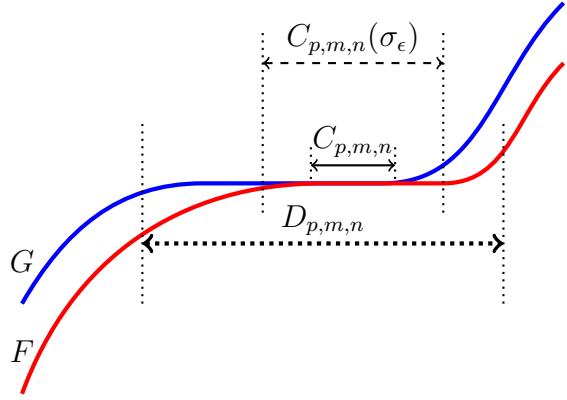

  Before proving Theorem~\ref{thm: dist based: null distribution: B}, we introduce some notations and two lemmas.
Recall that we defined the underlying probability space to be $(\Omega,\iint(\Hm_0),P)$. There exists $A\subset\Omega$ with  $P(A)=1$ such that as $m,n\to\infty$, the following assertions hold on $A$:
\begin{flalign}\label{intheorem: emp:  as}
\text{(a)} &&\norm{\Fmx-\Fx}_{\infty}\to 0,\quad\norm{\Fmy-\Fy}_{\infty}\to 0,\quad\quad \quad\quad \quad\quad \quad\quad
\end{flalign}
\begin{flalign}\label{intheorem: emp: weak}
 \text{(b)} && \norm{\sqrt{m}(\Fmx-\Fx)-\mathbb{V}_1\circ F}_{\infty}\to 0,\quad \norm{\sqrt{n}(\Fny-\Fy)-\mathbb{V}_2\circ \Fy}_{\infty}\to 0,\quad\quad \quad\quad \quad\quad
 \end{flalign}
where  $\mathbb{V}_1$ and $\mathbb{V}_2$ are the Brownian bridges defined in \eqref{intheorem: convg of empirical process : F} and \eqref{intheorem: convg of empirical process : G}, respectively. 

\noindent (c) The trajectories of the Brownian bridges $\mathbb{V}_1$ and $\mathbb{V}_2$  are continuous on $\RR$.

     Define
  \begin{align}\label{def: wmn}
  \omega_{m,n}(x)=\dfrac{\Fmy(x)-\Fmx(x)}{\sqrt{\dfrac{\Fmx(x)\slb 1-\Fmx(x)\srb}{m}+\dfrac{\Fmy(x)\slb 1 -\Fmy(x)\srb}{n}}}
  \end{align}
  and
  \begin{equation}\label{def: w1}
   \omega_{0}(x)=\dfrac{\sqrt{\lambda }\mathbb{V}_2(\Fy)-\sqrt{1-\lambda }\mathbb{V}_1(\Fx)}{\sqrt{(1-\lambda )\Fx(x)\slb 1-\Fx(x)\srb +\lambda \Fy(x)\slb 1-\Fy(x)\srb}}.
  \end{equation}
 Let us denote $C_{p,m,n}=D_{p,m,n}\cap \{x:\Fx(x)=\Fy(x)\}$. Proceeding as in the proof of Fact~\ref{fact: quantile}, we can show that
 \begin{equation}\label{conv: D and A}
  \text{Dist}(D_{p,m,n},D_p)\to 0\quad\text{ and }\quad \text{Dist}(C_{p,m,n},C_p)\to 0\quad\text{as}\quad m,n\to\infty
 \end{equation}
  on A, where $C_p$ is the contact set $D_p\cap \{x\in\RR\ :\ F(x)=G(x)\}$ discussed in Section~\ref{sec: tests}. Now we state the first lemma, which we require for proving part (A).
 
 \begin{lemma}\label{Lemma: Thm 2: 1}
 Suppose $D\subset \RR$ such that $\Fx$ and $\Fy$ are bounded away from $0$ and $1$ on $D$.   
Then, under the conditions of Theorem~\ref{thm: dist based: null distribution: B}, the following holds on $A$:
 \begin{equation*}
 \lim_{m,n\to\infty}\norm{ \omega_{m,n}-\omega_0-\sqrt{\dfrac{mn}{N}}\vartheta_{m,n}^{-1/2}(\Fy-\Fx) }_{D}=0,
 \end{equation*}   
 where
 \begin{equation}\label{def: vartheta}
\vartheta_{m,n}(x)= \dfrac{n\Fmx(x)\slb 1-\Fmx(x)\srb}{N}+\dfrac{m\Fmy(x)\slb 1 -\Fmy(x)\srb}{N}.
 \end{equation}
 Moreover, on $A$,
  \begin{equation}\label{intheorem: convg: vn}
 \lim_{m,n\to\infty}\|\vartheta_{m,n}^{-1/2}-\vartheta_0^{-1/2}\|_D= 0,
 \end{equation}
 where
 \[\vartheta_0(x)=(1-\lambda )\Fx(x)(1-\Fx(x))+\lambda \Fy(x)(1-\Fy(x)).\]
 \end{lemma}
 
 \begin{proof}
 Note that
 \begin{align*}
 \MoveEqLeft \omega_{m,n}(x)-\omega_0(x)- \sqrt{\dfrac{mn}{N}}\vartheta_{m,n}(x)^{-1/2}\slb \Fy(x)-\Fx(x)\srb \\
 =&\ \sqrt{\dfrac{mn}{N}}\dfrac{\Fmy(x)-\Fy(x)-(\Fmx(x)-\Fx(x))}{\vartheta_{m,n}(x)^{1/2}} -\dfrac{\sqrt{\lambda }\mathbb{V}_2(\Fy(x))-\sqrt{1-\lambda }\mathbb{V}_1(\Fx(x))}{\vartheta_{m,n}(x)^{1/2}}\\
 &\ +\lb \sqrt{\lambda }\mathbb{V}_2(\Fy(x))-\sqrt{1-\lambda }\mathbb{V}_1(\Fx(x))\rb\lb\dfrac{1}{\vartheta_{m,n}(x)^{1/2}}-\dfrac{1}{\vartheta_0(x)^{1/2}}\rb.
 \end{align*}
Because   $\Fx$ and $\Fy$ are bounded away from $0$ and $1$ on $D$, there exist $c>0$ and $c'<1$ such that $F(x),G(x)\in(c,c')$ for $x\in D$.
 We see that \eqref{intheorem: emp:  as} and  $m/N\to\lambda$ imply on $A$, the following hold:
 \begin{equation}\label{intheorem: bound : vn}
 \limsup_{m,n\to\infty}\|\vartheta_{m,n}\|_{D}, \|\vartheta_{0}\|_{D}<c'(1-c),
 \end{equation}
 \begin{equation}\label{intheorem bound: inv vn}
 \limsup_{m,n\to\infty}\|\vartheta_{m,n}^{-1/2}\|_{D}, \|\vartheta_{0}^{-1/2}\|_{D}<c^{-1/2}(1-c')^{-1/2}.
 \end{equation}
Since on the probability one set $A$, $\vartheta_{m,n}$ converges to $\vartheta_0$ uniformly, and both functions are bounded below on $D$, the following also holds:
 \begin{equation*}
 \lim_{m,n\to\infty}\|\vartheta_{m,n}^{-1/2}-\vartheta_0^{-1/2}\|_D= 0.
 \end{equation*}
 Therefore, 
 \begin{align*}
 \MoveEqLeft\limsup\limits_{m,n\to\infty}\norm{ \vartheta_{m,n}^{-1/2}\lbs \sqrt{\dfrac{mn}{N}}\lb\Fmy-\Fy-(\Fmx-\Fx)\rb-\sqrt{\lambda }\mathbb{V}_2(\Fy)+\sqrt{1-\lambda }\mathbb{V}_1(\Fx)\rbs}_{D}\\
 \leq & \limsup\limits_{m,n\to\infty}\|\vartheta_{m,n}^{-1/2}\|_{D}\\
 &\ \times\limsup\limits_{m,n\to\infty}\norm{\sqrt{\dfrac{mn}{N}} \lb\Fmy-\Fy-(\Fmx-\Fx)\rb-\sqrt{\lambda }\mathbb{V}_2(\Fy)+\sqrt{1-\lambda }\mathbb{V}_1(\Fx)}_{D},
 \end{align*}
 which,  since $\limsup\limits_{m,n\to\infty}\|\vartheta_{m,n}^{-1/2}\|_{D}$ is bounded, converges to $0$ on $A$ since  $m/N\to\lambda$ and  \eqref{intheorem: emp: weak} holds. Also on $A$, the Brownian bridges $\mathbb V_1$ and $\mathbb V_2$ have continuous trajectories,  which indicates $\sqrt{\lambda }\mathbb{V}_2(\Fy)-\sqrt{1-\lambda }\mathbb{V}_1(\Fx)$
 is a continuous function, and hence bounded on $D$.  Therefore, on $A$,
 \begin{align*}
 \MoveEqLeft\limsup_{m,n\to\infty} \norm{\lb \sqrt{\lambda }\mathbb{V}_2(\Fy(x))-\sqrt{1-\lambda }\mathbb{V}_1(\Fx(x))\rb\lb\dfrac{1}{\vartheta_{m,n}(x)^{1/2}}-\dfrac{1}{\vartheta_0(x)^{1/2}}\rb}_D\\
 \leq &\ \norm{ \sqrt{\lambda }\mathbb{V}_2(\Fy)-\sqrt{1-\lambda }\mathbb{V}_1(\Fx)}_{D} \limsup_{m,n\to\infty}\|\vartheta_{m,n}^{-1/2}-\vartheta_0^{-1/2}\|_D
 \end{align*}
 equals zero, which completes the proof.
 \end{proof}

  The second lemma, which is required for proving part (B), relies on the objects $s_1 : \tilde{t}\mapsto \left.\frac{d}{dt}\Fx\circ H^{-1}(t)\right|_{t=\tilde{t}}$, $s_2 : \tilde{t}\mapsto \left.\frac{d}{dt}\Fy\circ H^{-1}(t)\right|_{t=\tilde{t}}$, and 
  \begin{equation}\label{def: L0}
\mathbb{L}_0(t)=(1-\lambda )\lb\lambda ^{-1/2} s_2(t)\mathbb{V}_1(F\circ H^{-1}(t))-(1-\lambda )^{-1/2}s_1(t)\mathbb{V}_2(G\circ H^{-1}(t))\rb.
\end{equation}
Note that
\begin{align}
    \label{def: s1 and s2}
&\  s_1(t)=\dfrac{f\circ H^{-1}(t)}{\lambda f\circ H^{-1}(t)+(1-\lambda)g\circ H^{-1}(t)}\nn\\
    \quad\text{and}&\ \quad s_2(t)=\dfrac{g\circ H^{-1}(t)}{\lambda f\circ H^{-1}(t)+(1-\lambda)g\circ H^{-1}(t)}.
\end{align}
Because $\mathbb V_1$ and $\mathbb V_2$ are independent Brownian bridges, it is immediate that $\mathbb L_0$ is a Gaussian process.

  We now state a lemma that concerns the two sample empirical process $\sqrt{N}( \Fmx\circ \Hm^{-1}(t)-\Fx\circ H^{-1}(t))$.
  \begin{lemma}\label{lem:pykeshorack}
  Under the setting of Theorem~\ref{thm: dist based: null distribution: B},
    \begin{equation}\label{intheeorwm: result: pyke}
    \sup_{t\in[p',1-p']}\left|\frac{\sqrt{N}\slb \Fmx\circ \Hm^{-1}(t)-\Fx\circ H^{-1}(t)\srb}{[t(1-t)]^{1/2}}-\dfrac{\mathbb{L}_0(t)}{[t(1-t)]^{1/2}}\right|\to_p 0
\end{equation}
for any $p'\in(0,1/2)$.
Moreover, $\mathbb{L}_0$ has  continuous trajectories almost surely.
  \end{lemma}
  \begin{proof}[Proof of Lemma~\ref{lem:pykeshorack}]
  Theorem~4.1 and Corollary~4.1 of \cite{pyke1968} indicate \eqref{intheeorwm: result: pyke}, where here we emphasize that \eqref{def: L0} represents the corrected formula for $\mathbb{L}_0$, given by (30) of \cite{LW2012}, rather than the original formula for this quantity given in (3.8) of \cite{pyke1968}. 

  Now observe that, because $h>0$ on $D_p$ under Condition \ref{assump: Nonparametric}, $H$ is strictly increasing, which implies that $H^{-1}$ is a continuous function.
 Since $f$, $g$ and $H^{-1}$ are continuous, \eqref{def: s1 and s2} implies that $s_1$ and $s_2$ are both continuous.
 Therefore, from \eqref{def: L0} and the fact that $\mathbb V_1$ and $\mathbb V_2$ have continuous trajectories almost surely,  it is not hard to see that $\mathbb{L}_0$ is  continuous almost surely.  
  \end{proof}
  Our next lemma characterizes the Gaussian process  $\mathbb L_0$ on the set $C_p$ when $C_p\subset \iint(D_p)$.
\begin{lemma}
\label{Dist of L0: under null}
Suppose $F$ and $G$ are as in Theorem \ref{thm: dist based: null distribution: B} and $C_p\subset \iint(D_p)$. Then the Gaussian process 
\[\lbs \sqrt{\frac{\lambda}{1-\lambda}}\mathbb L_0(t): t\in H(C_p)\rbs\] is distributed as $\{\mathbb U(t):t\in H(C_p)\}$ where   $\mathbb U$ is a Brownian bridge. 
\end{lemma}

\begin{proof}[Proof of Lemma \ref{Dist of L0: under null}]
If $C_p=\emptyset$ then the statement is vacuously true. So we will assume that $C_p\neq\emptyset$. 
We first claim that if $C_p\subset\iint(D_p)$, then $s_1(t)=s_2(t)=1$ for all $t\in H(C_p)$. If the claim is true, then from \eqref{def: L0} it follows that 
\[\mathbb L_0(t)=(1-\lambda )\lb\lambda ^{-1/2} \mathbb{V}_1(F\circ H^{-1}(t))-(1-\lambda )^{-1/2}\mathbb{V}_2(G\circ H^{-1}(t))\rb\]
for $t\in H(C_p)$.
Also for $t\in H(C_p)$,
\[F(H^{-1}(t))=G(H^{-1}(t))=\lambda F(H^{-1}(t))+(1-\lambda)G(H^{-1}(t))=H(H^{-1}(t)).\]
Under the set up of Theorem \ref{thm: dist based: null distribution: B}, $H$ has positive and continuous density on a open neighborhood of $D_p$. Therefore, $H$ is strictly increasing on this open set, which implies $H^{-1}$ is also continuous on this open set. Hence, $H(H^{-1}(t))=t$ for all $t\in H(C_p)$, leading to
\[\mathbb{V}_1(F\circ H^{-1}(t))=\mathbb{V}_1(t)\quad\text{and}\quad \mathbb{V}_2(G\circ H^{-1}(t))=\mathbb{V}_2(t),\]
which implies
\[\mathbb L_0(t)=(1-\lambda)\lb\lambda ^{-1/2} \mathbb{V}_1(t)-(1-\lambda )^{-1/2}\mathbb{V}_2(t)\rb\quad\text{for all } t\in H(C_p).\]
Since $\mathbb V_1$ and $\mathbb V_2$ are independent Brownian bridge processes, it follows that the Gaussian process $\{\mathbb L'(t):t\in[0,1]\}$ defined by 
\[\mathbb L'(t)=(1-\lambda)\lb\lambda ^{-1/2} \mathbb{V}_1(t)-(1-\lambda )^{-1/2}\mathbb{V}_2(t)\rb\]
has variance
\[\text{var}(\mathbb L'(t))=(1-\lambda)^2\slb \lambda^{-1}+(1-\lambda)^{-1}\srb (t-t^2)=\frac{1-\lambda}{\lambda}(t-t^2).\]
In particular, it can be seen that $\sqrt{\lambda/(1-\lambda)}\mathbb L'$ is a Brownian bridge. Hence the proof follows if we can prove the claim that $s_1(t)=s_2(t)=1$ for all $t\in H(C_p)$.

 Since $(F,G)\in\bd(\Hm_0)$, Lemma~\ref{Lemma: geometry of H0} implies that $F-G$ attains maxima at $C_p$ whenever $C_p\neq\emptyset$. Because $F-G$ is continuously differentiable and $C_p\subset \iint(D_p)$, $f-g=0$ on $C_p$. Therefore,   $f(H^{-1}(t))=g(H^{-1}(t))$ for all $t\in H(C_p)$. The claim now follows from \eqref{def: s1 and s2}.
\end{proof}

\subsubsection{Proof of Lemma \ref{Dist of L0: at bd}}
\begin{proof}
[Proof of Lemma~\ref{Dist of L0: at bd}]
Since $t\in H(C_p)$, $H^{-1}(t)\in C_p$, implying
\[F(H^{-1}(t))=G(H^{-1}(t))=H(H^{-1}(t)).\]
Arguing as in the proof of  Lemma \ref{Dist of L0: under null}, we can show that  $H(H^{-1}(t))=t$ under the set up of Theorem \ref{thm: dist based: null distribution: B}.
Thus \eqref{def: L0} implies
\[\mathbb L_0(t)=(1-\lambda)\slb \lambda^{-1/2}s_2(t)\mathbb V_1(t)-(1-\lambda)^{-1/2}s_1(t)\mathbb V_2(t)\srb.\]
Because $\mathbb V_1$ and $\mathbb V_2$ are independent Brownian bridges, $\mathbb L_0$ is distributed as a centered normal variable with variance
\begin{align}
\label{def: var L0 t}
 \text{var}(\mathbb L_0(t))=&\  (1-\lambda)^2\slb \lambda^{-1}s_2(t)^2 \text{var}(\mathbb V_1(t))+ (1-\lambda)^{-1}s_1(t)^2\text{var}(\mathbb V_2(t))\srb\nn\\
    =&\ (1-\lambda)^2 t(1-t)\slb \lambda^{-1}s_2(t)^2 + (1-\lambda)^{-1}s_1(t)^2\srb
\end{align}
because
\[\text{var}(\mathbb V_1(t))=\text{var}(\mathbb V_2(t))=t(1-t).\]
Thus it follows that
\begin{align*}
  \sigma_{TSEP}^2= \frac{\lambda} {1-\lambda}\frac{\text{var}(\mathbb L_0(t))}{t(1-t)}=&\  \lambda s_1(t)^2+(1-\lambda)s_2(t)^2\\
    =&\ \frac{\lambda (f\circ H^{-1}(t))^2+(1-\lambda)(g\circ H^{-1}(t))^2}{\slb\lambda f\circ H^{-1}(t)+(1-\lambda)g\circ H^{-1}(t)\srb^2}
\end{align*}
where the last step follows from \eqref{def: s1 and s2}. This completes the proof of \eqref{def: sigma TSEP}. 

Next, we will establish the lower bound on $\sigma^2_{TSEP}$. 
Since $x\mapsto x^2$ is convex, by Jensen's inequality,
\begin{equation}
\label{jensen: L0}
    (1-\lambda)s_2(t)^2+\lambda_1 s_1(t)^2\geq \slb (1-\lambda)s_1(t)+\lambda s_1(t)\srb^2=1
\end{equation}
 since $\lambda s_1+(1-\lambda) s_2=1$ by
 \eqref{def: s1 and s2}. 
 Therefore, \eqref{def: var L0 t} implies
 \[ \text{var}(\mathbb L_0(t))\geq \frac{1-\lambda}{\lambda}t(1-t)\]
 which implies
 \begin{equation}
 \label{inlemma: lower bound: L0}
  1\leq    \frac{\lambda}{1-\lambda} \frac{\text{var}(\mathbb L_0(t))}{t(1-t)}=\sigma^2_{TSEP}.
 \end{equation}
 Since $x\mapsto x^2$ is strictly convex, the inequality in \eqref{jensen: L0} is an equality if and only if $s_1(t)=s_2(t)$. Because $f$ and $g$ are positive on $D_p$, \eqref{def: s1 and s2} implies that the latter occurs if and only if $f(H^{-1}(t))=g(H^{-1}(t))$.
 
 Now we will establish the upper bound on $\sigma^2_{TSEP}$.
 Because  $\lambda s_1+(1-\lambda) s_2=1$, we also have $s_2=(1-\lambda s_1)/(1-\lambda)$. Therefore, using \eqref{def: var L0 t}, we derive that $\text{var}(\mathbb L_0(t))$ equals
\begin{align}\label{inlemma: var L0}
 \MoveEqLeft (1-\lambda)^2 t(1-t)\lb \lambda^{-1}\lb\frac{1-\lambda s_1(t)}{(1-\lambda)} \rb^2 + \frac{s_1(t)^2}{1-\lambda}\rb \nn\\
 =&\ t(1-t)\lb \lambda^{-1}-2s_1(t)+\lambda s_1(t)^2+(1-\lambda) s_1(t)^2\rb\nn\\
 =&\ t(1-t)\slb \lambda^{-1}-2s_1(t)+s_1(t)^2\srb\nn\\
 =&\ t(1-t) \slb \lambda^{-1}-1+(1-s_1(t))^2\srb.
\end{align}
Note that \eqref{def: s1 and s2} implies $s_1(t)\in[0, \lambda^{-1}]$. On any interval, the convex function $\lambda^{-1}-2x+x^2$ attains maxima at either endpoints of the interval.
Therefore,
\begin{align*}
   \frac{\text{var}(\mathbb L_0(t))}{t(1-t)}\leq &\  \max\{\lambda^{-1}, \lambda^{-1}-1+(1-\lambda^{-1})^2\}\\
  =&\ \max\lbs \lambda^{-1}, \lambda^{-1}(\lambda^{-1}-1)\rbs.
\end{align*}
Therefore,
\begin{equation}
    \label{inlemma: upper bound: L0}
 \sigma^2_{TSEP}=   \frac{\lambda}{1-\lambda} \frac{\text{var}(\mathbb L_0(t))}{t(1-t)}\leq \max\{(1-\lambda)^{-1}, \lambda^{-1}\},
\end{equation}
which, combined with \eqref{inlemma: lower bound: L0}, completes the proof of part (A) of the current lemma.

Now if $s_1(t)=\lambda^{-1}$, then $ \sigma^2_{TSEP}=\lambda^{-1}$. Using \eqref{def: s1 and s2}, we can write
\[s_1(t)=\frac{1}{\lambda+ (1-\lambda) g\circ H^{-1}(t)/f\circ H^{-1}(t)}.\]
Thus $s_1(t)=\lambda^{-1}$ if and only if $g\circ H^{-1}(t)/f\circ H^{-1}(t)= 0$.
However, under our set up, $f$ and $g$ are positive on $D_p$, implying $g\circ H^{-1}(t)/f\circ H^{-1}(t)> 0$ for any $t\in[p,1-p]$. Regardless, since the function $x\mapsto 1/(\lambda+(1-\lambda)x)$ is right continuous at $0$, given any  $\e>0$, there exists $C_\lambda>0$, depending only on $\lambda>0$, so that if $g\circ H^{-1}(t)/f\circ H^{-1}(t)<C_\lambda$, then $s_1(t)>\lambda^{-1}-\e/2$. Suppose $\e$ is so small such that $\e/2<\lambda^{-1}-1$.  Then from \eqref{inlemma: var L0} it also follows that
\begin{align*}
     \sigma^2_{TSEP}\geq &\ \frac{\lambda}{1-\lambda}\slb \lambda^{-1}-1+(\lambda^{-1}-\e/2-1)^2\srb\\
     \geq &\ \frac{\lambda}{1-\lambda}\slb \lambda^{-1}-1+(\lambda^{-1}-1)^2-\e(\lambda^{-1}-1)\srb\\
     =&\ \lambda^{-1}-\e,
\end{align*}
which completes the proof of part B of the current lemma.

Similarly, we can show that if $s_1(t)=0$, then $\sigma^2_{TSEP}=(1-\lambda)^{-1}$. Using \eqref{def: s1 and s2} again, we can write 
\[s_1(t)=\frac{f\circ H^{-1}(t)/g\circ H^{-1}(t)}{\lambda f\circ H^{-1}(t)/g\circ H^{-1}(t)+ 1-\lambda }.\]
Therefore $s_1(t)=0$ if and only if $f\circ H^{-1}(t)/g\circ H^{-1}(t)=0$, which is impossible since $f,g>0$ on $D_p$ under our set up. However, since the map $x\mapsto x/(\lambda x+1-\lambda)$ is right continuous at $0$, given any  $\e>0$, we can find $C'_\lambda>0$, depending only on $\lambda>0$, so that if $f\circ H^{-1}(t)/g\circ H^{-1}(t)<C'_\lambda$, then $s_1(t)<(1-\lambda)\e/(2\lambda)$.   If $\e<2\lambda/(1-\lambda)$, then from \eqref{inlemma: var L0}, it also follows that
\begin{align*}
     \sigma^2_{TSEP}\geq &\ \frac{\lambda}{1-\lambda}\slb \lambda^{-1}-1+(1-(1-\lambda)\e/(2\lambda))^2\srb\\
     \geq &\ \frac{\lambda}{1-\lambda}\slb \lambda^{-1}-(1-\lambda)\e/\lambda\srb\\
     =&\ (1-\lambda)^{-1}-\e,
\end{align*}
which completes the proof of part C of the current lemma.
\end{proof}

  \subsubsection{Proof of part (A)}

First we will consider the case when $C_p\neq D_p$. The main steps of the proof are as follows:
   \begin{itemize}
  \item[(a)] We fix $\e>0$, and  choose some set $C_{p,m,n}(\sigma_\e)$ satisfying
  $ C_{p,m,n}(\sigma_\e)\subset D_{p,m,n}$. \textcolor{black}{Here $C_{p,m,n}(\sigma_\e)$ depends on  $\sigma_\epsilon$, which is a random  positive number that will be chosen appropriately.} Next we 
    partition $D_{p,m,n}$ as follows (see Figure~\ref{Fig: Cp, etc.}):
  \[D_{p,m,n}= \underbrace{ [D_{p,m,n}\setminus  C_{p,m,n}(\sigma_\e)]}_{B_{p,m,n}(\s_\e)}\cup \underbrace{[C_{p,m,n}(\sigma_\e)\setminus  C_{p,m,n}]}_{E_{p,m,n}(\s_\e)}\cup C_{p,m,n}.\]
  Therefore,   we can rewrite $T_{m,n}^{\text{min}}(\Fmx,\Fmy)$ as
  \begin{equation}\label{intheorem: def of T1mn}
  \min\lbs  \inf_{x\in B_{p,m,n}(\sigma_\e)}\omega_{m,n}(x),  \inf_{x\in E_{p,m,n}(\sigma_\e)}\omega_{m,n}(x),  \inf_{x\in C_{p,m,n}}\omega_{m,n}(x) \rbs.
  \end{equation}
  \item[(b)] We show that  on $A$,
  \[\liminf_{m,n\to\infty} \inf_{x\in B_{p,m,n}(\sigma_\e)}\omega_{m,n}(x)=\infty.\]
  \item[(c)] We show that the following holds on $A$:
  \[\liminf_{m,n\to\infty}\inf_{x\in E_{p,m,n}(\sigma_\e)}\omega_{m,n}(x)>\limsup_{m,n\to\infty}\inf_{x\in C_{p,m,n}}\omega_{m,n}(x)-\e.\]
  \item[(d)] Finally we show that on $A$,
  \[  \lim_{m,n\to\infty} \inf_{x\in C_{p,m,n}}\omega_{m,n}(x)= \inf_{x\in C_{p}}\omega_{0}(x),\]
from  which, we show that, \eqref{eq:asT1mnlim} follows.
\end{itemize}
We will restrict our attention only to the set $A$ for this part of the proof. However, because $P(A)=1$,  this serves our purpose.
\subsubsection*{Proof of step (a)}
 For $\sigma>0$ and $p\in[0,1]$, let us define
  \begin{equation}\label{def: C p sigma}
  C_{p}(\sigma)=\{x\in D_{p}\  :\ Dist(x,C_{p})<\sigma\}.
  \end{equation}
  Because $H$ is continuous, there exists  $p'\in(0,p)$ such that $H^{-1}(p')<H^{-1}(p)$, which implies $D_{p'}\supset D_p$. For any $\sigma>0$, we   define
  \begin{equation}\label{def: C p m n: sigma}
C_{p,m,n}(\sigma)=C_{p'}(\sigma)\cap D_{p,m,n}.
\end{equation}
The first task is to properly choose a  $\sigma_\e$ so that certain properties hold on $C_{p'}(\sigma_\e)$.

  Since $A\subset \Omega$ was chosen so that on this set $\mathbb V_1$ and $\mathbb V_2$ have continuous trajectories,  $\omega_0$ also has continuous trajectory, which is also
   uniformly continuous on $D_{p'}$ because the latter is a compact set. Hence, for each such trajectory, 
   \[\sigma_{\e}^1=\sup\lbs \sigma>0:  |x-y|<\sigma\text{  implies }|\omega_0(x)-\omega_0(y)|<\epsilon/2\quad\text{ for all }x,y\in D_{p'}\rbs\]
    is well defined and positive.  Note that $\sigma_{\e}^1$ is a random quantity, which can take the value $0$ on the set $A^c$ but $\sigma_{\e}^1>0$ on $A$. 
    
    On the other hand, since $C_{p}\neq D_p$ and $p'<p$, we have $C_{p'}\neq D_{p'}$. Therefore  there exists $x\in D_{p'}$ so that $G(x)-F(x)>0$. Suppose $\delta_p=(G(x)-F(x))/4$. Because  $G-F$ is continuous and $D_{p'}$ is compact,  $D_{p'}\setminus C_{p'}$  contains at least one interval where $G-F>2\delta_p$.
  Using the continuity of $G-F$, we  can choose $\sigma_2$ so small such that
  $ D_{p'}\setminus C_{p'}(\sigma_2)$ contains at least an interval where $G-F>\delta_p$, i.e. $ D_{p'}\setminus C_{p'}(\sigma_2)\neq \emptyset$. 
  We will take 
  \[\sigma_\e=\min(\sigma_\e^1,\sigma_2).\]
Since $\sigma_\e^1$ is random, $\sigma_\e$ is also random. Moreover,  $\sigma_\e>0$ on $A$.

Note that, $B_{p,m,n}(\sigma_\e)$ and $E_{p,m,n}(\sigma_\e)$ can be empty for small  $m$ and $n$. In that case, we define the infimum of $\omega_{m,n}$ over those set to be $\infty$. Also since 
  \[\inf_{x\in D_{p'}}\min\{f(x),g(x)\}>0,\] $H^{-1}$ is continuous on $p'$.
 
\subsubsection{Proof of step (b)}
For any $\sigma>0$, let us denote
$B_{p'}(\sigma)=D_{p'}\setminus C_{p'}(\sigma)$. We claim that $B_{p'}(\sigma_\e)$, is non-empty
which follows because  our choice of $\sigma_{\e}$ in step (a) implies  $B_{p'}(\sigma_\e)\supset B_{p'}(\sigma_2)$ where $ B_{p'}(\sigma_2)\neq\emptyset$ by definition of $\sigma_2$. By the continuity of $G-F$ it follows that there exists a random quantity $\delta_{\sigma_{\epsilon}}>0$ so that  $\Fy-\Fx>\delta_{\sigma_{\epsilon}}$ on  $B_{p'}(\sigma_\e)$.
Recall the definition of $\vartheta_{m,n}$ from \eqref{def: vartheta}.
   Since $f$ and $g$ are bounded away from $0$ on $D_{p'}$, Lemma~\ref{Lemma: Thm 2: 1}  implies that on $A$,
   \begin{align*}
 \inf_{x\in B_{p'}(\s_\e)} \omega_{m,n}(x)\geq \inf_{x\in B_{p'}(\s_\e)} \omega_{0}(x)+\sqrt{\dfrac{mn}{N}}\dfrac{\delta_p}{\sup_{x\in B_{p'}(\s_\e)}\vartheta_{m,n}(x)^{1/2}}+o(1).
   \end{align*}
 where the $o(1)$ term approaches zero as $m,n\to\infty$. However, by \eqref{intheorem bound: inv vn}, $\vartheta_{m,n}^{-1/2}$ is bounded above by constant depending only on $p'$, $F$, and $G$ on $D_{p'}$. 
   Noting that the continuous function $\omega_0$ is bounded on $D_{p'}$, and using  $m/N\to\lambda$, we deduce that 
   $$\liminf_{m,n\to\infty} \inf_{x\in B_{p'}(\s_\e)} \omega_{m,n}(x)=\infty$$
    on $A$. 
   Note that \eqref{def: C p m n: sigma} implies
   \[B_{p, m,n}(\s_{\e})=D_{p,m,n}\cap(C_{p'}(\s_\e)\cap D_{p,m,n})^c=D_{p,m,n}\setminus C_{p'}(\sigma_\e).\]
  On the set  $A$, we thus have $D_{p,m,n}\subset D_{p'}$ eventually as $m,n\to\infty$, which implies
   \[B_{p, m,n}(\s_{\e})=D_{p,m,n}\setminus C_{p'}(\sigma_\e)\subset D_{p'}\setminus C_{p'}(\sigma_\e)=B_{p'}(\sigma_{\e})\]
eventually as $m,n\to\infty$,   leading to
\begin{equation}\label{intheorem: contact: first minimum}
\liminf_{m,n\to\infty}\inf_{x\in B_{p, m,n}(\s_{\e})} \omega_{m,n}(x)=\infty.
\end{equation}

  \subsubsection*{Proof of step (c)} 
 We denote  the boundary of the set $C_p$ by $ \text{bd}(C_p)$. We can show that $C_p$ is a closed set, which implies $\text{bd}(C_p)\subset C_p$. Therefore, for all $y\in \text{bd}(C_p)$, we have $G(y)=F(y)$. Therefore, by Lemma~\ref{Lemma: Thm 2: 1}, for any $x\in  E_{p,m,n}(\s_{\e})$, and $y\in \text{bd}(C_p)\cap D_{p,m,n}$,   the following holds on $A$ for all sufficiently large $m$ and $n$:
   \begin{align*}
 \MoveEqLeft  \bl \omega_{m,n}(x)-\omega_{m,n}(y)-\sqrt{\dfrac{mn}{N}}\vartheta_{m,n}(x)^{-1/2}\lb G(x)-F(x)\rb \bl\\
   \leq &\ \e/2+\sup_{\substack{x\in E_{p,m,n}(\s_{\e}),\\ y\in \in \text{bd}(C_p)\cap D_{p,m,n}}}|\omega_0(x)-\omega_0(y)|\\
  \stackrel{(a)}{\leq} &\ \e/2+\sup_{|x-y|<\s_{\e},\ x,y\in D_{p'}}|\omega_0(x)-\omega_0(y)|,
   \end{align*}
   which, by our choice of $\s_\e$, is not larger than $\e/2$.  Here $(a)$ follows because $E_{p,m,n}(\s_{\e})\subset C_{p,m,n}(\s_{\e})\subset C_{p'}(\s_{\e})$ for sufficiently large $m$ and $n$.
  The above leads to
  \begin{align*}
  \inf_{x\in E_{p,m,n}(\s_{\e})}\omega_{m,n}(x)\geq &\  \sqrt{\dfrac{mn}{N}}\inf_{x\in E_{p,m,n}(\s_{\e})}\dfrac{\Fy(x)-\Fx(x)}{\vartheta_{m,n}(x)^{1/2}}\\
  &\ + \inf_{y\in D_{p,m,n}\cap \text{bd}(C_p)}\omega_{m,n}(y)-\e.
  \end{align*}
Since $(F,G)\in\text{bd}(\Hm_0)$, we have $G-F\geq 0$, which yields
  \[\liminf_{m,n\to\infty}\inf_{x\in E_{p,m,n}(\s_\e)}\omega_{m,n}(x)\geq \limsup_{m,n\to\infty}\inf_{y\in D_{p,m,n}\cap C_p}\omega_{m,n}(y)-\e.\] 
  Noting $C_{p,m,n}=D_{p,m,n}\cap C_p$, we conclude this step.
  
  \subsubsection*{Proof of step (d)}
  This step follows from 
Lemma~\ref{Lemma: Thm 2: 1}. To see this, note that, Lemma~\ref{Lemma: Thm 2: 1}  implies that $\|\omega_{m,n}-\omega_0\|_{C_{p'}}\to 0$ on $A$ as $m,n\to\infty$. Because $\Fx=\Fy$ on $C_{p'}$, it also follows that
\[\omega_0(x)=\dfrac{\sqrt{\lambda }\mathbb{V}_2(\Fx(x))-\sqrt{1-\lambda }\mathbb{V}_1(\Fx(x))}{\sqrt{\Fx(x)\slb 1-\Fx(x)\srb}}=\dfrac{\mathbb{U}\circ F(x)}{\sqrt{\Fx(x)\slb 1-\Fx(x)\srb}}\]
where $\mathbb U=\sqrt{\lambda}\mathbb V_2-\sqrt{1-\lambda}\mathbb{V}_1$ is a Brownian bridge.
Since on $A$, $D_{p,m,n}\subset D_p\subset D_{p'}$ for sufficiently large $m$ and $n$, it entails that $C_{p,m,n}\subset C_{p}\subset C_{p'}$ for sufficiently large $m$ and $n$ as well. Therefore,
\[\norm{\omega_{m,n}-\dfrac{\mathbb{U}\circ F}{\sqrt{\Fx(1-\Fx)}}}_{C_{p,m,n}}\to 0\]
as well.
  Since $Dist(C_{p,m,n},C_p)\to 0$ on $A$ by \eqref{conv: D and A}, the above readily yields that on $A$,
  \begin{align}\label{step d: convg}
   \lim_{m,n\to\infty} \inf_{x\in C_{p,m,n}}\omega_{m,n}(x)= \inf_{x\in C_p} \dfrac{\mathbb{U}\circ F(x)}{\slb \Fx(x)\slb 1-\Fx(x)\srb\srb^{1/2}}.
  \end{align}
Combining  steps (a)--(d) with \eqref{intheorem: def of T1mn} yields on $A$,
 \begin{align*}
  &\limsup_{m,n\to\infty} \inf_{x\in C_{p,m,n}}\omega_{m,n}(x)-\e\leq   \liminf_{m,n\to\infty}T_{m,n}^{\text{min}}(\Fmx,\Fmy)\\
  \leq &\  \limsup_{m,n\to\infty}T_{m,n}^{\text{min}}(\Fmx,\Fmy)\leq  \liminf_{m,n\to\infty}\inf_{x\in C_{p,m,n}}\omega_{m,n}(x).
 \end{align*}
 Letting $\e\to 0$, we have
 \[\lim_{m,n\to\infty}T_{m,n}^{\text{min}}(\Fmx,\Fmy)=  \lim_{m,n\to\infty}\inf_{x\in C_{p,m,n}}\omega_{m,n}(x)= \inf_{x\in C_p} \dfrac{\mathbb{U}\circ F(x)}{\sqrt{\Fx(x)\slb 1-\Fx(x)\srb}}\]
 on $A$, where the last step follows from \eqref{step d: convg}. 
The above concludes the proof of  \eqref{eq:asT1mnlim} when $C_p\neq D_p.$

Now suppose $C_p=D_p$. In this case, we will only use Step C and D. Let us restrict our attention to only $A$. We define  $\s_\e=\s_\E^1$.  Letting $C_{p'}(\s_\e)$ be as in \eqref{def: C p sigma}, we have 
$D_p\subset  \iint(C_{p'}(\s_\e))$,
and also, 
$D_{p,m,n}\subset  \iint(C_{p'}(\s_\e))$ for sufficiently large $m$ and $n$. Let us also denote $C_{p,m,n}$ as in \eqref{def: C p m n: sigma}
and  $E_{p,m,n}=C_{p,m,n}(\s_{\e})\setminus C_{p,m,n}$ as in step C. Then for large $m$ and $n$, the partition $D_{p,m,n}=E_{p,m,n}(\s_{\e})\cup C_{p,m,n}$ is valid. Therefore, the proof follows from combining Step C and D.

\subsubsection*{Proof of part (B) of Theorem~\ref{thm: dist based: null distribution: B}}
     We now prove part~(B). For the ease of reference, we let
 \[\nu_{m,n}(t)=\dfrac{\sqrt{\dfrac{mn}{N}}\lb\Fmy\circ\Hm^{-1}(t)-\Fmx\circ \Hm^{-1}(t)\rb}{[t(1-t)]^{1/2}}.\]
    Note that $T_{m,n}^{\text{tsep}}(\Fmx,\Fmy)=\inf_{t\in[p,1-p]} \nu_{m,n}(t)$.

    We start by studying the numerator of the above display. Noting that $N\Hm=m\Fmx+n\Fmy$, we derive that
\begin{align*}
\MoveEqLeft\sqrt{\dfrac{mn}{N}}\lb\Fmy\circ\Hm^{-1}(t)-\Fmx\circ \Hm^{-1}(t)\rb\\
= &\ -\sqrt{\dfrac{mN}{n}}\lb\Fmx\circ \Hm^{-1}(t)-t\rb+\sqrt{\dfrac{mN}{n}}\lb \Hm\circ\Hm^{-1}(t)-t\rb.
\end{align*}
Combining the fact that $\sup_{t\in[0,1]}\bl \Hm\circ\Hm^{-1}(t)-t\bl \leq 1/N$ \citep[p. 762 of][]{pyke1968} with the fact that   $m/N\to\lambda$, we obtain that
 \begin{equation}\label{inlemma: equivelence to pyke}
 \sup_{t\in[0,1]}\bl\sqrt{\dfrac{mn}{N}}\lb\Fmy\circ\Hm^{-1}(t)-\Fmx\circ \Hm^{-1}(t)\rb+\sqrt{\dfrac{mN}{n}}\lb\Fmx\circ \Hm^{-1}(t)-t\rb\bl\to 0
 \end{equation}
 with probability one.
The above readily shows that
    \[
 \sup_{t\in[p,1-p]}\left|\nu_{m,n}(t)+\sqrt{\dfrac{mN}{n}}\dfrac{\Fmx\circ \Hm^{-1}(t)-t}{[t(1-t)]^{1/2}}\right|\as 0.\]
Combining the above with \eqref{intheeorwm: result: pyke}, we see that
 \begin{equation*}
 \sup_{t\in[p,1-p]}\bl\nu_{m,n}(t)+\sqrt{\dfrac{\lambda }{1-\lambda }}\dfrac{\mathbb{L}_0(t)}{[t(1-t)]^{1/2}}+\sqrt{\dfrac{mN}{n}}\dfrac{(F\circ H^{-1}(t)-t)}{[t(1-t)]^{1/2}}\bl\to_p 0.
 \end{equation*}
Upon noting that
 \[F\circ H^{-1}(t)-t=F\circ H^{-1}(t)-H\circ H^{-1}(t)=(1-\lambda )\lb F\circ H^{-1}(t)-G\circ H^{-1}(t)\rb,\]
the preceding limit reduces to
\begin{equation}\label{intheorem: smthing smthing}
 \sup_{t\in[p,1-p]}\bl\nu_{m,n}(t)+\sqrt{\dfrac{\lambda }{1-\lambda }}\dfrac{\mathbb{L}_0(t)}{[t(1-t)]^{1/2}}-(1-\lambda )\sqrt{\dfrac{mN}{n}}\nu(t)\bl\to_p 0,
 \end{equation} 
 where
 \begin{equation}\label{def: nu}
 \nu(t)=[G\circ H^{-1}(t)-F\circ H^{-1}(t)]/[t(1-t)]^{1/2}.
 \end{equation}
If we take any subsequence of the random sequence on the left side of the above, we can find a further subsequence that approaches zero almost surely. Suppose  that we can show, along the latter subsequence, that $T_{m,n}^{\text{tsep}}(\Fmx,\Fmy)-\inf_{t\in H(C_p)}\mathbb{U}(t)/[t(1-t)]^{1/2}$ converges almost surely to zero. 
In light of the fact that the limit does not depend on the choice of sequence or subsequence, Theorem 5.7 of \cite{shorack2000} would then imply  that the whole sequence converges weakly to the same limit, namely zero. Since weak convergence to a constant is equivalent to convergence in probability to that constant, this would complete the proof. Therefore, in what follows, we use $m',n'$ to denote members of a subsequence along which \eqref{intheorem: smthing smthing} holds almost surely and set out to prove that, as $m',n'\rightarrow\infty$,
 \begin{equation}\label{intheorem: T2 convergence}
T_{2,m',n'}(\mathbb{F}_{m'},\mathbb{G}_{n'})+\sqrt{\frac{\lambda}{1-\lambda}}\inf_{t\in H(C_p)}\dfrac{\mathbb{L}_0(t)}{\sqrt{t(1-t)}}\as 0.
 \end{equation}
Hence we assume that there exists $A'\subset\Omega$ such that $P(A')=1$ and, as $m',n'\to\infty$, 
 \begin{equation}\label{intheorem: contact: as}
 \sup_{t\in[p,1-p]}\bl\nu_{m',n'}(t)+\sqrt{\dfrac{\lambda }{1-\lambda }}\dfrac{\mathbb{L}_0(t)}{[t(1-t)]^{1/2}}-(1-\lambda )\sqrt{\dfrac{m'N'}{n'}}\nu(t)\bl\to 0
 \end{equation} 
 on $A'$, where $N'=m'+n'$. We choose $A'$ so that  $\mathbb{L}_0$ has continuous trajectories on $A'$, which Lemma~\ref{lem:pykeshorack} shows is possible.
 
  The rest of the proof is similar to the proof of part (A) because the asymptotics of the infimum of $\nu_{m',n'}$ over $[p,1-p]$ are largely governed by its numerator. Indeed, replacing the denominator of \eqref{def: wmn} from part~(A) by the denominator of $[t(1-t)]^{1/2}$ for part~(B) changes little since this new denominator is also bounded away from $0$ on $[p,1-p]$.  Nonetheless, there are some differences, which we detail below.

  Fix $\epsilon>0$. We replace $\s_{\e}$ from the proof of part~(A) by $\s'_{\e}$, where we define $\s'_{\e}$ as follows. 
  If $t,t'\in[p,1-p]$ satisfy $|t-t'|<\s'_{\e}$, then, on $A'$,
  \begin{equation}\label{intheorem: contavt: def of sigma}
  \sqrt{\dfrac{\lambda }{1-\lambda }}\bl\dfrac{\mathbb{L}_0(t)}{[t(1-t)]^{1/2}}-\dfrac{\mathbb{L}_0(t')}{(t'(1-t'))^{1/2}}\bl<\e/2.
  \end{equation}
  Note that since $[p,1-p]$ is a compact set, and the function $\mathbb{L}_0$ has continuous trajectories on $A'$, the random quantity  $\sigma'_\e>0$ on $A'$. 
 
 Recalling $h=\lambda \fx+(1-\lambda )\fy$,
we let 
\begin{equation}\label{def: bp}
    b_p=\sup_{z\in [H^{-1}(p),H^{-1}(1-p)]} h(z),
\end{equation}
 which is clearly positive. Because $f$ and $g$ are continuous, $h$ is also continuous, and therefore $b_p<\infty$. Taking $\tilde\sigma_{\e}=\s'_{\e}/b_p$, and similar to \eqref{def: C p sigma}, defining
 \begin{equation*}
  C_{p}(z)=\{x\in D_{p'}\  :\ Dist(x,C_{p'})<z\},\quad \text{ for all }z>0,
  \end{equation*}
we observe that $[p,1-p]$ can be written as the union of the following three sets:
 \begin{align*}
 \tilde B_p(\tilde\sigma_{\e})&\ =\{t\in[p,1-p]\ :\ H^{-1}(t)\in D_p\setminus C_p(\tilde\sigma_{\e})\},\\
 \tilde E_p(\tilde\sigma_{\e}) &\ = \{t\in[p,1-p]\ :\ H^{-1}(t)\in C_p(\tilde\sigma_{\e})\setminus C_p\},\\
 H(C_p) &\ = \{t\in[p,1-p]\ :\ H^{-1}(t)\in C_p\}.
 \end{align*}
 
  Note that, the $D_{p'}$ used in the proof of part (a) of the current theorem is replaced by $D_p$ in the above partitioning.
 Part (B) differs from part (A) in that $T_{m,n}^{\text{tsep}}(\Fmx,\Fmy)$ is the infimum of a random quantity, namely $\nu_{m,n}$, over a fixed set $[t,1-t]$, whereas $T_{m,n}^{\text{min}}(\Fmx,\Fmy)$ calculates the  infimum of $\omega_{m,n}$ over a random set $D_{p,m,n}$. To deal with this randomness, the asymptotics in part (A) were analyzed on a set $D_{p'}\supset D_p$ constructed so as to ensure $D_{p'}\supset D_{p,m,n}$ for sufficiently large $m$ and $n$ almost surely. Since part (B) does not have this additional difficulty, it suffices to study the behavior of $\nu_{m,n}$ on $D_p$, circumventing the need to consider $D_{p'}$.

We now show that 
 \begin{equation}\label{intheorem: contact: ttp inequality}
 |t-t'|\leq \s'_{\e}\quad\textnormal{for all $t\in \tilde E_p(\tilde\sigma_{\e})$ and $t'\in \text{bd}(H(C_p))$.}
 \end{equation}
To prove this, we first fix $t\in \tilde E_p(\tilde\sigma_{\e})$ and $t'\in \text{bd}(H(C_p))$. Because $H^{-1}$ is continuous, it holds that $H^{-1}(t')\in\text{bd}(C_p)$. Thus, $|H^{-1}(t)-H^{-1}(t')|\leq \tilde\sigma_{\e}$.  Note that the continuity of $H$ also implies that $H(H^{-1}(t_0))=t_0$ for all $t_0\in [p,1-p]$. By the mean value theorem applied to the function $H$, there exists an $a$ between $H^{-1}(t)$ and $H^{-1}(t')$ such that $t-t' = h(a)[H^{-1}(t)-H^{-1}(t')]$. Hence, $|t-t'|\le b_p |H^{-1}(t)-H^{-1}(t')|$ where $b_p$ is as defined in \eqref{def: bp}. Combining this display with the fact that $|H^{-1}(t)-H^{-1}(t')|\leq \tilde\sigma_{\e}$ and plugging in  $\tilde\sigma_{\e}=\s'_\e/b_p$ shows that \eqref{intheorem: contact: ttp inequality} indeed holds.

Recall the definition of $\nu(t)$ from \eqref{def: nu}. Because $\nu(t)\geq 0$, one can show that the infimum of  $\nu$ over the random set $\tilde B_p(\tilde\sigma_{\e})$ is  bounded below by some random number $\delta_{\tilde\sigma_\e}>0$. Therefore, using \eqref{intheorem: contact: as} and imitating the proof of \eqref{intheorem: contact: first minimum},  we can show that
 \begin{align}\label{intheorem: contact: part 2: first convgnce}
 \liminf_{m',n'\to\infty}\inf_{t\in \tilde B_p(\tilde\sigma_{\e})}\nu_{m',n'}(t)\to\infty\quad\textnormal{ on $A'$.}
 \end{align}
 
 Next let us consider $t\in \tilde E_p(\tilde\sigma_{\e})$ and $t'\in \text{bd}(H(C_p))$. Note that \eqref{intheorem: smthing smthing} and \eqref{intheorem: contact: ttp inequality} imply that  the following holds on $A'$:
 \begin{align*}
\MoveEqLeft \bl \nu_{m',n'}(t)-\nu_{m',n'}(t')-(1-\lambda )\sqrt{\dfrac{m'N'}{n'}}\nu(t)\bl\\
\leq &\ o(1)-\sqrt{\dfrac{\lambda }{1-\lambda }}\sup_{t,t'\in S_t}\bl\dfrac{\mathbb{L}_0(t)}{[t(1-t)]^{1/2}}-\dfrac{\mathbb{L}_0(t')}{(t'(1-t'))^{1/2}}\bl,
\end{align*}
 where $S_t=\{t,t'\in[p,1-p]\ :\ |t-t'|<\s'_{\e}\}$. Equation~\ref{intheorem: contavt: def of sigma} yields that
 \[\sqrt{\dfrac{\lambda }{1-\lambda }}\sup_{t,t'\in S_t}\bl\dfrac{\mathbb{L}_0(t)}{[t(1-t)]^{1/2}}-\dfrac{\mathbb{L}_0(t')}{(t'(1-t'))^{1/2}}\bl<\e/2\]
on $A'$, which indicates that, on this set, for all sufficiently large $m'$ and $n'$,
 \[\bl\nu_{m',n'}(t)-\nu_{m',n'}(t')-(1-\lambda )\sqrt{\dfrac{m'N'}{n'}}\nu(t)\bl<\e.\]
 Because the above holds for any $t\in \tilde E_p(\tilde\sigma_{\e})$ and $t'\in \text{bd}(H(C_p))$, we obtain that
\[
\inf_{t\in \tilde E_p(\tilde\sigma_{\e})}\nu_{m',n'}(t)\geq \inf_{t'\in H(C_p)}\nu_{m',n'}(t')+\sqrt{\dfrac{m'N'}{n'}}\inf_{t\in \tilde E_p(\tilde\sigma_{\e})}\nu(t)-\e.
\]
The fact that $\nu$ is non-negative on $[p,1-p]$ yields that $\inf_{t\in \tilde E_p(\tilde\sigma_{\e})}\nu(t)\geq 0$.
Thus,  the above shows that $\inf_{t\in \tilde E_p(\tilde\sigma_{\e})}\nu_{m',n'}(t)\geq \inf_{t\in H(C_p)}\nu_{m',n'}(t)-\e$. Therefore,  using \eqref{intheorem: contact: part 2: first convgnce}, we derive that on $A'$, for sufficiently large $m'$ and $n'$,
\[\bl\inf_{t\in [p,1-p]}\nu_{m',n'}(t)- \inf_{t\in H(C_p)}\nu_{m',n'}(t)\bl\leq \e.\]
As $\epsilon>0$ was arbitrary, the above shows that on $A'$, $\inf_{t\in [p,1-p]}\nu_{m',n'}(t)- \inf_{t\in H(C_p)}\nu_{m',n'}(t)$ converges to zero as $m',n'\rightarrow\infty$. 
Finally, \eqref{def: L0} and  \eqref{intheorem: contact: as} yield that on $A'$,
\[\limsup_{m',n'\to\infty}\sup_{t\in H(C_p)}\bl\nu_{m',n'}(t)-\sqrt{\frac{\lambda}{1-\lambda}}\dfrac{\mathbb{L}_0(t)}{\sqrt{t(1-t)}}\bl= 0.\]
Recall that $A'$ was chosen to satisfy $P(A')=1.$ Thus,  the above convergence holds with probability one, from which, \eqref{intheorem: T2 convergence}
follows. As was discussed above \eqref{intheorem: T2 convergence}, the fact that this equation holds completes the proof of part~(B).  Part (C) follows from Lemma~\ref{Dist of L0: under null}.
\hfill $\Box$

 The proof of Theorem \ref{thm:critical values} for the minimum t-statistic can be found in \cite{whang2019}  \citep[see also][]{davidson2013}. However, we still include it here for the sake of completeness.
\subsubsection*{Proof of Theorem \ref{thm:critical values}}

If $(F,G)\in \Hm_1$, we have $G(x)>F(x)$ for all $x\in D_p$. Since $D_p$ is compact, the continuous function $G-F$ attains its minima at some $x_0\in D_p$. Therefore, it follows that
  \[\inf_{x\in D_p}(G(x)-F(x))>3\delta\] 
  for some $\delta>0$. We will show that this implies that 
  \begin{equation}\label{inlemma: Theorem 2: plain diff of ecdf}
 \liminf_{m,n\to\infty} \inf_{x\in D_{p,m,n}}\lb\Fmy(x)-\Fmx(x)\rb>\delta
  \end{equation}
  with probability one. 
 As a result, $T_{m,n}^{\text{min}}\to_p\infty$ follows because \eqref{inlemma: Theorem 2: plain diff of ecdf}  indicates that with probability one, for all $x\in D_{p,m,n}$,
\[\dfrac{\Fmy(x)-\Fmx(x)}{\sqrt{\dfrac{\Fmx(x)(1-\Fmx(x))}{m}+\dfrac{\Fmy(x)(1-\Fmy(x))}{n}}}\geq \dfrac{\delta}{\sqrt{\dfrac{\Fmx(x)(1-\Fmx(x))}{m}+\dfrac{\Fmy(x)(1-\Fmy(x))}{n}}}\]
 for all large $m$ and $n$.
 However, the right hand side of the last display is bounded below by
 $\sqrt{mn/N}\delta$. Since $m/N\to\lambda$, it follows that $T_{m,n}^{\text{min}}\to_p\infty$. Hence, it suffices to prove \eqref{inlemma: Theorem 2: plain diff of ecdf}. To this end, note that, since $\max(f,g)>0$ on an open neighborhood of $D_p$,  $H=\lambda F+(1-\lambda)G$ has positive density  on this neighborhood. Therefore, $H$  is a continuous and strictly increasing function on this neighborhood. Therefore $H^{-1}$ is continuous and strictly increasing in this neighborhood as well. Hence, we can choose $p'<p$ such that $H^{-1}(p')<H^{-1}(p)$, and
  \[\inf_{x\in D_{p'}}(G(x)-F(x))>2\delta.\]
  For all sufficiently large $m$ and $n$,  $D_{p,m,n}\subset D_{p'}$ almost surely  by Fact~\ref{fact: quantile}.
  Therefore, with probability one,
  \begin{equation}\label{intheorem: power}
 \inf_{x\in D_{p,m,n} }\lb \Fy(x)-\Fx(x)\rb\geq 2\delta\quad 
  \end{equation}
  as $m,n\to\infty$. 
 On the other hand, note that
 \begin{align}\label{intheorem: type 1 error}
\MoveEqLeft\sup_{x\in D_{p,m,n}}\bl \lb\Fmy(x)-\Fmx(x)\rb-\lb \Fy(x)-\Fx(x)\rb\bl\nn\\
 \leq &\ \sup_{x\in\RR}\bl \Fmy(x)-\Fy(x)\bl+\sup_{x\in\RR}\bl\Fmx(x)-\Fx(x)\bl,
 \end{align}
 which converges to $0$ almost surely. Therefore, \eqref{inlemma: Theorem 2: plain diff of ecdf} follows, which completes the proof for $T_{m,n}^{\text{min}}$.


 For $T_{m,n}^{\text{tsep}}$, note that \eqref{inlemma: Theorem 2: plain diff of ecdf}  implies
  \[T_{m,n}^{\text{tsep}}(\Fmx,\Fmy)\geq \sqrt{\frac{mn}{N}}\inf_{x\in D_{p,m,n}}\frac{\Fmy(x)-\Fmx(x)}{\sqrt{t(1-t)}}>\sqrt{\dfrac{mn}{N}}\delta\inf_{t\in[p,1-p]}(t(1-t))^{-1/2}, \]
 which diverges to $+\infty$,  thus completing the proof.  \hfill $\Box$

\subsection{Proofs for the shape-constrained test statistics}
\label{app: shpe constrained tests}
Before going into the proof for the shape-constrained test statistics, 
we state and prove a useful lemma.

 \begin{lemma}\label{lemma: weak convergence of dist of unimodal density functions}
   Suppose that $\fx$ is a unimodal density satisfying Condition~\ref{Cond A}. Let $\fnx$ be the unimodal density estimator of \citeauthor{birge1997},  based on the independent observations $X_1,\ldots,X_m$ with density  $\fx$. Here we take  $\eta=o(m^{-1})$, where $\eta$ is the tuning parameter  in Section \ref{sec: density estimation}. Denote by $\Fnx$ the distribution function of $\fnx$.  Further suppose that $\fnxm$ is the Grenander estimator of $\fx$ based on the true mode $M$. 
    Then the following assertions hold:
  \begin{itemize}
  \item[(A)] $ \sqrt{m}\rint|\fnx(x)-\fnxm(x)|dx\as 0.$
  \item[(B)] 
  $\sqrt{m}\|\Fnx-\Fmx\|_{\infty}\as 0.$
 \item[(C)]   $\norm{\sqrt{m}(\Fnx-\Fx)-\mathbb{V}_1\circ\Fx}_{\infty}\as 0$,
 where $\mathbb{V}_1$ is as defined in \eqref{intheorem: convg of empirical process : F}.
  \end{itemize}   
   \end{lemma}
 \begin{proof}
   Suppose that $M$ is the true mode of the density $\fx$.  In this case $\Fx$ can be written as \citep{prokasa} 
  \[\Fx=\alpha \Fx^{+}+(1-\alpha)\Fx^{-},\]
  where $\alpha=P_F(X\leq M)$, and $\Fx^{+}$ and $\Fx^{-}$ are the conditional distributions on $(-\infty,M]$ and  $[M,\infty)$, respectively, i.e.
  \[F^{+}(x)=\frac{F(x)1[x\leq M]}{F(M)}\quad\text{ and }\quad F^{-}(x)=\frac{\slb F(x)-F(M)\srb 1[x> M]}{1-F(M)}.\]
  
  Let us denote the distribution function of $\fnxm$ by $\Fnxm$. From \cite{prokasa} it follows that $\Fnxm$ can be expressed as 
  \begin{equation}\label{inlemma: shape: S2 Grenander estimator}
      \Fnxm=\hat{\alpha}_m\Fnx^{0,+}+(1-\hat{\alpha}_m)\Fnx^{0,-}
  \end{equation}
  where $\hat{\alpha}_m$ is the sample proportion on $(-\infty.M],$ and $\Fnx^{0,+}$ and $\Fnx^{0,-}$ are the monotonoe Grenander estimates of $\Fx^{+}$ and $\Fx^{-}$, respectively. Denote by $\Fmx^{+}$ and $\Fmx^{-}$, respectively, the empirical distribution functions corresponding to the observations in $(-\infty,M]$ and $[M,\infty)$. Since $\Fx$ is continuous, the probability that $X_i=M$ for some $i\in\{1,\ldots,m\}$ is $0$. Hence, there is no ambiguity in the above definition of  $\Fnxm$.  Also,  the empirical distribution of the $X_i$' writes as 
  \[\Fmx=\hat{\alpha}_m\Fmx^{+}+(1-\hat{\alpha}_m)\Fmx^{-}.\] 
  
   It is well known that under Condition~\ref{Cond A}, the Grenander estimator $\Fnx^{0,+}$ satisfies $\|\sqrt{m}(\Fnx^{0,+}-\Fx^{+})-\mathbb{V}_1\circ\Fx^{+}\|_{\infty}\as 0$ and $\sqrt{m}\|\Fnx^{0,+}-\Fmx^{+}\|_\infty\as 0$ (see Theorem 2.1 of \cite{beare2017}, the original result dates back to \cite{kiefer}). Similar results hold for $\Fnx^{0,-}$ and $\Fmx^{-}$.

  Since $\hat{\alpha}_m\as\alpha$ with probability one, we conclude that
  \begin{align}\label{weak convg: Grenander}
  \sqrt{m}\|\Fnxm-\Fmx\|_{\infty}\leq &\ \hat{\alpha}_m\sqrt{m}\|\Fnx^{0,-}-\Fmx^{-}\|_{\infty}+(1-\hat{\alpha}_m)\sqrt{m}\|\Fnx^{0,+}-\Fmx^{+}\|_{\infty}\as 0.
  \end{align}
  To prove part (A) of the current lemma,  now we invoke Theorem 1 of \citeauthor{birge1997}, which states that
 \[\dfrac{\sqrt{m}}{2}\rint|\fnx(x)-\fnxm(x)|dx\leq \sqrt{m}\eta+\sqrt{m}\|\Fnxm-\Fmx\|_{\infty},\]
 where $\eta$ is   as defined in Section \ref{sec: density estimation}, which implies that, in our case, $\eta=O(1/m)$. This, combined with \eqref{weak convg: Grenander}, proves that the right hand side of the above display approaches $0$ almost surely. Thus part (A) of the current lemma is proved. 
 
 Now note that since
  $\sqrt{m}\|\Fnx-\Fmx\|_{\infty}\leq \sqrt{m}\|\Fnx-\Fnxm\|_{\infty}+\sqrt{m}\|\Fnxm-\Fmx\|_{\infty}$,
  and
 \begin{align*}
 \|\Fnx-\Fnxm\|_{\infty}\leq \rint|\fnx(x)-\fnxm(x)|dx,
 \end{align*}
part (B) of the current lemma follows
by part (A) and \eqref{weak convg: Grenander}. 

  Finally, part (C) of the current lemma follows by noting
  that
  \begin{align*}
 \MoveEqLeft \norm{\sqrt{m}(\Fnx-F)-\mathbb{V}_1\circ F}_{\infty}\leq   \norm{\sqrt{m}(\Fnx-\Fmx)}_{\infty}+ \norm{\sqrt{m}(\Fmx-F)-\mathbb{V}_1\circ F}_{\infty},
  \end{align*}
  which converges to $0$ as $m\to\infty$ by part (B) of the current lemma and \eqref{intheorem: convg of empirical process : F}.
\end{proof} 
 
 \subsubsection{Proof of Lemma \ref{theorem: convergence: Ti: unimodal }}
Let us define
 \[\widehat{\vartheta}_{m,n}(x)=\dfrac{n}{N}\Fnx(x)\slb 1-\Fnx(x)\srb +\dfrac{m}{N}\Fny(x)\slb 1-\Fny(x)\srb.\]
 Recalling the definition of $\vartheta_{m,n}$ from \eqref{def: vartheta}, we obtain that
 \begin{align}\label{inlemma: T1 hat}
\MoveEqLeft T_{m,n}^{\text{min}}(\Fnx,\Fny)-T_{m,n}^{\text{min}}(\Fmx,\Fmy)\nn\\
\leq &\ \lb\dfrac{mn}{N}\rb ^{1/2}\lbs\inf_{x\in D_{p,m,n}}\dfrac{\Fnx(x)-\Fny(x)}{\widehat \vartheta_{m,n}(x)^{1/2}}
-\inf_{x\in D_{p,m,n}}\dfrac{\Fmx(x)-\Fmy(x)}{\vartheta_{m,n}(x)^{1/2}}\rbs\nn\\
\leq &\ \lb\dfrac{mn}{N}\rb ^{1/2}\sup_{x\in D_{p,m,n}}\bl 
\dfrac{\Fnx(x)-\Fmx(x)-\slb \Fny(x)-\Fmy(x)\srb}{\vartheta_{m,n}(x)^{1/2}}\nn\\
&\ +(\Fnx(x)-\Fny(x))\lb\widehat\vartheta_{m,n}(x)^{-1/2}-\vartheta_{m,n}(x)^{-1/2}\rb\bl\nn\\
\leq &\ \lb\dfrac{mn}{N}\rb ^{1/2}\lbs \sup_{x\in D_{p,m,n}}|\Fnx(x)-\Fmx(x)|+\sup_{x\in D_{p,m,n}}|\Fny(x)-\Fmy(x)|\rbs\|\vartheta_{m,n}^{-1/2}\|_{D_{p,m,n}}\nn\\
&\ + \lb\dfrac{mn}{N}\rb ^{1/2}\sup_{x\in D_{p,m,n}}\bl\dfrac{\vartheta_{m,n}(x)^{1/2}-\widehat\vartheta_{m,n}(x)^{1/2}}{\vartheta_{m,n}(x)^{1/2}\widehat \vartheta_{m,n}(x)^{1/2}}\bl.
 \end{align}
 Since $\fx,\fy$ are bounded away from $0$ on an open interval that includes $D_p$, we can find $p'<p$ so that $\fx$ and $\fy$ are bounded away from $0$ on $D_{p'}$. 
 Lemma~\ref{Lemma: Thm 2: 1} then indicates that $\|\vartheta_{m,n}^{-1/2}\|_{D_{p'}}$ is bounded away from $0$ with probability one. 
 Since $Dist(D_{p,m,n},D_p)\as 0$ by \eqref{conv: D and A}, $D_{p,m,n}\subset D_{p'}$ almost surely for all sufficiently large $m$ and $n$. As a result,  $\|\vartheta_{m,n}^{-1/2}\|_{D_{p,m,n}}$ is also bounded with probability one.
 
 Therefore, using  $m/N\to\lambda$ and Part (B) of Lemma \ref{lemma: weak convergence of dist of unimodal density functions}, we conclude that
 \begin{equation}\label{inlemma: T1 hat decomposition}
 \lb\dfrac{mn}{N}\rb ^{1/2}\sup_{x\in D_{p,m,n}}\lbs |\Fnx(x)-\Fmx(x)|+ |\Fny(x)-\Fmy(x)|\rbs\|\vartheta_{m,n}^{-1/2}\|_{D_{p,m,n}}\as 0.
 \end{equation}
 Now, observe that we can write
 \begin{align*}
 \MoveEqLeft \dfrac{\vartheta_{m,n}(x)^{1/2}-\widehat\vartheta_{m,n}(x)^{1/2}}{\vartheta_{m,n}(x)^{1/2}\widehat \vartheta_{m,n}(x)^{1/2}}=
 \dfrac{\vartheta_{m,n}(x)-\widehat{\vartheta}_{m,n}(x)}{\vartheta_{m,n}(x)^{1/2}\widehat \vartheta_{m,n}(x)^{1/2}\slb \vartheta^{1/2}_{m,n}(x)+\widehat{\vartheta}^{1/2}_{m,n}(x)\srb}.
 \end{align*}
 Since $\Fmx,\Fnx,\Fmy,$ and $\Fny$ take values between $0$ and $1$, it follows that
 \[|\vartheta_{m,n}(x)-\widehat{\vartheta}_{m,n}(x)|\leq \dfrac{2n}{N}|\Fnx(x)-\Fmx(x)|+\dfrac{2m}{N}|\Fny(x)-\Fmy(x)|.\]
 Therefore, another application of Part (B) of Lemma \ref{lemma: weak convergence of dist of unimodal density functions} combined with the fact that $m/N\to\lambda$ entails that
 \[\lb\dfrac{mn}{N}\rb^{1/2}\sup_{x\in D_{p,m,n}}|\widehat\vartheta_{m,n}(x)-{\vartheta}_{m,n}(x)|\as 0.\]
 Since $\|\vartheta_{m,n}^{-1/2}\|_{D_{p,m,n}}$ is bounded on $D_{p,m,n}$ almost surely, the above implies that $\|\widehat \vartheta_{m,n}^{-1/2}\|_{D_{p,m,n}}$ is also bounded on $D_{p,m.n}$ almost surely.
 Hence, 
  \[\lb\dfrac{mn}{N}\rb^{1/2}\sup_{x\in D_{p,m,n}}\dfrac{\vartheta_{m,n}(x)-\widehat{\vartheta}_{m,n}(x)}{\vartheta_{m,n}(x)^{1/2}\widehat \vartheta_{m,n}(x)^{1/2}\slb \vartheta^{1/2}_{m,n}(x)+\widehat{\vartheta}^{1/2}_{m,n}(x)\srb}\as 0,\]
  which combined with \eqref{inlemma: T1 hat} and \eqref{inlemma: T1 hat decomposition}, implies that
  \[T_{m,n}^{\text{min}}(\Fnx,\Fny)-T_{m,n}^{\text{min}}(\Fmx,\Fmy)\as 0.\]
   Similarly one can show that $T_{m,n}^{\text{min}}(\Fmx,\Fmy)-T_{m,n}^{\text{min}}(\Fnx,\Fny)\as 0$,  leading to 
 \[|T_{m,n}^{\text{min}}(\Fmx,\Fmy)-T_{m,n}^{\text{min}}(\Fnx,\Fny)|\as 0.\]

Using part (B) of Lemma \ref{lemma: weak convergence of dist of unimodal density functions} and  $m/N\to\lambda$ in the second step, we also deduce that
\begin{align*}
\MoveEqLeft |T_{m,n}^{\text{tsep}}(\Fnx,\Fny)-T_{m,n}^{\text{tsep}}(\Fmx,\Fmy)|\\
\leq &\ \lb\dfrac{mn}{N}\rb ^{1/2}\dfrac{\sup_{x\in \RR}|\Fnx(x)-\Fmx(x)|+\sup_{x\in \RR}|\Fny(x)-\Fmy(x)|}{ \inf\limits_{z\in [p,1-p]}\sqrt{z(1-z)}},
\end{align*}
which converges to zero almost surely.

 It remains to prove that  $|T_{m,n}^{\text{wrs}}(\Fnx,\Fny)-T_{m,n}^{\text{wrs}}(\Fmx,\Fmy)|$ converges to $0$ almost surely. To this end, we first note that
$\xi(\Fx,\Fy)=\rint\Fx d\Fy$ is Hadamard differentiable  with respect to the norm $\|\cdot\|_{\infty}$ at every pair of distribution functions $(\Fx,\Fy)$ \citep[see  Section 5, pages 362 - 371][]{lehmann1975},  where the derivative at $(\Fx,\Fy)$ is given by
\[\dot{\xi}( \Fx,\Fy;\mu_X,\mu_Y)=\rint \mu_{\subx}d\Fy-\rint \mu_{\suby}d\Fx,\]
where $\mu_X:\RR\mapsto\RR$ and $\mu_Y:\RR\mapsto\RR$ are bounded continuous functions.
Observe that we can write
\begin{align*}
\MoveEqLeft \lb\dfrac{mn}{N+1}\rb^{1/2}\lb\xi(\Fnx,\Fny)-\xi(F,G)\rb\\
=&\ \lb\dfrac{mn}{N(N+1)}\rb^{1/2}\ \dfrac{\xi\lb F+N^{-1/2}\widehat{\Delta}_{m,F},G+N^{-1/2}\widehat{\Delta}_{n,G}\rb-\xi(F,G)}{N^{-1/2}},
\end{align*}
where
\[\widehat{\Delta}_{m,\Fx}=\sqrt{N}(\Fnx-F);\quad \widehat{\Delta}_{n,\Fy}=\sqrt{N}(\Fny-G).\]
Note that part (C) of Lemma \ref{lemma: weak convergence of dist of unimodal density functions} and the fact that  $m/N\to\lambda$ imply that as $m,n\to\infty$,
\[\norm{\widehat{\Delta}_{m,\Fx}-\lambda ^{-1/2}\mathbb{V}_1\circ \Fx}_{\infty}\as 0;\quad [\norm{\widehat{\Delta}_{n,\Fy}-(1-\lambda )^{-1/2}\mathbb{V}_2\circ \Fy}_{\infty}\as 0,\]
and $\sqrt{mn/\{N(1+N)\}}\to\sqrt{\lambda(1-\lambda)}$.
Therefore, the Hadamard differentiability of $\xi$ implies that
\[\bl \lb\dfrac{mn}{N(N+1)}\rb^{1/2}\ \dfrac{\xi\lb F+N^{-1/2}\widehat{\Delta}_{m,F},G+N^{-1/2}\widehat{\Delta}_{n,G}\rb-\xi(F,G)}{N^{-1/2}}-\mathbb{Y}\bl\as 0,\]
where $\mathbb{Y}$ is the random variable $\dot{\xi}( \mu_{\subx}, \mu_{\suby};\Fx,\Fy)$ with $\mu_{\subx}=\sqrt{1-\lambda }\mathbb{V}_1\circ F$, and $\mu_{\suby}=\sqrt{\lambda }\mathbb{V}_2\circ G$.
Thus, we have established
\[\bl\lb\dfrac{mn}{N+1}\rb^{1/2}\lb\xi(\Fnx,\Fny)-\xi(F,G)\rb-\mathbb{Y}\bl\as 0.\]
Similarly using \eqref{intheorem: convg of empirical process : F}, \eqref{intheorem: convg of empirical process : G}, and  $m/N\to\lambda$, one can show that
\[\bl\lb\dfrac{mn}{N+1}\rb^{1/2}\lb\xi(\Fmx,\Fmy)-\xi(F,G)\rb-\mathbb{Y}\bl\as 0.\]
Then the proof for $T_{m,n}^{\text{wrs}}(\Fnx,\Fny)$ follows noting
\begin{align*}
\MoveEqLeft 12^{-1/2}|T_{m,n}^{\text{wrs}}(\Fnx,\Fny)-T_{m,n}^{\text{wrs}}(\Fmx,\Fmy)|\\
\leq &\  \bl\lb\dfrac{mn}{N+1}\rb^{1/2}\lb\xi(\Fnx,\Fny)-\xi(F,G)\rb-\mathbb{Y}\bl\\
&\ +\bl\lb\dfrac{mn}{N+1}\rb^{1/2}\lb\xi(\Fmx,\Fmy)-\xi(F,G)\rb-\mathbb{Y}\bl,
\end{align*}
which converges to $0$ almost surely.  \hfill $\Box$

 \subsubsection{Proof of Lemma \ref{theorem: convergence: Ti: log-concave }}
 
 We can find $p'<p$ such that  $\fx$ and $\fy$ are positive on $D_{p'}$.
 Theorem $4.4$ of \citeauthor{2009rufi},  Condition~\ref{Cond: B1}, and Condition~\ref{Cond: B2} imply that 
  \[\sup_{z\in D_{p'}}|\Flx(z)-\Fmx(z)|=o_p(m^{-1/2}),\]
and
 \[\sup_{z\in D_{p'}}|\Fly(z)-\Fmy(z)|=o_p(n^{-1/2}).\]
 Recall the set $D_{p,m,n}$ defined in \eqref{def: Dp}.
  We note that,  for sufficiently large $N$,  $D_{p,m,n}\subset D_{p'}$ with probability one, 
 indicating
 \begin{equation}\label{intheorem: lc: convergence on compacta}
 \sup_{u\in D_{p,m,n}}|\Flx(u)-\Fmx(u)|+\sup_{x\in D_{p,m,n}}|\Fly(x)-\Fmy(x)|=o_p(m^{-1/2})+o_p(n^{-1/2}),
 \end{equation}
 which is $o_p(N^{-1/2})$ since  $m/N\to\lambda$. This result is similar to Lemma~\ref{lemma: weak convergence of dist of unimodal density functions}(B) for unimodal densities, which is critical to proving the asymptotic equivalence between $T_{m,n}^{\text{min}}(\Fnx,\Fny)$ and $T_{m,n}^{\text{min}}(\Fmx,\Fmy)$, and between $T_{m,n}^{\text{tsep}}(\Fnx,\Fny)$ and $T_{m,n}^{\text{tsep}}(\Fmx,\Fmy)$ in Lemma~\ref{theorem: convergence: Ti: unimodal }. 
 As a consequence, the rest of the proof will be nearly identical to the proof of Lemma~\ref{theorem: convergence: Ti: unimodal }. Hence, we only highlight the main steps of the proof.
 
Let us define
 \[\tilde {\vartheta}_{m,n}(x)=\dfrac{n}{N}\Flx(x)\slb 1-\Flx(x)\srb +\dfrac{m}{N}\Fly(x)\slb 1-\Fly(x)\srb.\]
 Recalling the definition of $\vartheta_{m,n}$ from \eqref{def: vartheta}, and proceeding like the proof of Lemma~\ref{theorem: convergence: Ti: unimodal }, we can prove a log-concave analogue of \eqref{inlemma: T1 hat}, that is
 \begin{align*}
\MoveEqLeft T_{m,n}^{\text{min}}(\Flx,\Fly)-T_{m,n}^{\text{min}}(\Fmx,\Fmy)\\
\leq &\ \lb\dfrac{mn}{N}\rb ^{1/2}\lbs\inf_{x\in D_{p,m,n}}\dfrac{\Flx(x)-\Fly(x)}{\tilde \vartheta_{m,n}(x)^{1/2}}
-\inf_{x\in D_{p,m,n}}\dfrac{\Flx(x)-\Fly(x)}{\vartheta_{m,n}(x)^{1/2}}\rbs\\
\leq &\ \lb\dfrac{mn}{N}\rb ^{1/2}\lbs \sup_{x\in D_{p,m,n}}|\Flx(x)-\Fmx(x)|+\sup_{x\in D_{p,m,n}}|\Fly(x)-\Fmy(x)|\rbs\|\vartheta_{m,n}^{-1/2}\|_{D_{p,m,n}}\\
&\ + \lb\dfrac{mn}{N}\rb ^{1/2}\sup_{x\in D_{p,m,n}}\bl\dfrac{\vartheta_{m,n}(x)^{1/2}-\tilde\vartheta_{m,n}(x)^{1/2}}{\vartheta_{m,n}(x)^{1/2}\tilde \vartheta_{m,n}(x)^{1/2}}\bl.
 \end{align*}
Since $\fx,\fy$ are bounded away from $0$ on an open interval that includes $D_p$, using Lemma~\ref{Lemma: Thm 2: 1}, we can show that  $\|\vartheta_{m,n}^{-1/2}\|_{D_{p,m,n}}$ is bounded with probability one.
Since  $m/N\to\lambda$, by \eqref{intheorem: lc: convergence on compacta} it follows that
 \begin{equation*}
 \lb\dfrac{mn}{N}\rb ^{1/2}\sup_{x\in D_{p,m,n}}\lbs |\Flx(x)-\Fmx(x)|+ |\Fly(x)-\Fmy(x)|\rbs\|\vartheta_{m,n}^{-1/2}\|_{D_{p,m,n}}=o_p(1).
 \end{equation*}
Note that
 \begin{align*}
 \MoveEqLeft \dfrac{\vartheta_{m,n}(x)^{1/2}-\tilde\vartheta_{m,n}(x)^{1/2}}{\vartheta_{m,n}(x)^{1/2}\tilde \vartheta_{m,n}(x)^{1/2}}=
 \dfrac{\vartheta_{m,n}(x)-\tilde{\vartheta}_{m,n}(x)}{\vartheta_{m,n}(x)^{1/2}\tilde \vartheta_{m,n}(x)^{1/2}\slb \vartheta^{1/2}_{m,n}(x)+\tilde{\vartheta}^{1/2}_{m,n}(x)\srb}.
 \end{align*}
 Analogous to the proof of lemma~\ref{theorem: convergence: Ti: unimodal },  using \eqref{intheorem: lc: convergence on compacta}, we can show that
 \[\lb\dfrac{mn}{N}\rb^{1/2}\sup_{x\in D_{p,m,n}}|\tilde\vartheta_{m,n}(x)-{\vartheta}_{m,n}(x)|=o_p(1).\]
 Since $\|\vartheta_{m,n}^{-1/2}\|_{D_{p,m,n}}$ is bounded almost surely, the above implies that $\|\tilde \vartheta_{m,n}^{-1/2}\|_{D_{p,m,n}}$ is also bounded almost surely.
 Therefore another application of \eqref{intheorem: lc: convergence on compacta} yields that
  \[T_{m,n}^{\text{min}}(\Flx,\Fly)-T_{m,n}^{\text{min}}(\Fmx,\Fmy)=o_p(1).\]
 Similarly one can show that $T_{m,n}^{\text{min}}(\Fmx,\Fmy)-T_{m,n}^{\text{min}}(\Flx,\Fly)$ is $o_p(1)$,   which implies $|T_{m,n}^{\text{min}}(\Flx,\Fly)-T_{m,n}^{\text{min}}(\Fmx,\Fmy)|$ converges to $0$ in probability. The proof of
\[|T_{m,n}^{\text{tsep}}(\Flx,\Fly)-T_{m,n}^{\text{tsep}}(\Fmx,\Fmy)|\to_p 0\]
 is analogous to the  proof of $|T_{m,n}^{\text{tsep}}(\Fnx,\Fny)-T_{m,n}^{\text{tsep}}(\Fmx,\Fmy)|=o_p(1)$ in Lemma~\ref{theorem: convergence: Ti: unimodal }.
  \hfill $\Box$
 
\section{ Proofs for Section \ref{sec: measure of discrep}}
\label{sec: appendix: measure of discrep}

Although the aim of the current section is to derive the asymptotic distribution of  $\hd^2(\fnx,\fny)$, we will prove a more general result on plug-in estimators of integrated functionals, which may be of independent interest. Theorem \ref{thm: Hellinger distance} then follows as a special case.  

 We keep using the notations and terminologies developed in Appendix~\ref{sec: appendix: tests}.
  Recall  that we defined the set of all densities on $\RR$ by $\mP$. Let $\mP_1\subset\mP$. Suppose that $T:\mP_1^2\mapsto\RR$ is a functional of the form
\begin{equation}\label{def: measure of discrep}
T(f,g)=\rint v\slb f(x),g(x)\srb dx,
\end{equation}
  where $v:\RR^2\mapsto\RR$ is a known function. In our case,  $T(f,g)$ equals $H(f,g)^2$,  leading to
  \[v\slb f(x),g(x)\srb=2^{-1}\slb\sqrt{f(x)}-\sqrt{g(x)}\srb^2.\]
  
  Now we provide a brief background on a needed concept, namely on influence functions.  Define the set of all densities on $\RR$ by $\mP$. Consider a functional $T:\mP^{2}\mapsto\RR$. Suppose that $f$ and $g$ belong to $\mathcal{P}$, and denote the corresponding distribution functions by $F$ and $G$, respectively.  
Suppose the functions  $x\mapsto \psi_f(x;f,g)$ and $x\mapsto \psi_g(x;f,g)$   satisfy the following display for all $f_1$ and $g_1$ in $\mathcal{P}$: 
  \begin{align}\label{def: influence functions}
  \pdv{}{ t}T(F+t(F_1-F),G)\bl_{t=0}=&\rint \psi_f(x;f,g)f_1(x)dx,\\
\pdv{}{ t}T(F,G+t(G_1-G))\bl_{t=0}=&\rint \psi_g(x;f,g)g_1(x)dx,
  \end{align}
  where above $F_1$ and $G_1$ represent the cumulative distribution functions corresponding to $f_1$ and $g_1$. Then $\psi_f$ and $\psi_g$ represent the influence functions of $T$ (respectively $F$- and $G$-almost surely unique) under the nonparametric model \citep[][p. 292]{vdv}.
 When $T(f,g)$ equals the Hellinger distance $\hd^2(f,g)$, it follows that
 \begin{align}
 \psi_f(x;f,g)&=2^{-1}\lb 1-\sqrt{\dfrac{g(x)}{f(x)}}- \hd^2(f,g)\rb  1_{\supp(f)}(x),\label{def: influence function: f} \\
  \psi_g(x;f,g)&=2^{-1}\lb 1-\sqrt{\dfrac{f(x)}{g(x)}}- \hd^2(f,g)\rb  1_{\supp(g)}(x).\label{def: influence function: g}
  \end{align}
  
  We have already mentioned in Section \ref{sec: measure of discrep}  that the Von Mises Expansion (VME)  plays a critical role in the proofs of this section. 
  We define the first order VME of $T$ in the same lines as \cite{robin2015}. Suppose that $T$ is Gatea\^{u}x differentiable, and the corresponding influence functions $\psi_f$ and $\psi_g$ (see \eqref{def: influence functions} in Section~\ref{sec: measure of discrep}) exist.  Then we say that $T:\mP_1^2\mapsto\RR$  has a first order VME  if it satisfies the following  for all $f_1,f_2,g_1,g_2\in\mP_1$:
  \begin{align}\label{expansion: VME}
T(f_2,g_2)=&\ T(f_1,g_1)+\rint  \psi_f(x;f_1,g_1)f_2(x)dx+\rint \psi_g(x;f_1,g_1)g_2(x)dx\nonumber\\
&\ +O(\|f_1-f_2\|_2^2)+O(\|g_1-g_2\|_2^2).
\end{align}   
  The first order VME implies that $T$ can be written as a linear term plus second order bias term, i.e. $T$ is sufficiently smooth.  \cite{robin2015} gives examples of many  $T$ which has first order VME.
   
   Let $\hf$ and $\hg$ be estimators of $f$ and $g$ based on samples of size $m$ and $n$, respectively. We denote the corresponding distribution functions by $\hF$ and $\hG$. We aim to show that under some  regularity conditions,  the plug-in estimator $T(\hf,\hg)$ is $\sqrt{N}$-consistent for estimating $T(f,g)$. 
   
   The first condition we require is related to the weak convergence of the processes $ \sqrt{m}(\hF-F)$ and $ \sqrt{n}(\hG-G)$ to Brownian  processes.
\begin{cond}{C1}\label{supp: cond: C1}
The distribution functions $\hF$ and $\hG$ corresponding to density estimators $\hf$ and $\hg$ satisfy
 $\sqrt{m}(\hF-F)\to_d{\mathbb{V}}_1(F)$ and 
 $\sqrt{n}(\hG-G)\to_d{\mathbb{V}}_2(G)$,
 where ${\mathbb{V}}_1$ and ${\mathbb{V}}_2$ are Brownian bridges.
 \end{cond} 

 The second condition involves the order of the $L_2$ error in estimating 
 $f$ and $g$. In particular, we require $\|\hf-f\|^2_2$ and $\|\hg-g\|_2^2$ to be of order $o_p(m^{-1/2})$ and $o_p(n^{-1/2})$,  respectively.
 \begin{cond}{C2}\label{supp: Cond C2}
The density estimators $\hf$ and $\hg$ of $\fx$ and $\fy$ satisfy
 \begin{align}\label{condition: L2}
 O_p(\|\hf-f\|_2^2)=o_p(m^{-1/2})\quad \text{ and }\quad O_p(\|\hg-g\|_2^2)=o_p(n^{-1/2}).
 \end{align}
 \end{cond} 
 If the model is correctly specified, and $f$ is bounded, many density estimators $\hf$ are also bounded with high probability, leading to
 \[ O_p(\|\hf-f\|_2^2)=\hd^2(\hf, f)O_p(1).\]
Note that if $\hf$ also satisfies $\hd^2(\hf, f)=o_p(m^{-1/2})$,  Condition \ref{supp: Cond C2} follows. 
 
 { \color{black}
 Our next condition requires  the  influence functions $\psi_f(\cdot;f,g)$ and $\psi_g(\cdot;f,g)$ to be of bounded total variation on $\RR$. 
   We say a function $\mu:\RR\mapsto\RR$ is of bounded total variation on $\RR$,  if there exists a generalized derivative (in the sense of distribution) $\mu'$ of $\mu$ \citep[cf. Section 3.2 of][]{pallara2000} so that $\rint |\mu'(x)|dx<\infty$.
If $\mu$ is of bounded total variation on $\RR$, then $\mu$ is also of bounded variation on $\RR$. }
 \begin{cond}{I}\label{supp: cond: I}
 The maps
$x\mapsto\psi_f(x;f,g)$ and  $x\mapsto\psi_g(x;f,g)$
are of bounded total variation.
 \end{cond}
 
 \textcolor{black}{
 We are now ready to state the main theorem of this section, which gives the asymptotic distribution of $T(\hf,\hg)$ under the above-stated conditions. Later we will show that when $T(f,g)=\hd^2(f,g)$, the conditions are satisfied. Thus Theorem~\ref{thm: Hellinger distance} will follow as a corollary to Theorem~\ref{theorem: divergence: general}. Related literature \citep[cf.][]{robin2015} implies that the asymptotic variance of $T(\hf,\hg)$ as given by Theorem~\ref{theorem: divergence: general} agrees with the asymptotic lower bound for this case under the nonparametric model.  
  }
 
\begin{theorem}\label{theorem: divergence: general}
 Suppose  $\mP_1\subset\mP$. Let $T:\mP_1^2\mapsto\RR$ be a functional of the form \eqref{def: measure of discrep} satisfying the first order VME in \eqref{expansion: VME}.
Consider $f,g\in\mP_1$. We assume that the influence functions $\psi_f$ and $\psi_g$ defined in \eqref{def: influence functions} satisfy Condition \ref{supp: cond: I} .
 Let $\hf$ and $\hg$ be estimators of $f$ and $g$  based on two samples of size $m$ and $n$, respectively, where $m$ and $n$ satisfy  $m/N\to\lambda$.  Let us denote $N=m+n$. Further suppose  $\hf,\hg\in\mP_1$ satisfy Conditions \ref{supp: cond: C1} and \ref{supp: Cond C2}. 
 Then we have
 \[\sqrt{N}\lb T(\hf,\hg)-T(f,g)\rb \to_d N(0,\s^2_{f,g}),\]
 where
 \[\s^2_{f,g}=\lambda^{-1}\rint \psi_{f}(x;f,g)^2f(x)dx+(1-\lambda)^{-1}\rint \psi_{g}(x;f,g)^2g(x)dx.\]
  \end{theorem}

\begin{proof}
Since $T$ satisfies the first order VME, \eqref{expansion: VME} indicates that
 \begin{align*}
\MoveEqLeft T(\hf,\hg)-T(f,g)\\
=&\ 
\rint\psi_f(x;f,g)\hf(x)dx+\rint \psi_g(x;f,g)\hg(x)dx\\
&\ +O_p(\|\hf-f\|_2^2)+O_p(\|\hg-g\|_2^2)\\
=&\ \rint\psi_f(x;f,g)\hf(x)dx+\rint \psi_g(x;f,g)\hg(x)dx+o_p(N^{-1/2}),
\end{align*}
where the last step follows from Condition \ref{supp: Cond C2}. Denote by $F$, $\hF$, $G$, and $\hG$ the distribution functions corresponding to $f$, $\hf$, $g$, and $\hg$, respectively.
Since $\psi_f(x;f,g)$ is an influence function with respect to $f$, it satisfies
\[\rint \psi_f(x;f,g)f(x)dx=0.\]
Hence we can write
\[\rint\psi_f(x;f,g)\hf(x)dx=\rint \psi_f(x;f,g)d (\hF(x)-F(x)).\]
Now note that  $\psi_f(\cdot;f,g)$ is of bounded total variation on $\RR$ by Condition~\ref{supp: cond: I}.
 Therefore, integration by parts yields that
\begin{align*}
\MoveEqLeft\rint \psi_f(x;f,g)d (\hF(x)-F(x))\\
=&\ \psi_f(x;f,g)(\hF(x)-F(x))\bl _{-\infty}^{\infty}-\rint (\hF(x)-F(x))d\psi_f(x;f,g).
\end{align*}
The Riemann-Stieltjes integral in the second term on the right hand side of the last display exists because $\psi_f$ is of bounded total variation and $\hF-\Fx$ is continuous. 
Since $\psi_f(\cdot;f,g)$ is of bounded total variation,  it is also bounded, leading to
\[\lim_{x\to\pm\infty}\psi_f(x;f,g)(\hF(x)-F(x))=0.\]
Therefore, we deduce that
\[\rint \psi_f(x;f,g)\hf(x)dx=-\rint (\hF(x)-F(x))d\psi_f(x;f,g).\]
Similarly we can show that
\[\rint\psi_g(x;f,g)\hg(x)dx=-\rint(\hG(x)-G(x))d\psi_g(x;f,g).\]
Since $\hF$ and $\hG$ satisfy Condition \ref{supp: cond: C1}, it follows that
\[\lb\sqrt{m}(\hF-F),\sqrt{n}(\hG-G)\rb\to_d \lb{\mathbb{V}_1}(F),{\mathbb{V}_2}(G)\rb,\]
where $\mathbb{V}_1$ and $\mathbb{V}_2$ are independent standard Brownian bridges.
Here the underlying metric space corresponding to the weak convergence is $(l^{\infty},\|\cdot\|_{\infty})\times(l^{\infty},\|\cdot\|_{\infty}) $, where $l^{\infty}$  was defined to be the set of all bounded functions on $\RR$.
Since $m/N\to\lambda$,  Slutsky's Theorem yields
\[\lb\sqrt{N}(\hF-F),\sqrt{N}(\hG-G)\rb\to_d \lb \lambda^{-1/2}\mathbb{V}_1(F),(1-\lambda)^{-1/2}{\mathbb{V}_2}(G)\rb.\]
 Since  Condition \ref{supp: cond: I} holds, it follows that, for $\mu_1,\mu_2\in l^{\infty}$, the map
\[(\mu_1,\mu_2)\mapsto \rint \mu_1(x) \ d\psi_f(x;f,g)+\rint \mu_2(x) \ d\psi_g(x;f,g)\] 
is continuous with respect to the uniform metric $\|\cdot\|_{\infty}$. Therefore, invoking the continuous mapping theorem we obtain that
\begin{align*}
\MoveEqLeft \rint\sqrt{N}(\hF(x)-F(x))d\psi_f(x;f,g) +\rint\sqrt{N}(\hG(x)-G(x))d\psi_g(x;f,g)\\
 \to_d &\ \lambda^{-1/2}\rint {\mathbb{V}_1}(F(x))d\psi_f(x;f,g)
  +(1-\lambda)^{-1/2}\rint{\mathbb{V}_2}(G(x))d\psi_g(x;f,g).
\end{align*}
Now for any continuous distribution function $F$,  any Brownian bridge $\mathbb{V}$, and any function  $\mu$ with finite total variation, the random variable
\[\mathbb{Y}=\rint \mathbb{V}(F(x))d\mu(x)\sim N(0,\s_\mu^2),\]
where 
\[\s_\mu^2=\rint \mu(x)^2f(x)-\lb\rint \mu(x)f(x)dx\rb^2.\]
The above follows from the proof of Theorem $2.3$ of \cite{ramu2018}.
Therefore,
\begin{align*}
\lambda^{-1/2}\rint {\mathbb{V}_1}(F)d\psi_f(x;f,g)
  +(1-\lambda)^{-1/2}\rint{\mathbb{V}_2}(G)d\psi_g(x;f,g),
\end{align*}
which is distributed as a Gaussian random variable with variance $\s^2_{f,g}$, thus completing the proof.
 \end{proof}
 
 Now we focus on the special case at hand, i.e. $T(f,g)=\hd^2(f,g)$.  Towards this end, our first task is  to show the existence of the first order VME. We take $\mP_1$ to be $\mP(b,B)$, where
 f $\mP(b,B)$ is as defined in \eqref{def: bounded density class}. 
  
\begin{lemma}\label{lemma: VME}
Let $0<b<B<\infty$. Define the map $T:\mP(b,B)^2\mapsto\RR$  by
\[T(f,g)=\hd^2(f,g).\] 
 Then the first order VME in \eqref{expansion: VME} holds for $T$ for any $b,B>0$. 
\end{lemma}  
  \begin{proof}
  Follows from  Lemma 10 of \cite{robin2015}.
 \end{proof}
 
  Recall that we defined $\fnxm$  to be the Grenander estimator of $\fx$ based on the true mode of $\fx$. Denote by $\fnym$ the the Grenander estimator of $\fy$ based on the true mode of $\fy$.
 Our next step is to obtain the asymptotic distribution of $\hd^2(\fnxm,\fnym)$.  
  \begin{corollary}\label{thm: Hellinger: grenander}
 Let   $f$ and $g$ be continuous unimodal densities in $\mP(b,B)$ for some $b,B>0$. Suppose $f$ and $g$ satisfy condition~\ref{Cond A}. 
  We let $\fnxm$ and $\fnym$ be  the Grenander estimators of $f$ and $g$  based on the true modes, constructed from  samples of size $m$ and $n$, respectively. Suppose $m$ and $n$ satisfy  $m/N\to\lambda$. Then  \[\sqrt{N}(\hd^2(\fnxm,\fnym)-\hd^2(f,g))\to_d N(0,\s^2_{f,g}),\]
  where $\s^2_{f,g}$ is as in \eqref{def: sigma f g}.
  \end{corollary}

   \begin{proof}
   First we will show that the conditions of Theorem \ref{theorem: divergence: general} are satisfied. Then we will show that the $\s_{f,g}$ of Theorem \ref{theorem: divergence: general} takes the form of \eqref{def: sigma f g} when $T(f,g)=\hd^2(f,g)$.
   Suppose $M$ is the mode of $f$. Since $f(M)<B$, and $\fnxm$ satisfies \eqref{inlemma: shape: S2 Grenander estimator}, the behavior of $\fnxm$ at $M$ is similar to that of the Grenander estimator of a monotone density at its maxima. Therefore, using Corollary 1.2(i) of \cite{balabdaoui2009}   \citep[see also ][]{Woodroofe1993}, we obtain that
   $\fnxm(M)\to_d f(M)/\mathbb U$, where $\mathbb U\sim \text{Uniform}(0,1)$, which implies $\sup_{x\in\RR}\fnxm(x)$ is $O_p(1)$.
 On the other hand, with probability one, $\fnxm$ converges to $\fx$ uniformly over any interval of the from $[M+c,\infty)$ or $(-\infty,M-c]$ where $c>0$ \citep[cf.][]{balabdaoui2009}. Therefore it can be shown that
 \begin{equation}\label{inlemma: lower bound of grenander}
    P\slb \liminf_n\inf_{x\in\supp(\fnxm)}\fnxm(x)>b/2\srb=1.
 \end{equation}
  Similar results hold for $\fnym$ as well. Thus given any $\e>0$, we can find $B_\e>B$ so that 
  \[\liminf_{m\to\infty}P\slb\fnxm,\fnym\in\mP(b/2,B_\e)\srb>1-\e.\]
  Since $\mP(b,B)\subset \mP(b/2,B_\e)$, $f,g\in\mP(b/2,B_\e)$ as well. 
  Thus it suffices to show that the conditions of Theorem \ref{theorem: divergence: general} are satisfied when $\fnxm,\fnym\in\mP_1\equiv\mP(b/2,B_\e)$. 
  
  Notice that Lemma \ref{lemma: VME} implies that the first order VME holds for the functional $\hd^2:\mP_1\mapsto\RR$ when $\mP_1=\mP(b/2,B_\e)$.
 Condition \ref{supp: cond: I} also follows in a straightforward way once we note that, when $f,g\in\mP(b/2,B_\e)$, \eqref{def: influence function: f} and \eqref{def: influence function: g} indicate that $\psi_f(\cdot;f,g)$ and $\psi_g(\cdot;f,g)$ are differentiable functions with integrable derivatives.
  Condition~\ref{supp: cond: C1} follows  from  \eqref{intheorem: convg of empirical process : F}, \eqref{intheorem: convg of empirical process : G}, and \eqref{weak convg: Grenander}.
  
  It remains to verify only Condition \ref{supp: Cond C2}, which we will do only for $\fnx$,  because the calculations for $\fny$ will be identical. Observe that
   \[\|\fnxm-\fx\|_2^2\lesssim \norm{\lb\sqrt{\fnxm}+\sqrt{\vphantom{\fnxm}\fx}\rb^2}_{\infty}\hd^2(\fnxm,\fx)\lesssim \lb \|\fnxm\|_{\infty}+\|\fx\|_{\infty}\rb \hd^2(\fnxm,\fx).\]
  Using Theorem $7.12$ of \cite{vandergeer} one can show that $\hd^2(\fnxm,\fx)=O_p(n^{-2/3})$, and we have already established that $\|\fnxm\|_{\infty}=O_p(1)$. Thus Condition $\ref{supp: Cond C2}$ also follows.
   Now the proof will follow if we can show that 
   \begin{align}\label{incorollary: hellinger: sigma}
       \frac{1}{\lambda}\edint \psi_f(x;f,g)^2f(x)dx+\frac{1}{1-\lambda}\edint \psi_g(y;f,g)^2g(x)dx=\frac{2\hd^2(f,g)-\hd^4(f,g)}{4\lambda(1-\lambda)},
   \end{align}
 where $\psi_f$ and $\psi_g$ are as in \eqref{def: influence functions}.
 To that end, note that
 \[\hd^2(f,g)=\frac{1}{2}\edint(\sqrt{f(x)}-\sqrt{g(x)})^2dx=1-\underbrace{\edint\sqrt{f(x)g(x)}dx}_{\rho(f,g)}.\]
We calculate 
\begin{align*}
 4\edint \psi_f(x;f,g)^2f(x)dx
=&\  \edint \lb 1-\frac{\sqrt{g(x)}}{\sqrt{f(x)}}-\hd^2(f,g)\rb^2 f(x)dx\\
=&\ \edint \slb \sqrt{f(x)}-\sqrt{g(x)}-\sqrt{f(x)}\hd^2(f,g)\srb^2dx,
\end{align*}
which equals
\begin{align*}
\MoveEqLeft \edint \slb \sqrt{f(x)}-\sqrt{g(x)}\srb^2dx +\hd^4(f,g)-2\hd^2(f,g)\edint \slb \sqrt{f(x)}-\sqrt{g(x)}\srb\sqrt{f(x)}dx\\
=&\ 2\hd^2(f,g)+\hd^4(f,g)-2\hd^2(f,g)(1-\rho(f,g))\\
=&\ 2\hd^2(f,g)+\hd^4(f,g)-2\hd^4(f,g)
\end{align*}
because $1-\rho(f,g)=\hd^2(f,g)$. Therefore,
\[\edint \psi_f(x;f,g)^2f(x)dx=\frac{2\hd^2(f,g)-\hd^4(f,g)}{4}.\]
By symmetry, \eqref{incorollary: hellinger: sigma} follows, thus completing the proof.
    \end{proof}
  Our next lemma establishes  that $\hd^2(\fnx,\fny)$ and $\hd^2(\fnxm,\fnym)$ differ by an $o_p(N^{-1/2})$ term. Observe that Corollary \ref{thm: Hellinger: grenander} combined with Lemma \ref{lemma: equivalence of Hellinger} implies Theorem \ref{thm: Hellinger distance}, and thus establishes the asymptotic distribution of $\hd^2(\fnx,\fny)$ as well.
  
    {\color{black}
 \begin{lemma}\label{lemma: equivalence of Hellinger}
Under the set up of Corollary~\ref{thm: Hellinger: grenander},
   \[\bl \hd^2(\fnx,\fny)-\hd^2(\fnxm,\fnym)\bl =o_p(N^{-1/2}).\]
 \end{lemma}
 
 \begin{proof}
 
Adding and subtracting terms shows that
\begin{align*}
    \hd^2&(\widehat{f}_m,\widehat{g}_n) -\hd^2(\widehat{f}_m^0,\widehat{g}_n^0) \\
    &= [\hd^2(f,\widehat{g}_n) - \hd^2(f,\widehat{g}_n^0)] + [\hd^2(\widehat{f}_m,g) - \hd^2(\widehat{f}_m^0,g)] \\
    &\quad+ \left[\hd^2(\widehat{f}_m,\widehat{g}_n) - \hd^2(\widehat{f}_m,g) - \hd^2(f,\widehat{g}_n) + \hd^2(f,g)\right] \\
    &\quad- [\hd^2(\widehat{f}_m^0,\widehat{g}_n^0) - \hd^2(\widehat{f}_m^0,g) - \hd^2(f,\widehat{g}_n^0) + \hd^2(f,g)].   
\end{align*}
Now observe that
\begin{align}\label{inlemma: equiv: dtv}
    \hd^2(\widehat{f}^0_m,\widehat{g}^0_n) - \hd^2(\widehat{f}^0_m,g) - \hd^2(f,\widehat{g}_n) + \hd^2(f,g) &= \rint [\sqrt{\widehat{f}^0_m} - \sqrt{f}][\sqrt{\widehat{g}^0_n}-\sqrt{g}]\nn \\
    &\le 2 H(\widehat{f}^0_m,f) H(\widehat{g}^0_n,g),
\end{align}
and
\begin{align}\label{inlemma: equiv: dtv: 2}
    \hd^2(\widehat{f}_m,\widehat{g}_n) - \hd^2(\widehat{f}_m,g) - \hd^2(f,\widehat{g}_n) + \hd^2(f,g) &= \rint [\sqrt{\widehat{f}_m} - \sqrt{f}][\sqrt{\widehat{g}_n}-\sqrt{g}] \nn\\
    &\le 2 H(\widehat{f}_m,f) H(\widehat{g}_n,g).
\end{align}
Since squared Hellinger distance is smaller than the $L_1$ distance, using Lemma~\ref{lemma: weak convergence of dist of unimodal density functions} and the fact that  $m/N\to\lambda$, we obtain
\[\hd^2(\fnxm,f)\leq \| \fnxm-f\|_1=o_p(m^{-1/2})=o_p(N^{-1/2}).\]
A similar result holds for $\hd^2(\fnym,\fy)$ as well.
Thus it follows that the right hand side of \eqref{inlemma: equiv: dtv} is $o_p(N^{-1/2})$.
 Lemma~\ref{lemma: hellinger grenander and birge equivalence} in Section~\ref{sec: addlemma for measure}  implies $H(\fnxm, \fnx)$ and $H(\fnxm, \fnx)$ are $o_p(N^{-1/4})$. Therefore, using triangle inequality, we can show that $H(\fnx,\fx)$ and $H(\fny,\fy)$ are $o_p(N^{-1/4})$ as well, which establishes that the right hand side of \eqref{inlemma: equiv: dtv: 2} is also $o_p(N^{-1/2})$.

Therefore, we have shown that
\begin{align*}
\MoveEqLeft    \hd^2(\widehat{f}_m,\widehat{g}_n) -\hd^2(\widehat{f}_m^0,\widehat{g}_n^0)\\
     &= [\hd^2(f,\widehat{g}_n) - \hd^2(f,\widehat{g}_n^0)] + [\hd^2(\widehat{f}_m,g) - \hd^2(\widehat{f}_m^0,g)] + o_p(N^{-1/2}).
\end{align*}
Hence, the proof will be complete if we can show that
\begin{align*}
    \hd^2(f,\widehat{g}_n)  - \hd^2(f,\fnym)&= o_p(N^{-1/2}), \label{eq:gVaries}\\
    \hd^2(\widehat{f}_m,g)  - \hd^2(\fnxm,g)&= o_p(N^{-1/2}). \nonumber
\end{align*}
We will only prove the first line of the above display because the argument for the second line is similar. 

First, we denote $\mathcal{X}_n=\text{supp}(\fnym)\cup \text{supp}(\fny)$
 and
 $\mathcal Y_n=\mathcal X_n\setminus \text{supp}(\fnym)=\text{supp}(\fny)\setminus \text{supp}(\fnym)$. Also denote $b'=b/2$ and $B'=B+b/2$.
 Since $\sqrt {\fny}+\sqrt {\fnym}>0$ on $\mathcal X_n$, it follows that
 \begin{align*}
     |\hd^2(f,\fnym)-\hd^2(f,\fny)|=&\ \left|\dint_{\mathcal{X}_n}\sqrt{f(x)}(\sqrt{\fny(x)}-\sqrt{\fnym(x)})dx\right |\\ =&\ \bl\dint_{\mathcal{X}_n}\frac{\sqrt{f(x)}}{(\sqrt{\fnym(x)}+\sqrt{\fny(x)})}(\fnym(x)-\fny(x))dx\bl\\
     \leq &\ \sup_{x\in\text{supp}(\fnym)}\frac{\sqrt{f(x)}}{\sqrt{\fnym(x)}}\|\fnym-\fny\|_1+\sup_{x\in\mathcal{Y}_n}\frac{\sqrt{f(x)}}{\sqrt{\fny(x)}}\|\fnym-\fny\|_1
 \end{align*}
 where we used the fact that $\mathcal X_n=\text{supp}(\fnym)\cup \mathcal Y_n$.
 Since (a) $f\in\mathcal P(b,B)$, (b) $\fnym(x)>b/2$ for $x\in\supp(\fnym)$ with probability one by \eqref{inlemma: lower bound of grenander}, and (c) $\|\fnym-\fny\|_1$ is {$o_p(N^{-1/2})$} by Lemma~\ref{lemma: weak convergence of dist of unimodal density functions}(A), we have
 \[ |\hd^2(f,\fnym)-\hd^2(f,\fny)|=o_p(N^{-1/2})O_p\lb\sup\limits_{x\in \mathcal Y_n}\frac{1}{\sqrt{\fny(x)}}\rb.\]
Hence, it only remains to show that $\sup\limits_{x\in \mathcal Y_n}\fny(x)^{-1/2}=O_p(1)$.
 
   Let us denote the mode of $\fny$ by $\Mn$.
   First we show that it suffices to only consider the case when $\Mn\notin[Y_{(1)},Y_{(n)}]$. To that end, we will mainly use the following property of $\Fmy$ that follows from (2.7) of \cite{birge1997}:
   \begin{equation}\label{inlemma; measure: birge  2.7)}
     \Fny(x)\leq   \Fmy(x)\quad \text{for}\quad x\in(-\infty,\widehat M_n]\quad\text{and}\quad \Fmy(x)\leq \Fny(x)\quad\text{for}\quad x\in(\widehat M_n,\infty).
   \end{equation}
    Suppose $\Mn\in[Y_{(1)},Y_{(n)}]$. Then by \eqref{inlemma; measure: birge 2.7)},  any $x<Y_{(1)}$ and $y>Y_{(n)}$ satisfy
   $\Fny(x)=0$ and  $\Fny(y)=1$, respectively.
   Therefore,  the support of $\fny$ is contained in  $[Y_{(1)},Y_{(n)}]$. 
   However,  $[Y_{(1)},Y_{(n)}]\subset \text{supp}(\fnym)$, 
   which implies $\mathcal Y_n=\emptyset$ in this case. Therefore, we only consider the case when $\Mn\notin [Y_{(1)},Y_{(n)}]$.
   
   
    First consider the case when $\Mn>Y_{(n)}$. By \eqref{inlemma; measure: birge 2.7)}, any $y\geq \Mn$ satisfies $\Fny(y)\geq \Fmy(y)=1$, and any $x<Y_{(1)}$ satisfies $\Fny(x)\leq \Fmy(x)=0$. Therefore,  we have $\text{supp}(\fny)\subset [Y_{(1)},\Mn]$. Hence, $\Fny(Y_{(n)})=\int_{Y_{(1)}}^{Y_{(n)}}\fny(x)dx$.
    Since $\fny$ is non-decreasing on $[Y_{(1)},Y_{(n)}]$, we have
    \[\Fny(Y_{(n)})\leq (Y_{(n)}-Y_{(1)})\fny(Y_{(n)}).\]
   Note that Because $\fy(y)>b$, this density has a bounded support, implying $(Y_{(n)}-Y_{(1)})<\Rn(Y)$ for some $\Rn(Y)>0$.
  On the other hand,  
  \[|\Fny(Y_{(n)})-1|\leq \|\Fny-\Fmy\|_{\infty}=o_p(N^{-1/2})\]
  by Lemma~\ref{lemma: weak convergence of dist of unimodal density functions}(B).
  Therefore,
  \[\fny(Y_{(n)})\geq\frac{1-o_p(n^{-1/2})}{\Rn(Y)}.\]
  Because $\text{supp}(\fny)\subset[Y_{(1)},\Mn]$ and $\text{supp}(\fnym)\supset [Y_{(1)},Y_{(n)}]$, we have $\mathcal Y_n\subset [Y_{(n)},\Mn]$, indicating 
\[\sup_{x\in\mathcal{Y}_n}\fny(x)^{-1/2}\leq \fny(Y_{(n)})^{-1/2}\le \frac{\sqrt{\Rn(Y)}}{1+o_p(n^{-1/2})}=O_p(1).\]

 Now suppose $\Mn<Y_{(1)}$. Then using  \eqref{inlemma; measure: birge 2.7)}, we deduce that $\Fny(Y_{(n)})\geq \Fmy(Y_{(n)})=1$ and  $\Fny(x)\leq \Fmy(x)=0$ for any $x<Y_{(1)}$. Therefore, $\text{supp}(\fny)\subset[\Mn,Y_{(n)}]$ and $\mathcal Y_n\subset [\Mn,Y_{(1)}]$. 
 Because $\fny$ is non-increasing on $\mathcal Y_n$ in this case, $\fny(x)>\fny(Y_{(1)})$ for any $x\in \mathcal Y_n$. Therefore, if  we can show that $\fny(Y_{(1)})$ is bounded away from $0$, the rest of the proof will follow similar to the case of $\Mn>Y_{(n)}$.
 
 Because $\fny$ is non-increasing on its support,
   \[\fny(Y_{(1)})(Y_{(n)}-Y_{(1)})\geq \Fny(Y_{(n)})-\Fny(Y_{(1)})\stackrel{(a)}{=} 1-o_p(n^{-1/2}),\]
 where (a) follows from Lemma~\ref{lemma: weak convergence of dist of unimodal density functions}(B). Since $Y_{(n)}-Y_{(1)}=\Rn(Y)<\infty$, we have $\fny(Y_{(1)})^{-1}=O_p(1)$, which completes the proof.
\end{proof}
 }
  \subsubsection*{Proof of Theorem~\ref{thm: Hellinger distance}}
 Theorem \ref{theorem: convergence: Ti: log-concave } follows from Corollary~\ref{thm: Hellinger: grenander} and Lemma \ref{lemma: equivalence of Hellinger}.
  \hfill $\Box$
 \subsubsection*{Proof of Lemma~\ref{lemma: measure of discrep: log-concave}}
 Theorem 4 of \citeauthor{theory} implies that  $\flx$  uniformly converges  to $\fx$ almost surely  provided (i) $\fx$  has finite first moment, (ii) the support of $\fx$ has nonempty interior, and (iii)
 $ \int \max\{\log f(x),0\}f(x)dx<\infty$.
For log-concave $\fx$, (i) follows from Lemma 1 of \citeauthor{theory}, (ii) follows from the continuity of $\fx$, and (iii) follows because $\fx$ is bounded \citep[cf. Lemma 5,][]{theory}. The similar result holds for $\fly$ as well. Because uniform convergence implies pointwise convergence, the above implies $\flx\fly$ pointwise converges to $\fx\fy$ almost surely. Therefore, an application of Scheff\'{e}'s Theorem \citep[cf. Theorem 16.12,][]{billingsley2013} yields that as $m,n\to\infty$, 
\[\rint \lb \flx(x)\fly(x)\rb ^{1/2}dx \as \rint \lb \fx(x)\fy(x)\rb ^{1/2}dx, \]
which indicates
 \begin{align*}
\hd^2(\flx,\fly) =& 1-\rint \lb \flx(x)\fly(x)\rb ^{1/2}dx 
\as\ =\hd^2(\fx,\fy).
  \end{align*}

  For the smoothed log-concave MLE, the uniform convergence of $\flx$  to $\fx$ follows from 
  Theorem 1 of \citeauthor{smoothed} provided $\fx$ has  finite second moment, which follows trivially because all moments of a log-concave density are finite \citep[Lemma 1,][]{theory}. The rest of the proof then follows from Scheff\'{e}'s Theorem as in the case of the log-concave MLE. \hfill $\Box$
%

 \subsection{Additional Lemma:}
 \label{sec: addlemma for measure}
 \begin{lemma}\label{lemma: hellinger grenander and birge equivalence}
 Under the set up of Lemma~\ref{lemma: equivalence of Hellinger},
  \[ \hd^2(\fnx,\fnxm)=o_p(m^{-1/2})\quad\text{ and }\quad \hd^2(\fny,\fnym)=o_p(n^{-1/2}).\]
 \end{lemma}
 
 \begin{proof}
  We will prove the result only for $\hd^2(\fnx,\fnxm)$ because the proof  of the other case will be identical.
  
  Let us denote the mode of $f$ and $\fnx$ by $M$ and $\hM$, respectively. First we consider the case when $\hM>M$.
The proof of Lemma 1 of \citeauthor{birge1997} entails that, in this case, there exist $a,\beta,c\in\RR$ such that 
 \begin{enumerate}
 \item $a\leq M$, $\beta\geq \hM$, and $c\in[M,\hM]$.
 \item $\fnxm(x)=\fnx(x)$  for $x<a$ and $x>\beta$. Therefore, $\Fnxm(x)=\Fnx(x)$  for $x\leq a$ and $x\geq \beta$.
 \item $\fnxm(x)\geq\fnx(x)$ for $x\in(a,c)$, and $\fnxm(x)\leq\fnx(x)$ for $x\in(c,\beta)$.
 \end{enumerate}
 Using the above relations, and denoting the distribution functions of $\fnx$ and $\fnxm$ by $\Fnx$ and $\Fnxm$, respectively, we deduce that
 \begin{align*}
\MoveEqLeft \rint \lb\sqrt{\fnx(x)}-\sqrt{\fnxm(x)}\rb^2dx\\
=&\ \dint_{a}^{\beta} \lb\sqrt{\fnx(x)}-\sqrt{\fnxm(x)}\rb^2dx\\
=&\ \dint_{a}^{\beta}\fnx(x)dx+\dint_{a}^{\beta}\fnxm(x)dx-2\dint_{a}^{\beta}\sqrt{\fnx(x)\fnxm(x)}dx\\
\leq &\ \Fnx(\beta)-\Fnx(a)+\Fnxm(\beta)-\Fnxm(a)-2\dint_{a}^{c}\fnx(x)dx -2\dint_{c}^{\beta}\fnxm(x)dx\\
=&\ \Fnx(\beta)-\Fnx(a)+\Fnxm(\beta)-\Fnxm(a)-2\slb\Fnx(c) -\Fnx(a)+\Fnxm(\beta)-\Fnxm(c)\srb
 \end{align*}
 which is $ 2\{ \Fnxm(c)-\Fnx(c)\}$
 because $\Fnxm(x)=\Fnx(x)$ for $x=a,\beta$.
 Hence, we observe that
 \[\hd^2(\fnx,\fnxm)\leq 2\|\Fnxm-\Fnx\|_{\infty},\]
 which is less than $\eta=1/N$ by the construction of \birge's estimator (see Section~\ref{sec: density estimation}). Hence the proof of Lemma \ref{lemma: equivalence of Hellinger} follows for this case.
 
  Now suppose that $\hM<M$. Then from the proof of Lemma 1 in \citeauthor{birge1997}, one can prove the existence of  $a,\beta,c\in\RR$ such that 
 \begin{enumerate}
 \item $a\leq \hM$, $\beta\geq M$, and $c\in[\hM,M]$.
 \item $\fnxm(x)=\fnx(x)$  for $x<a$ and $x>\beta$. Therefore, $\Fnxm(x)=\Fnx(x)$  for $x\leq a$ and $x\geq \beta$. 
 \item $\fnx(x)\geq\fnxm(x)$ for $x\in(a,c)$, and $\fnx(x)\leq\fnxm(x)$ for $x\in(c,\beta)$.
 \end{enumerate}
 Then in the same way as in the case of $\hM>M$, we can show that
 $\hd^2(\fnx,\fnxm)\leq \|\Fnxm-\Fnx\|_{\infty}$,
 which completes the proof of the current lemma.
 \end{proof}
 
 \section{Additional simulations}
 \label{app: conservative test simulations}
 \textcolor{black}{
 In this section, we perform simulations on the exact same settings as in Section \ref{sec:simulation:1}, but we use the critical value $C_{m,n}z_\alpha$  for the TSEP tests. We remind the readers that  $C_{m,n}$ was set to be $\max\{\sqrt{N/m},\sqrt{N/n}\}$. Since in this case $m=n$, we have $N=2n$, which implies  $C_{m,n}=\sqrt 2$. The power curves are given by Figure \ref{fig: power: corrected}, which implies that the resulting TSEP tests, which we will refer to as the conservative TSEP tests, have inferior power compared to the minimum t-tests. Moreover, a  comparison between Figure \ref{fig: power} and Figure \ref{fig: power: corrected} indicates that the  power of the conservative TSEP tests is much less compared to that of the ordinary TSEP tests, which use the  critical value $z_\alpha$. However, Figure \ref{fig: power: corrected} implies that in case (d), where one distribution is the heavy-tailed Pareto distribution, the conservative TSEP tests succeed to control the type I error at the LFC configuration. All  other tests, including the nonparametric minimum t-test and the nonparametric ordinary TSEP test with critical value $z_\alpha$, have type I error slightly higher than 0.05  at  the LFC configuration in case (d); see Figure \ref{fig: power}. The nonparametric  tests  control the type I error at all other cases, however.  Also, all TSEP tests control type I error in case (b), where the distributions cross each other at the boundary. 
 To summarize, the conservative TSEP tests might have a slight advantage over the ordinary counterparts in terms of  type I error in some boundary cases, but this advantage comes at the cost of a drastic power-loss.  In view of the above, we do not recommend the conservative TSEP tests for implementation. }

 \begin{figure}[ht]
\includegraphics[width=\textwidth]{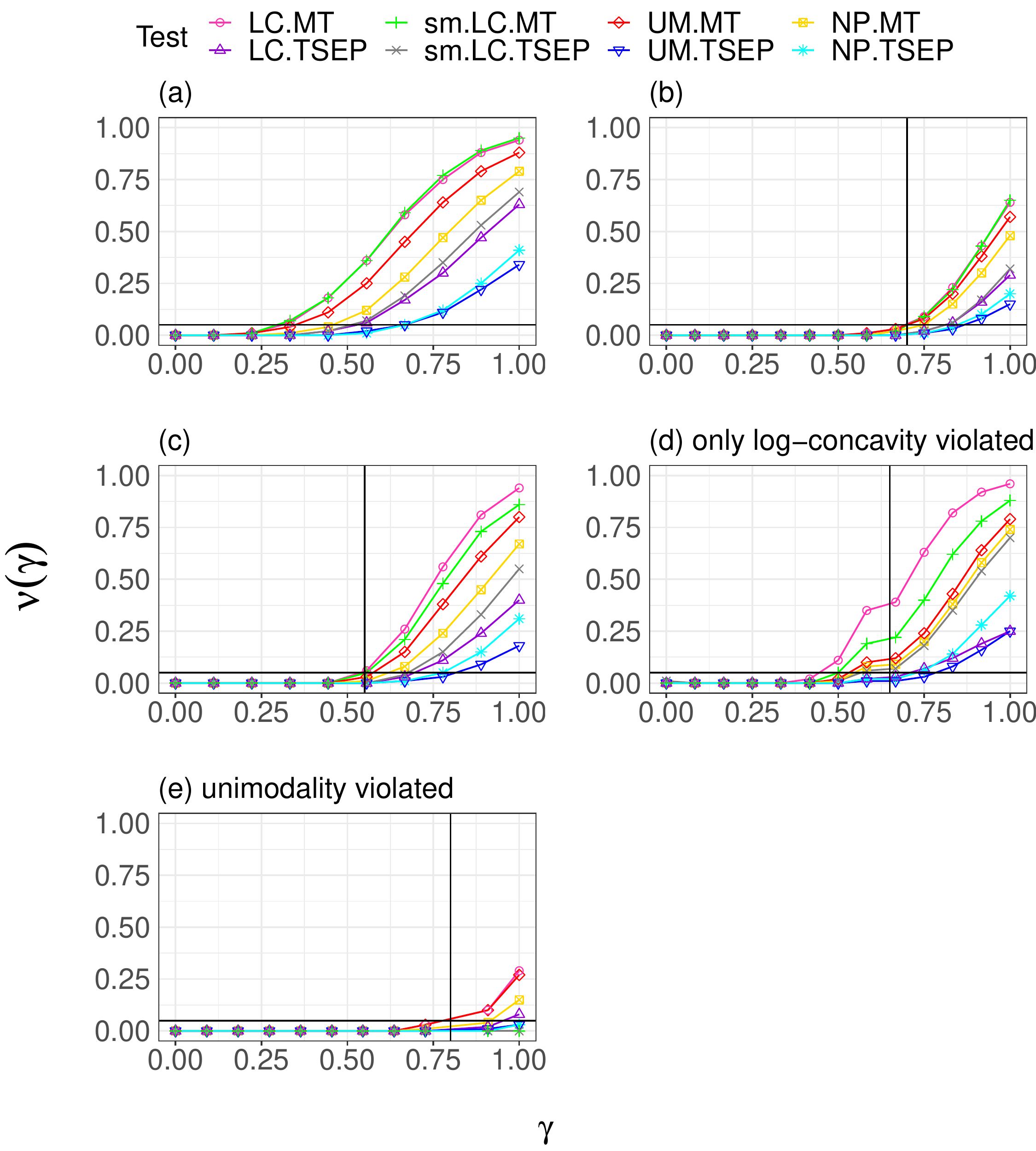}
 \caption{ Plot of estimated power $\nu(\gamma)$ vs $\gamma$  when $m=n=100$ for simulation schemes (a)-(e); here we use the critical value $C_{m,n}z_\alpha$ for thee TSEP tests. Here MT and TSEP correspond to the minimum t-test and the TSEP test, respectively. The standard deviation of the $\nu(\gamma)$ estimate in each case is less than $0.005$.
  The black  horizontal line corresponds to the level of the test, $\alpha=0.05$. For cases (b)-(e), the black vertical line   represents the LFC configuration $\gamma^*$, taking value 0.70 (b), 0.55 (c), 0.65 (d), and 0.80 (e). }
 \label{fig: power: corrected}
\end{figure}

 \section{Additional tables and Figures}
\label{app: additional tables and figures}

 \begin{table*}[h]
  \begin{tabular}{@{}lll@{}}
\toprule
Trials & HVTN 097 & HVTN 100 \\ 
\midrule

Phase & 1b & 1/2 \\

Site & 3 towns in South Africa & 6 towns in South Africa \\ 

Study design &  placebo controlled, randomized, & placebo controlled, randomized,\\ 
 &  double-blind &  double-blind\\

Enrollment & 100 & 252 \\ 

Vaccinee : Placebo ratio & 4:1 & 4:1\\ 

 Per protocol vaccinees & 73 & 185\\
 
Positive  respondents & 68 & 180\\

Age-range & 18-40 & 18-40\\ 
 
Enrollment period & June-December 2013 & February-May 2015 \\ 
 \textcolor{black}{\makecell{Clade of  HIV-1 insert\\ strains used in vaccines}} & B and E & C \\ 

Products used & ALVAC and AIDSVAX & ALVAC and gp120\\
\bottomrule
\end{tabular} 
\caption{Summary of the trial HVTN 097 and trial HVTN 100.  By positive respondents, we refer to vaccinees who developed immune response  for at least one of the seven clade C V1V2 antigens under consideration. }
\label{table: HVTN 097 and HVTN 100}
  \end{table*}

\begin{figure}[h]
  \begin{subfigure}{\textwidth}
\includegraphics[width=\textwidth]{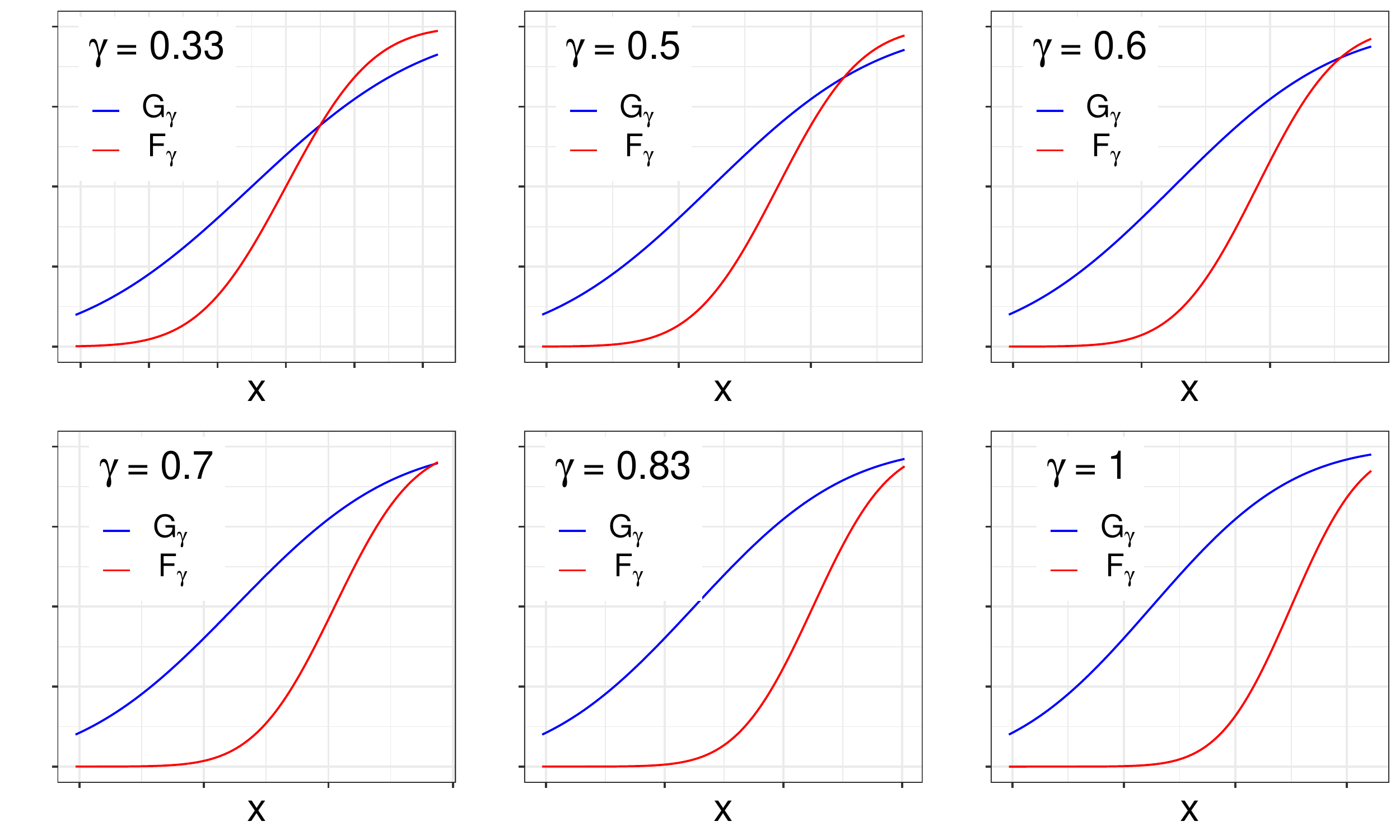}
\caption{Plots for case (b).}
\label{Figure: cross dist: d}
\end{subfigure}\\
\begin{subfigure}{\textwidth}
\includegraphics[width=\textwidth]{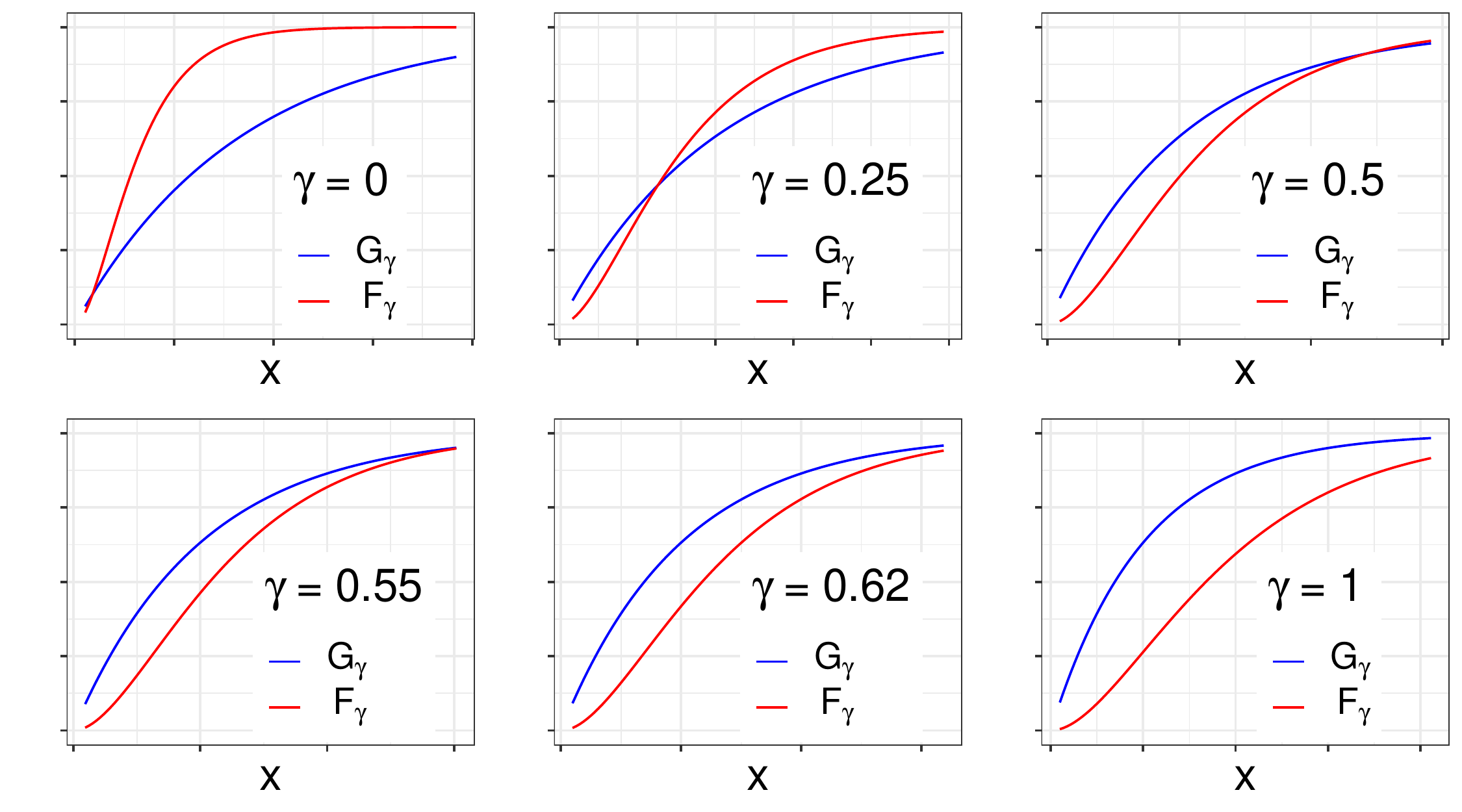}
\caption{Plots for case (c).}
\label{Figure: cross dist: e}
\end{subfigure}
\caption{Plots of the distribution functions $\Fx_{\gamma}$ and $\Fy_{\gamma}$ for several values of $\gamma$ in cases  (b) and (c). The distribution functions are shown  on the region $D_p(F_{\gamma},G_{\gamma})$.
}
\end{figure}

  \begin{figure}[h]
  \begin{subfigure}{\textwidth}
\includegraphics[width=\textwidth]{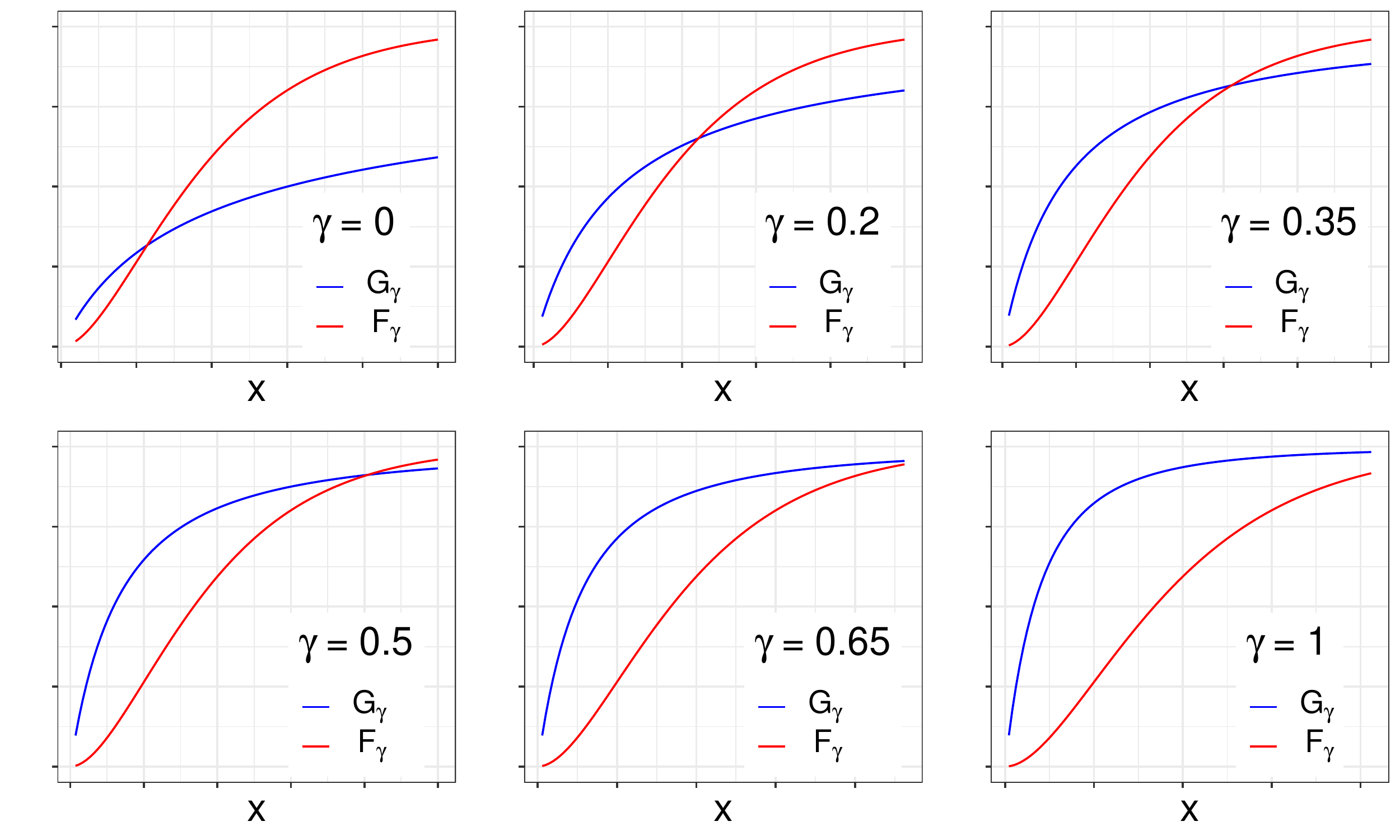}
\caption{Plots for case (d).}
\label{Figure: cross dist: a}
\end{subfigure}\\
\begin{subfigure}{\textwidth}
\includegraphics[width=\textwidth]{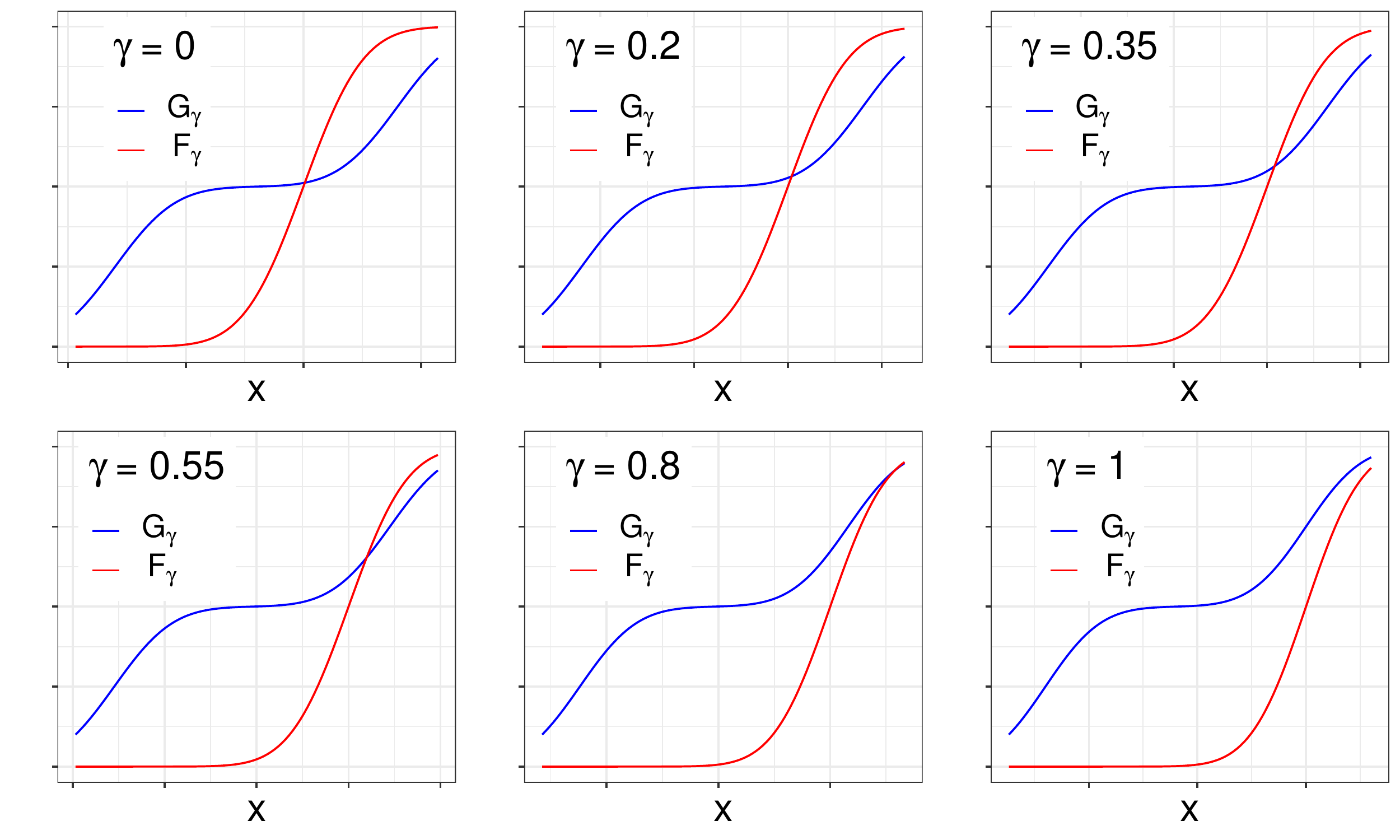}
\caption{Plots for case (e).}
\label{Figure: cross dist: b}
\end{subfigure}
\caption{Plots of the distribution functions $\Fx_{\gamma}$ and $\Fy_{\gamma}$ for several values of $\gamma$ in cases  (d) and (e). The distribution functions are shown  on the region $D_p(F_{\gamma},G_{\gamma})$.
}
\end{figure}

 


\begin{figure}
    \centering
    \includegraphics[width=\textwidth]{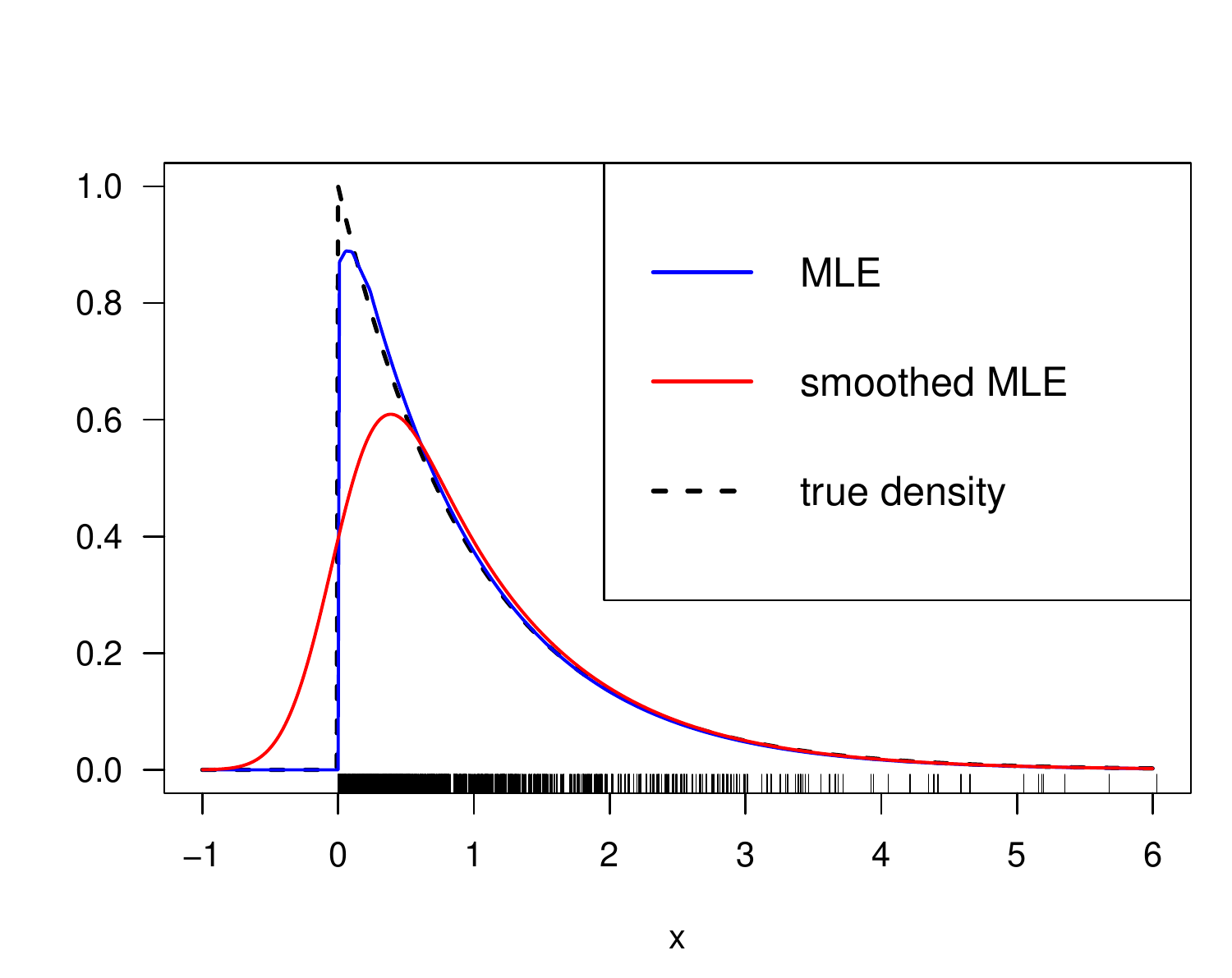}
    \caption{Plot of the log-concave density estimators based on a sample of size 1000 from standard exponential distribution. Observe that the smoothed log-concave MLE does not approximate the true density well near the origin, which is also the point of discontinuity. }
    \label{fig:smooth exp}
\end{figure}
\end{document}